\documentclass[a4paper,10pt]{article}
\usepackage{soul}
\usepackage[usenames,dvipsnames]{color}
\usepackage{jheppub}
\usepackage{esvect}
\usepackage{comment}
\usepackage{amsmath, amssymb, slashed, epsf, color, graphicx, latexsym, bm, amsfonts}
\usepackage{mathrsfs}
\usepackage{tensor}
\usepackage{physics}
\usepackage{pdflscape}
\usepackage{tikz}
\usetikzlibrary{positioning, calc}
\usepackage{dsfont}
\usepackage{epsfig}
\usepackage{graphics}
\usepackage[T1]{fontenc}
\usepackage{mathtools}
\usepackage{fixme}
\fxsetup{status=draft}
\usepackage{float}
\usepackage{bbm}
\usepackage{subcaption}
\usepackage{enumerate}
\usepackage{verbatim}
\usepackage{blindtext}
\usepackage{environ}

\NewEnviron{eqsp}{%
  \begin{equation}
    \begin{split}
      \BODY
    \end{split}
  \end{equation}
}
\usepackage{tikz}
  \usepackage{amsmath}
\newlength\squareheight
\makeatletter
\renewcommand*\env@matrix[1][*\c@MaxMatrixCols c]{%
  \hskip -\arraycolsep
  \let\@ifnextchar\new@ifnextchar
  \array{#1}}
\makeatother
\DeclareFontEncoding{LS1}{}{}
\DeclareFontSubstitution{LS1}{stix}{m}{n}
\DeclareSymbolFont{symbols2}{LS1}{stixfrak}{m}{n}
\DeclareMathSymbol{\typecolon}{\mathbin}{symbols2}{"25}
\usepackage[mathlines]{lineno}
\newcommand{\C}{\mathds{C}}

\newcommand{\Z}{\mathds{Z}}

\newcommand{\ov}[1]{\overline{#1}}
\newcommand{\mr}[1]{\mathsf{#1}}
\newcommand{\tsf}[1]{\textsf{#1}}
\newcommand{\mf}[1]{\mathfrak{#1}}

\def\PE{\mathsf{PE}}

\def\CB{{\cal B}}

\def\CI {{\cal I}}

\def\CN {{\cal N}}

\def\CV {{\cal V}}

\def\CI {{{\cal I}}}
\def\CB {{\cal B}}

\def\CS {{\cal S}}

\def\IF{\mathds{F}}

\def\IZ{{\mathds{Z}}}

\newcommand{\dplus}{\dot{+}}

\begin{document}
\title{Argyres-Douglas Theories, Macdonald Indices And the Arc Space Of Zhu's $C_2$-Algebra}
\author[a]{George Andrews,}
\author[b]{Anindya Banerjee,}\author[c]{Chinmaya Bhargava,}\author[d]{Ranveer Kumar Singh,} \author[d]{Runkai Tao}
\affiliation[a]{Department of Mathematics, Penn State University, University Park, PA 16802}
\emailAdd{gea1@psu.edu}
\affiliation[b]{Department of Physics, University of Cincinnati, Cincinnati, OH 45211}
\emailAdd{ab1702@scarletmail.rutgers.edu}
\affiliation[c]{CTP and Department of Physics and Astronomy, Queen Mary University of London, London E1 4NS, UK}
\emailAdd{c.bhargava@qmul.ac.uk}
\affiliation[d]{NHETC and Department of Physics and Astronomy, Rutgers University, 126
Frelinghuysen Rd., Piscataway NJ 08855, USA}
\emailAdd{rks158@scarletmail.rutgers.edu}
\emailAdd{runkai.tao@physics.rutgers.edu}
\abstract{In this paper, we relate the Macdonald index of a 4d $\CN=2$ \tsf{SCFT} with the Hilbert series of the arc space of the Zhu's $C_2$-algebra of the corresponding Schur \tsf{VOA}. Using this, we conjecture a simple formula for the Macdonald index of $(A_1,D_{2n+1})$ Argyres-Douglas theory. We perform checks of the formula against the known Schur limits and \tsf{RG} flows. To match the Schur limit, we prove new $q$-series identities.} 
\maketitle
\section{Introduction}
We study some families of 4d $\mathcal{N}=2$ superconformal field theories known as Argyres-Douglas theories \cite{Argyres:1995jj, Argyres:1995xn, Witten:1997sc, Cecotti:2010fi, Xie:2012hs}. These theories are characterized by the lack of a description in terms of any Lagrangian preserving the $\mathcal{N}=2$ supersymmetry, and fractional scaling dimensions of Coulomb branch generators (a subset of well-understood local operators that are expected to be universal in any interacting 4d $\mathcal{N}=2$ theory). Despite the absence of a Lagrangian description, these theories are amenable to powerful methods of analysis that give us insight into the dynamics of the theory without recourse to any perturbative approximation. This is made possible by several special features of these theories. These include: 
\begin{enumerate}[(1)]
\item standard simplifying features of conformal symmetry like radial quantization, the state-operator correspondence and the finite radius of convergence of the operator product expansion \cite{Rattazzi:2008pe} 
\item applicability of the representation theory of the superconformal algebra \cite{Dolan:2002zh, Cordova:2016emh} in constraining and characterizing the local operator spectrum of the theory, \item the presence of a continuous moduli space of vacua in the quantum theory that is not lifted by any radiative corrections \cite{Argyres:1996eh}, and 
\item the existence of subsectors of local operators (BPS operators) that are protected against (or transform in a tractable manner under) large classes of deformations of the theory. Such deformations include certain supersymmetry-preserving renormalization group flows (henceforth, \tsf{RG} flows) associated with spontaneous breaking of the global symmetry of these theories, and these in particular will be of interest to us \cite{Maruyoshi:2016aim, Maruyoshi:2016tqk}. 
\end{enumerate}
In this context, we will make extensive use of a familiar counting function called the superconformal index \cite{Kinney:2005ej} that counts (albeit with ambiguities) the \tsf{BPS} operators of a 4d \tsf{SCFT} and has good behavior under supersymmetry-preserving \tsf{RG} flows.
\paragraph{}
It was demonstrated in \cite{Beem:2013sza} that a subset of the ring of \tsf{BPS} local operators (called the Schur ring) has close connections to a vertex operator algebra, the algebraic building block of a two-dimensional rational conformal field theory. Throughout this work, we will focus our analysis on this subring. The superconformal index on this subring also restricts to two special limits called the Macdonald index and the Schur index \cite{Gadde:2011uv, Buican:2015tda, Cordova:2015nma,Song:2016yfd,Watanabe:2019ssf} (see Appendix \ref{app:conventions} for a brief review) that we will use to understand properties of this ring. 
\paragraph{}
The goal of the present paper is twofold:
\begin{enumerate}
    \item Provide a concrete algorithm to compute the Macdonald index of a 4d $\CN=2$ \tsf{SCFT} using the Schur \tsf{VOA}. To do this, we prove, in Section \ref{sec:Mac_Hil_ser_arc}, that under certain technical conditions, the Hilbert series of the arc space of Zhu's $C_2$-algebra for a strongly finitely generated \tsf{VOA} is equal to the refined character of the \tsf{VOA} defined by Song \cite{Song:2016yfd}. Then, Song's conjecture implies that the Macdonald index (with flavor fugacities set to 1) of a 4d $\CN=2$ \tsf{SCFT} is the Hilbert series of the arc space of Zhu's $C_2$-algebra (viewed as a variety) of the corresponding Schur \tsf{VOA}. We check this result for $(A_1,A_{2n})$ Argyres-Douglas theory for which the Macdonald index is known to take a simple form. One of the main technical obstacles in the evaluation of the Hilbert series of arc space of a variety is determining the dimension of varieties appearing in the arc space construction. We use the Gr\"{o}bner basis algorithm \cite{loustaunau1994introduction} for this computation, and prove a formula for the Hilbert series in terms of monomials in the Gr\"{o}bner basis, see \eqref{eq: hilbert_ser_Rinfty_gen} for the explicit expression.
    \item We propose a simple formula for the Macdonald index of $(A_1,D_{2n+1})$ Argyres-Douglas theory (see Section \ref{sec:AD_intro} for a brief account of Argyres-Douglas theories). The main observation that leads us to expect a simple formula is that: (1) the Higgs branch \tsf{RG} flow of the $(A_1,D_{2n+1})$ theory relates it to the $(A_1,A_{2n-2})$ Argyres-Douglas SCFT (see Section \ref{sec:RG_flow} for more details); (2) While this \tsf{IR SCFT} has no flavor symmetry, we find that the flavor symmetry of the \tsf{UV} theory can be accounted for in the \tsf{IR} by tensoring with a $\mathbb{Z}_2$-projected free hypermultiplet. \\ 
    We will see that the Macdonald index can be written as a modification of the Macdonald index of the free hypermultiplet\footnote{There is a subtle issue involving a $\Z_2$ projection of the operator content of free hypermultiplet that is required, see Remark \ref{rem:proj_free_hyper}.} tensored with the $(A_1,A_{2n-2})$ theory. Following this observation and using some numerical evidence, we propose a formula for Macdonald index of the $(A_1,D_3)$ theory in Section \ref{sec:A1D3_Mac_conj} and generalize our proposal to $(A_1,D_{2n+1})$ theory in Section \ref{sec:A1D2n+1_Mac_conj}.  We then perform several nontrivial checks of our formula. In particular, to match the Schur limit to known Schur indices of the theory, we proved a family of interesting $q$-series identities. Checking our formula against the Hilbert series of the arc space of Zhu's $C_2$-algebra of the Schur \tsf{VOA} of these theories takes us into a technical discussion of Gr\"{o}bner bases and constraint counting. We also check that our formula is consistent with the \tsf{RG} flow to the $(A_1,A_{2n-2})$ theory (but we need to tensor with a free hypermultiplet for our purpose of computing the Macdonald index of one theory from the other). Concretely, we show in Section \ref{sec:prod_A2n_fH_Mac} that the product of the Macdonald indices of $(A_1,A_{2n-2})$ theory and the free hypermultiplet can be used to construct our formula for the Macdonald index of the $(A_1,D_{2n+1})$ theory by inserting a simple correction factor.  
    \end{enumerate} 
In the following, we summarize the results of the paper for readers who are only interested in the final results. Conjectures (or results which assume another conjecture) are indicated with a question mark in parenthesis (?).
\begin{enumerate}
    \item (?) The Macdonald index of a 4d $\CN=2$ \tsf{SCFT} is the Hilbert series for the arc space of the Zhu's $C_2$-algebra of the corresponding Schur \tsf{VOA} if it satisfies the technical conditions of Theorem \ref{thm:Hil_ser_mac}. 
    \item (?) The Macdonald index for $(A_1,D_{2n+1})$ theory is given by (see Appendix \ref{app:special_func} for the definition of the special functions appearing in the formulas below) 
    \begin{equation}
    \CI^{(A_1,D_{2n+1})}_{\mr{M}} = 1 + \sum_{m=1}^\infty 
    t^m \sum_{l=1}^m
    \frac{f^n_{m,l}(q)}{(q;q)_l}\sum_{k=0}^{2l}\binom{2l}{k}_q~,
\end{equation}
where $f^n_{m,l}(q)$ is a class of polynomial functions in $q$ which satisfy the hypergeometric relation 
\begin{equation}
    1 + \sum_{m=1}^\infty q^m \sum_{l=1}^m 
    \frac{f^n_{m,l}(q)}{(q;q)_l} (q^N; q)_l 
    (q^{1-N};q)_l = q^{nN(N-1)}~.
\end{equation}
A closed form expression for the functions $f^n_{m,l}$ is given by 
\begin{equation}
f^k_{m,l}(q)=\sum_{i_k,i_{k-1},\dots,i_3\geq 0}\frac{q^{Q(m-(k-1)i_k-(k-2)i_{k-1}-\dots-2i_3-l,i_3,\dots,i_k)}(q;q)_{l}}{(q;q)_{i_k}\dots (q;q)_{i_3}(q;q)_{m-(k-1)i_k-(k-2)i_{k-1}-\dots-2i_3-l}(q;q)_{2l+i_3+2i_4+\dots+(k-2)i_k-m}}~,    
\end{equation}
where 
\begin{equation}
    Q(i_2,i_3,\dots,i_k):=i_k^2+(i_k+i_{k-1})^2+\dots+(i_k+\dots+i_3+i_2)^2~.
\end{equation}
In particular
\begin{equation}
    \CI_{\mr{M}}^{(A_1,D_3)}=\sum_{m=0}^\infty 
    \frac{t^m}{(q;q)_m} \sum_{k=0}^{2m}\binom{2m}{k}_q~.
\end{equation}
\item To match the Schur limit of our formula with the known Schur indices, we proved the following identity:
\begin{equation}
\left(\frac{\left(q^{2n+1} ; q^{2n+1}\right)_{\infty}}{(q ; q)_{\infty}}\right)^3= 1 + \sum_{m=1}^\infty 
    q^m \sum_{l=1}^m
    \frac{f^n_{m,l}(q)}{(q;q)_l}\sum_{k=0}^{2l}\binom{2l}{k}_q~.    
\end{equation}
We also proved the identity
\begin{equation}
\left(\frac{\left(q^{2} ; q^{2}\right)_{\infty}}{(q ; q)_{\infty}}\right)^3= \frac{1}{(q;q)_\infty(q^2;q^2)_\infty}\sum_{m=0}^\infty(-1)^mq^{m^2+m}\left(\frac{1+q^{2m+1}}{1-q^{2m+1}}\right)~,    
\end{equation}
of potential mathematical interest in its own right. 
\end{enumerate}
\paragraph{} 
While our motivations in this work stem from physics, there is a significant overlap of the questions of interest to us with those of the community of mathematicians working on algebro-geometric and representation theoretic aspects of vertex operator algebras. We therefore give a rapid (and unfortunately insufficient) account of some connections. 

Because of the difficulty in rigorously defining quantum field theories, a persistent question for physicists is whether and to what extent a quantum field theory can be uniquely specified by more systematically understood mathematical structures. To this end, we try to uniquely pin down quantum field theories in terms of simpler objects like moduli spaces of vacua, stringy geometries, chiral rings of local operators, BPS spectra and conformal bootstrap data. The Zhu's $C_2$-algebra (a commutative associative Poisson algebra) \cite{Zhu1995ModularIO}, the associated variety of a \tsf{VOA} and its arc space \cite{arakawa2019singular,arakawa2021arc} play a somewhat analogous role in the theory of vertex operator algebras. We refer the reader to the work of Arakawa and Moreau \cite{ArakawaMoreau2021} on  
the realization of a quasi-lisse \tsf{VOA} as the quantization of the arc space of its associated variety. 

A vertex operator algebra $V$ involves an infinite-dimensional vector space $V=\oplus_{n\in\mathbb{Z}}V_n$ that admits a product which is neither commutative nor associative. One of the steps in understanding the representation theory of a vertex operator algebra is to descend to a Poisson vertex algebra and from there, to a commutative, associative Poisson algebra via appropriate simplifications. On the one hand, this can be accomplished by using Li's canonical decreasing filtration \cite{li2005abelianizing} (which leads to the associated graded algebra $\text{gr}^F(V)$ and the singular support $\text{SS}(V)$), and on the other, by using Zhu's $C_2$-algebra $R_V$ \cite{Zhu1995ModularIO} (which leads to the associated scheme of a \tsf{VOA} $X_V$ and its arc scheme $(X_V)_{\infty}$). An important ingredient in our work involves the bridge between these two approaches. While the same Poisson scheme may be obtained from both approaches,
\begin{equation}
    R_{V}\cong V/C_2(V)\cong F^0(V)/F^1(V), \quad X_V=\text{Spec }R_V
\end{equation}
the same is not true of the infinite-dimensional differential rings $\text{SS}(V)$ and $(R_V)_{\infty}$. Roughly, their exists a map between the singular support of a \tsf{VOA} and the arc space of its associated variety, whose being an isomorphism will be taken as the definition of the \tsf{VOA} being classically free. In the context of 4d $\mathcal{N}=2$ \tsf{SCFT}s, it was shown in \cite{Beem:2013sza} that a subring of quarter-BPS local operators forms a \tsf{VOA}. This \tsf{VOA} is further expected to be quasi-lisse \cite{arakawa2018quasi,Beem:2017ooy} and classically free \cite{van2024first,Creutzig:2022sfx}. 

The notion of classical freeness will be important for us, but it is worth pointing out that it is a relatively new notion in mathematics, and we do not have a full understanding of which vertex operator algebras are classically free (see \cite{arakawa2019singular} for a counterexample). The Schur \tsf{VOAs} that we will encounter in this paper have fortunately all been demonstrated to be classically free. For a class of 4d $\mathcal{N}=2$ \tsf{SCFT}s whose Schur \tsf{VOAs} are strongly finitely generated, quasi-lisse and classically free (terms we will define in Section \ref{sec:Mac_Hil_ser_arc}), we will be able to compute the Macdonald index directly from the \tsf{VOA} using the above isomorphism. In more general cases, this process is made ambiguous by the so-called ``R-filtration problem'', since the Schur \tsf{VOA} does not always keep track of the $\mr{SU(2)_R}$ symmetry of the associated 4d $\mathcal{N}=2$ \tsf{SCFT}. 

Let us put the results of this paper in context. The superconformal index (and its special limits including the Macdonald and Schur indices) is one of the observables of interest in the study of 4d $\CN=2$ \tsf{SCFT}s \cite{Kinney:2005ej,Gadde:2011ik,Gadde:2011uv}. The computation of the Schur index of a 4d $\CN=2$ \tsf{SCFT} from the associated Schur \tsf{VOA} was one of the main results of the famous paper \cite{Beem:2013sza}. Since then, many new methods have been found for the computation of superconformal index \cite{Maruyoshi:2016aim} and its Schur and Macdonald limits \cite{Buican:2015hsa,Buican:2015ina,Bhargava:2023hsc}. More recently, there has been a renewal of interest in the computation of the Macdonald index using novel structures like BPS charge lattice and the chiral ring \cite{Cordova:2015nma,Song:2016yfd,Kim:2025klh,Kang:2025zub,Watanabe:2019ssf,Foda:2019guo,Jiang:2026ffv,Andrews:2025tko,Wang:2026jad,Xie:2019zlb}. In this work we extend the work of \cite{Bhargava:2023hsc} to a new series of Argyres-Douglas theories, called the $(A_1,D_{2n+1})$ theories. We also prove a theorem that suggests that our method may be extended more generally to other 4d $\CN=2$ \tsf{SCFT}s where the Schur \tsf{VOA} is quasi-lisse and classically free. 

The organization of this paper is as follows: In Section \ref{sec:AD_intro}, we give a brief overview of Argyres-Douglas theories, outlying some of their constructions and properties. We introduce Higgs branch \tsf{RG} flows and their \tsf{VOA} analogue, with the expectation that our results are a consequence of such an \tsf{RG} flow between 4d $\mathcal{N}=2$ \tsf{SCFT}s. In Section \ref{sec:Mac_Hil_ser_arc}, we introduce some of the properties of vertex operator algebras we will need. We also introduce arc spaces, and connect them to the Macdonald index, the observable of interest for us. We conclude this section with an account of computing the Macdonald index for the $(A_1,A_{2n})$ Argyres-Douglas theories using their Schur \tsf{VOAs}. We dedicate Section \ref{sec:A1D3_Mac_conj} to the simplest \tsf{SCFT} to which we apply our methods and obtain new results, the $(A_1,D_3)$ theory. In many ways, this is the simplest 4d $\mathcal{N}=2$ \tsf{SCFT} with nontrivial flavor symmetry. We connect it to the $\mathbb{Z}_2$-projected free hypermultiplet and use this connection to extract an ansatz for the Macdonald index. At the moment, we are able to prove our ansatz only in the simplified Schur limit. In Section \ref{sec:A1D2n+1_Mac_conj}, we generalize our methods of the previous section to compute the Macdonald index for the full family of $(A_1,D_{2n+1})$ theories. For this, we will need to tensor the previously used $\mathbb{Z}_2$-projected free hypermultiplet with the $(A_1,A_{2n-2})$ theory. Note that this tensor product is only carried out at the level of the space of Schur operators that are counted by the Macdonald index, we have no reason to expect an actual factorization of \tsf{SCFT}s. Throughout the paper, we will use numerical methods when working with arc spaces to test our ansatz. We provide some details of these computations in the Appendices, and also include our Mathematica notebooks as supplemental material.
\section{Constructing Argyres-Douglas Theories}\label{sec:AD_intro}
Argyres-Douglas theories were originally discovered as special singular points on the Coulomb branch of 4d $\mathcal{N}=2$ gauge theories where mutually non-local \tsf{BPS} particles became simultaneously massless, leading to an \tsf{IR} \tsf{SCFT} that could not have a Lagrangian description preserving $\mathcal{N}=2$ supersymmetry \cite{Argyres:1995jj,Argyres:1995xn}. A second feature of these theories was that one or more Coulomb branch operator dimensions of these \tsf{SCFT}s were fractional. Later, other constructions for these theories were found. It is often the case that different observables become more tractable within different constructions of these theories. Three such constructions include: 
\begin{enumerate}[(1)]
    \item From D-branes probing isolated hypersurface singularities in Type \tsf{IIB} string theory \cite{Cecotti:2010fi}. 
    \item As generalized Class $\CS$ theories\footnote{See \cite{Witten:1997sc,Gaiotto:2009hg,Gaiotto:2009we} for discussion on class $\CS$ theories.} obtained from twisted compactification of 6d $(2,0)$ theories on punctured Riemann surfaces with defects \cite{Xie:2012hs}.
    \item As endpoints of $\mathcal{N}=1$-preserving \tsf{RG} flows with accidental supersymmetry enhancement in the \tsf{IR} \cite{Maruyoshi:2016aim, Maruyoshi:2016tqk}.
\end{enumerate}    
    Having multiple constructions has allowed us to catalog various properties of these theories, including global symmetry (0-form and generalized), anomaly coefficients, moduli spaces of vacua, \tsf{BPS} spectrum and chiral ring, closed sub-rings of local operators like the Schur ring and the associated Schur \tsf{VOA}, and partition functions on various 4-manifolds and various choices of couplings to background fields, allowing access to information about compactifications and dualities. In this section, we collect some details about each of these constructions for the theories we will be working with and catalog the properties that will be of interest to us.  
\subsection{Construction within Type \tsf{IIB} string theory} 
4d $\mathcal{N}=2$ \tsf{SCFT}s can be engineered within string theory in multiple ways. One way is to consider Type \tsf{IIB} string theory on a noncompact, singular Calabi-Yau threefold (the allowed types of singularities are constrained by additional consistency conditions). For our purposes, we can restrict ourselves to isolated hypersurface singularities of certain kinds that have an A-D-E classification. More concretely, the isolated hypersurface singularities are given by $\C^2/\Gamma$ where $\Gamma\subset \mr{SU(2)}$ is a discrete subgroup. Such discrete subgroups have an \tsf{ADE} classification. For $\mf{g}\in\{A_n,D_n,E_{6,7,8}\}$, let $\Gamma_\mf{g}\subset\mr{SU(2)}$ be the discrete subgroup and let $f_{\mf{g}}(x,y)$ be a polynomial such that  $\C^2/\Gamma_{\mf{g}}$ is the zero set of $f_{\mf{g}}$. The corresponding singular Calabi-Yau threefold is given by the zero set of 
\begin{equation}
x^2+y^2+f_{\mf{g}}(z,w)~.    
\end{equation}
Explicit form of the polynomials is shown in Table \ref{tab:pol_sing_CY}.
\begin{table}[]
\centering
\begin{tabular}{|c|l|}
\hline
\(\mathfrak{g}\) & \(f_{\mathfrak{g}}(x,y)\) \\
\hline
\(A_n\)   & \(x^2 + y^{n+1} \) \\
\hline
\(D_n\)   & \(y\,(x^2 + y^{n-2}) \) \\
\hline
\(E_6\)   & \(x^3 + y^4 \) \\
\hline
\(E_7\)   & \(x\,(x^2 + y^3) \) \\
\hline
\(E_8\)   & \(x^3 + y^5 \) \\
\hline
\end{tabular}
\caption{Polynomials defining singular Calabi-Yau threefolds}
\label{tab:pol_sing_CY}
\end{table}
\paragraph{$(A_1,A_{N-1})$ Theory.}
The $(A_1,A_{N-1})$ theories are constructed using the singular Calabi-Yau given by \cite{Cecotti:2010fi}
\begin{equation}
    x^2+y^2+z^2+w^{N}=0~.
\end{equation}
This allows us to define the Seiberg-Witten curve $\Sigma$ and differential $\lambda$ as,
\begin{equation}
    \Sigma:~~z^2+w^{N}=0, \quad \lambda=z\,dw~.
\end{equation}
The above form allows us to read off the scaling dimensions of $z$ and $w$ as follows:
\begin{equation}
    \begin{split}
        &[z]+[w]=[\lambda]=1~,\quad 2[z]=N[w]  \\
    &\implies [z]=\frac{N}{N+2}, \quad [w]=\frac{2}{N+2}~.
    \end{split}
\end{equation}
Consider the simplest case of the $(A_1,A_2)$ theory. The Coulomb branch spectrum can be read off from the following deformed curve:
\begin{equation}
    z^2+w^3=u+\cdots
\end{equation}
where $u$ is the unique relevant operator appearing on the \tsf{RHS} since any other coefficient $v$ in a monomial involving $z$ and $w$ would have $[v]<1$ and therefore would not correspond to an operator in a unitary 4d $\mathcal{N}=2$ \tsf{SCFT}. 
This gives the following scaling dimensions:
\begin{equation}
    [z]=\frac{3}{5}, \quad [w]=\frac{2}{5}, \quad [u]=\frac{6}{5}~.
\end{equation}
Generalizing the above to arbitrary $N$, we find the following spectra of Coulomb branch operators:
\begin{equation}
    \begin{split}
        &\Delta\in\left\{\frac{2N}{N+2},\frac{2N-2}{N+2},\frac{2N-4}{N+2},\cdots,\frac{N+3}{N+2}\right\},\quad  N\text{ odd}~, \\
        &\Delta\in\left\{\frac{2N}{N+2},\frac{2N-2}{N+2},\frac{2N-4}{N+2},\cdots,\frac{N+4}{N+2}\right\},\quad \text{ }N\text{ even}~.
    \end{split}
\end{equation}
We therefore see that the rank (complex dimension of its Coulomb branch) of the $(A_1,A_{N-1})$ \tsf{SCFT} is $r=\lfloor\frac{N-1}{2}\rfloor$. We will be interested in the cases with $N$ odd above. Denoting $N=2n+1$ with $n\geq 1$, we can write down the central charges for these theories,
\begin{equation}
    c=\frac{n(6n+5)}{6(2n+3)}, \quad a=\frac{n(24n+19)}{24(2n+3)}~.
\end{equation}
These central charges may be derived from the Shapere-Tachikawa formula \cite{Shapere:2008zf}: 
\begin{equation}
    2a-c=\frac{1}{4}\sum_{i=1}^{r}(2\Delta_i-1)~,
\end{equation}
and the modular properties of the Schur index\footnote{The Schur index of a 4d 
$\CN=2$ \tsf{SCFT} counts certain class of operators in the theory, see Appendix \ref{app:conventions} for more details. In Type \tsf{IIB} construction, the Schur index can be obtained by computing the trace of the Kontsevich-Soibelman wall-crossing invariant $\mathcal{I}_{\mr{S}}(q)=\mr{Tr}~\mathcal{O}(q)$ \cite{Cordova:2015nma,Cecotti:2010fi}.}: 
\begin{equation}
    \lim_{q\to 1}\text{log }\mathcal{I}_{\mr{S}}(q)\approx \frac{8\pi^2}{\beta}\cdot (c-a), \quad  \quad q=e^{-\beta}~.
\end{equation}
These theories have trivial flavor symmetry for all $n$. 
\paragraph{$(A_1,D_N)$ Theory.}
The $(A_1,D_N)$ Argyres-Douglas theories are obtained from the hypersurface \cite{Cecotti:2010fi} 
\begin{equation}
x^2+y^2+wz^2+w^{N-1}=0~.    
\end{equation}
This implies the scaling dimensions
\begin{equation}
    \begin{split}
    [z]=\frac{N-2}{N}, \quad [w]=\frac{2}{N}~.
    \end{split}
\end{equation}
For the case of $(A_1,D_3)$, this leads to the following deformed Seiberg-Witten curve:
\begin{equation}
    \Sigma:~~w^2+wz^2=u~,
\end{equation}
which gives us the following scaling dimensions:
\begin{equation}
    [w]=\frac{2}{3}, \quad [z]=\frac{1}{3}, \quad [u]=\frac{4}{3}~.
\end{equation}
For general $N$, we obtain the following spectrum of Coulomb branch dimensions \cite{Cecotti:2010fi}:
\begin{equation}
    \begin{split}
        &\Delta\in\left\{\frac{2N-2}{N},\frac{2N-4}{N},\frac{2N-6}{N},\cdots,\frac{N+1}{N}\right\}, \quad  N\text{ odd}~, \\
        &\Delta\in\left\{\frac{2N-2}{N},\frac{2N-4}{N},\frac{2N-6}{N},\cdots,\frac{N+2}{N}\right\},\quad  N\text{ even}~, \\
    \end{split}
\end{equation}
so that the rank of the $(A_1,D_N)$ \tsf{SCFT} is $\lfloor\frac{N-1}{2}\rfloor$.  We will be interested in the cases with $N$ odd. Denoting $N=2n+1$ with $n\geq 1$, we can also use the Type \tsf{IIB} construction to infer the central charges for these theories:
\begin{equation}
    c=\frac{n}{2}, \quad a=\frac{n(8n+3)}{8(2n+1)}~.
\end{equation}
These theories generically have $\mr{U(1)}$ flavor symmetry, which enhances to $\mr{SU(2)}$ for $n=1$. 
\subsection{Construction As Generalized Class $\CS$ Theories}
A second construction of these theories that will be useful for us is to start from the 6d $(2,0)$ theory denoted as $\mathcal{S}[\mf{g}]$ (where $\mf{g}=\mf{a,d,e}$ is a simply laced finite-dimensional Lie algebra) and compactify on a Riemann surface with regular and irregular punctures. The theories that form the building blocks of this family of theories involve compactifications on the two-punctured Riemann sphere with one irregular puncture and one regular puncture. The 6d theory has codimension 2 defects that live in the directions transverse to these punctures at the singular locus. In the case where $\mf{g}=\mf{a}_{N-1}$ (which is what we will be concerned with), regular punctures are labeled by the data of a Young diagram $Y$ of $\mf{a}_{N-1}$ consisting of $N$ boxes \cite{Gaiotto:2009we,Chacaltana:2010ks,Tachikawa:2013kta,LeFloch:2020uop}. This corresponds to an irreducible representation of $\mf{a}_{N-1}$, and describes a boundary condition at the puncture for the Higgs field $\Phi(z)$ of the Hitchin system living on the Riemann surface. A full puncture $Y=\mr{F}$ corresponds to the Young diagram labeled $[1^N]$ in our convention, i.e., with a single row containing $N$ boxes. A simple puncture $Y=\phi$ corresponds to the Young diagram denoted as $[N]$ in our convention, and it refers to the Young diagram with $N$ rows with one box each, or the trivial irreducible representation of $\mf{a}_{N-1}$. The irregular puncture is labeled by two integers $b$ and $p$ which refer to the pole structure of the singular Higgs field. Together, this allows a packaging of the data defining the construction into the form $D_p^b(\mf{g},Y)$, which labels the 4d $\mathcal{N}=2$ \tsf{SCFT}. Some well-studied families of Argyres-Douglas theories that will be relevant to our work, constructed as generalized Class $\CS$ theories, have the following parameters in this notation \cite{Beem:2023ofp}:
\begin{equation}
    \begin{split}
        (A_1,D_{2N+2})&=D_{2N+2}^2(\mf{sl}_2,[1^2])~, \\
        (A_1,A_{2N-1})&=D_{2N+2}^2(\mf{sl}_2,[2])~, \\
        (A_1,D_{2N+3})&=D_{2N+3}^2(\mf{sl}_2,[1^2])~, \\
        (A_1,A_{2N})&=D_{2N+3}^2(\mf{sl}_2,[2])~. \\
    \end{split}
\end{equation}
This construction is useful for deriving the Macdonald index (see Appendix \ref{app:conventions} for a brief account of Macdonald index) of these theories. For example, the Macdonald index for $(A_1,D_{2N-3})=D_{2N-3}^2(\mf{sl}_2,[1^2])$ theory is given by \cite{Buican:2015tda}:
\begin{equation}\label{eq:Mac_class_S-1}
    \CI^{(A_1,D_{2N-3})}_{\mr{M}}(q,t,a)=\sum_{n=0}^{\infty}C_n\tilde{F}_n^{(N)}(q,t,a)~,
\end{equation}
where
\begin{equation}
    C_n=\bigg[\frac{(t;q)_{\infty}}{(q;q)_{\infty}} \bigg]^{1/2}\frac{P_n(q,t,\sqrt{t})}{(t^2;q)_{\infty}}~,
\end{equation}
and
\begin{equation}
    \tilde{F}_n^{(N)}(q,t,a)=\bigg[\frac{(t;q)_{\infty}}{(q;q)_{\infty}} \bigg]^{1/2}\frac{\tilde{P}_n^{(N)}(q,t,a)}{(t;q)_{\infty}}~,
\end{equation}
is the wavefunction associated to the irregular puncture. The special functions $P_n(q,t,a)$ and $\tilde{P}_n^{(N)}(q,t,a)$ respectively have the expressions
\begin{equation}
    \begin{split}
        &P_n(q,t,a)=N_n(q,t)\sum_{m=0}^{n}\frac{(t;q)_m(t;q)_{n-m}}{(q;q)_m(q;q)_{n-m}}a^{2m-n}~, \\
        &\tilde{P}_n^{(N)}(q,t,a)=N_n(q,t)\sum_{m=0}^n\frac{(t;q)_m(t;q)_{n-m}}{(q;q)_m(q;q)_{n-m}}a^{2m-n}q^{-N(\frac{n}{2}-m)^2}~,
    \end{split}
\end{equation}
and the normalization factor $N_n(q,t)$ is given as
\begin{equation}\label{eq:Mac_class_S-7}
    N_n(q,t)=\frac{(q;q)_n}{(t;q)_{\infty}}\sqrt{(1-q^nt)(q^{n+1};q)_{\infty}(q^nt^2;q)_{\infty}}
\end{equation}
For completeness, we also note that the wavefunction associated with a regular puncture is given as
\begin{equation}
    F_n(q,t,a)=\bigg[\frac{(t;q)_{\infty}}{(q;q)_{\infty}} \bigg]^{1/2}\frac{P_n(q,t,a)}{(ta^{\pm 2,0};q)_{\infty}}~.
\end{equation}

\subsection{Higgs Branch \tsf{RG} Flow And Quantum Hamiltonian Reduction}\label{sec:RG_flow}
The construction of the $(A_1,A_N)$ and the $(A_1,D_N)$ theories within the Class $\CS$ framework makes it transparent that the $(A_1,D_{2n+3})=D_{2n+3}^2(\mf{sl}_2,[1^2])$ theory flows to the $(A_1,A_{2n})=D_{2n+3}^2(\mf{sl}_2,[2])$ theory under a Higgs branch \tsf{RG} flow, implemented as a partial closure of the regular puncture from $Y=\mr{F}=[1^2]$ to $Y=\phi=[2]$. We briefly review the procedure of partial closure of the puncture applied for our case and refer the reader to \cite{Tachikawa:2015bga} for more details. Let us first recall the notion of Higgsing. Suppose a 4d $\CN=2$ \tsf{SCFT} with flavor symmetry group $G$. The flavor multiplet has a chiral operator $\mu^+$ in the adjoint of $G$. One can then give a nilpotent \tsf{VEV} to this operator which triggers a Higgs branch flow. By Jacobson-Morozov theorem, any nilpotent element of a semi-simple Lie
algebra $\mf{g}$ is given by the image $\rho(\sigma^+)$ of the raising operator $\sigma^+\in\mf{su}(2)$ of an embedding $\rho: \mf{su}(2) \longrightarrow \mf{g}$. The flavor symmetry algebra of the \tsf{IR} fixed point\footnote{Note that, in general, this procedure does not preserve the full $\mr{SU(2)_R}$-symmetry of the \tsf{UV} theory and is broken to $\mr{U(1)_R}$.} then is the commutant of $\rho(\sigma^+)$ in $\mf{g}:=\mr{Lie}(G)$. By the usual Higgs mechanism, this spontaneous breaking of flavor symmetry results in some Nambu-Goldstone hypermultiplets at the \tsf{IR} fixed point. 
\par In a class $\CS$ theory, flavor symmetries are associated to the defects at the punctures of the Riemann surface. In case of a class $\CS$ theory of type $A_{n-1}$, Higgsing of the flavor symmetry associated to a regular puncture is called partial closure of puncture. The $(A_1,D_{2n+3})$ Argyres-Douglas theory has $\mr{SU(2)}$ flavor symmetry. The chiral operator $\mu^+$ is in the adjoint of $\mr{SU(2)}$. The only choice of nilpotent \tsf{VEV} for $\mu^+$ is
\begin{equation}
    \langle\mu^+\rangle=\sigma^+~.
\end{equation}
The flavor symmetry then breaks from $\mr{SU(2)}$ to the commutant of $e^{i\sigma^+}$ which is $\mr{U(1)}$. At the level of defect at the puncture, the nilpotent \tsf{VEV} $\langle\mu^+\rangle=\sigma^+$ closes the puncture $[1^2]$ to $[2]$. The \tsf{IR} fixed point\footnote{One can show that this partial closure of puncture the \tsf{UV} $\mr{SU(2)_R}$-symmetry breaks to $\mr{U(1)_R}$ in the \tsf{IR} but sits as the Cartan of an $\mr{SU(2)'_R}$ which becomes the \tsf{R}-symmetry of the \tsf{IR}. The \tsf{IR} fixed point is thus a 4d $\CN=2$ \tsf{SCFT}. In the language of Tachikawa \cite{Tachikawa:2015bga}, this partial closure of puncture is \textit{good}.} is thus the $(A_1,A_{2n})$ Argyres-Douglas theory with flavor symmetry $\mr{U(1)}$. Following the counting of free Nambu-Goldstone hypermultiplets in \cite{Tachikawa:2015bga}, we see that the partial closure of puncture $[1^2]\to[2]$ results in a single free hypermultiplet. We thus conclude that the Higgs branch \tsf{RG} flow of the $(A_1,D_{2n+3})$ theory results in a single free hypermultiplet and the $(A_1,A_{2n})$ theory. We summarize this as 
\begin{equation}\label{eq:RG_flowD2n+1to2n}
(A_1,D_{2n+3})\xrightarrow[\tsf{RG flow}]{\langle\mu^+\rangle=\sigma^+} (A_1,A_{2n})\otimes 1\,\mr{HM}~.   
\end{equation}
At the level of the corresponding Schur \tsf{VOA}s, the analogous operation is the quantum Drinfeld-Sokolov (or quantum Hamiltonian) reduction (see Section 4 of \cite{Beem:2014rza}). This is a special case of the \tsf{BRST} procedure applied to \tsf{VOA}s, typically affine Lie \tsf{VOA}s $V_k(\mathfrak{g})$. It requires the input data of the affine \tsf{VOA} $V_k(\mathfrak{g})$ and a Lie algebra homomorphism $\rho:\mf{su}(2)\longrightarrow\mathfrak{g}$. The output is a second \tsf{VOA}, typically a $\mathcal{W}$-algebra obtained by ``gauging'' or the imposition of certain constraints (related to the nilpotent orbit $\rho$) on the parent \tsf{VOA} by the \tsf{BRST} procedure. Additional useful properties appear (e.g., modular properties of the characters of modules) when the parent affine \tsf{VOA} $V_k(\mathfrak{g})$ are at \textit{boundary admissible levels} in the sense of Kac and Wakimoto, i.e., there exists $q\geq h^\vee$ with $\mr{gcd}(q,h^\vee)=1$ such that
\begin{equation}
    \tilde{k}=-h^{\vee}+\frac{h^{\vee}}{q}~, 
\end{equation}
where $h^\vee$ is the dual Coxeter number of $\mf{g}$. For the \tsf{RG} flow \eqref{eq:RG_flowD2n+1to2n}, the quantum Drinfeld-Sokolov reduction is \cite{adamovic2019realizations}
\begin{equation}
    \widehat{\mf{su}}(2)_{-\frac{4n+4}{2n+3}}\to\mr{Vir}\left(-\frac{2n(6n+5)}{2n+3},0\right)~.
\end{equation}
Based on this, we expect that the indices in the theories related by the \tsf{RG} flow \eqref{eq:RG_flowD2n+1to2n} will have some factorization property\footnote{See \cite{Deb:2025ypl} for an intriguing relation of Schur indices of theories lying on the same \tsf{RG} flow line.}. In Section \ref{sec:D3_free_hyper_rel}, we will see that while \underline{the full Macdonald index is not factorizable}, there exists a simple relation between the full index and a counting function that is a product of factors derived from the IR theory under the above RG flow. We will exploit this connection to conjecture new expressions for the Macdonald indices of the corresponding theories.
\par
It has been suggested by the authors of \cite{Beem:2019tfp} that an inversion of the quantum Drinfeld-Sokolov reduction is implemented by their free field realizations\footnote{See also \cite{Beem:2024fom} for free-field realizations of rank 1 \tsf{SCFT}s.} for the corresponding Schur \tsf{VOA}s. This aligns with the spirit of this work, where we use analogous machinery to these free field realizations to come up with a new proposal for the Macdonald indices of these theories. Note that while the free field realizations of \cite{Beem:2019tfp} come with an R-filtration and therefore can be used to obtain the Macdonald index, our techniques do not seem to derive in any obvious way from this R-filtration, and the operator rearrangement scheme we work out here seems to have a different (though assuredly related) origin. \\
\section{Macdonald Index And Zhu's $C_2$-Algebra}\label{sec:Mac_Hil_ser_arc}
In \cite{Song:2016yfd}, Song conjectured an algorithmic procedure to compute the Macdonald index of a 4d $\CN=2$ theory using only the data of the corresponding \tsf{VOA}. \cite{Song:2016yfd} defined a certain refined character for a strongly finitely generated \tsf{VOA} (see \eqref{eq:ref_char_song}) and checked\footnote{See \cite{Watanabe:2019ssf} for more checks.} that it reproduces the Macdonald index for a large class of 4d $\CN=2$ \tsf{SCFT}s.   In this section, we relate the refined character of a \tsf{VOA} satisfying certain properties to the Hilbert series of the arc space of its Zhu's $C_2$-algebra.

\subsection{The Refined Character Of A \tsf{VOA}}\label{sec:ref_char_song}
We begin by recalling the notion of the refined character of a \tsf{VOA}. Recall that a \tsf{VOA} $\CV$ is said to be \textit{strongly generated} by a subset $U\subseteq \CV$ if  
\begin{equation}
    \CV=\mathsf{Span}\left\{u^1_{(-n_1)}\cdots u^k_{(-n_k)}|0\rangle:u^i\in U,n_i\geq 1,k\geq 0\right\}/\{\mr{null~~ vectors}\}~,
\end{equation}
where $k=0$ corresponds to the vacuum $\ket{0}$ and our convention for the expansion of the vertex operators is 
\begin{equation}
    Y(u,z)=\sum_{r\in\Z}u_{(-r)}z^{-r-1}~.
\end{equation}
If $U$ is a finite-dimensional subspace, then $\CV$ is called \textit{strongly finitely generated}. Note that in this convention, we have 
\begin{equation}
    L_0\cdot u^1_{(-r_1)}\cdots u^k_{(-r_k)}|0\rangle=\left(\sum_{i=1}^k(d_i+r_i-1)\right)u^1_{(-r_1)}\cdots u^k_{(-r_k)}|0\rangle~,
\end{equation}
where $d_i$ is the conformal dimension of $u^i$.

The notion of the refined character, as we will see below, requires the \tsf{VOA} to be strongly generated. It is widely believed that the \tsf{VOA} arising from a 4d $\CN=2$ \tsf{SCFT} is strongly finitely generated \cite{Beem:2019tfp,Arakawa:2016hkg}. We will need this assumption to prove our statement.

A \tsf{VOA} $\CV$ strongly generated by $U$ admits the filtration
\begin{equation}\label{eq:song_filtration}
    \CV_0\subset\CV_1\subset\CV_2\subset\dots
\end{equation}
where 
\begin{equation}\label{eq:CV_k_def}
    \CV_k=\mr{Span}\left\{u_{(-r_1)}^{1} u_{(-r_2)}^{2} \cdots u_{(-r_m)}^{m}|0\rangle: r_1 \geq r_2 \geq \ldots \geq r_m>0, m \leq k;u^i\in U\right\}/\{\mr{null~~ vectors}\}~.
\end{equation}
Notice that if $\CV$ has a one-dimensional vacuum space generated by the vacuum vector, then $\CV_0=\C|0\rangle$. 
Using this filtration, we can define the graded space 
\begin{equation}
    \mr{Gr}(\CV):=\CV^{(0)}\oplus\bigoplus_{r\geq 1}\CV^{(r)}~,
\end{equation}
where 
\begin{equation}
\CV^{(0)}:=\CV_0=\C|0\rangle~,\quad \CV^{(r)}:=\CV_r/\CV_{r-1}~.   
\end{equation}
We can further decompose $\mr{Gr}(\CV)$ into $L_0$-homogeneous subspaces:
\begin{equation}
\mr{Gr}(\CV)=\bigoplus_{r\geq 0}\bigoplus_{h\in\Z}\CV^{(r)}_h~,    
\end{equation}
where $\CV^{(r)}_h$ is the subspace of $\CV^{(r)}$ with $L_0$-eigenvalue $h$. 
Let $\mf{g}$ be the algebra spanned by the zeromodes of the generators of $\CV$. Then the homogeneous spaces $\CV^{(r)}_h$ are $\mf{g}$-modules. 

The refined character of $\CV$ is then defined by\footnote{Note the difference in the exponent of $q$ compared to \cite[Eq. (1.8)]{Song:2016yfd}. The difference arises because of our conventions and can be related to \cite{Song:2016yfd} using $T=t/q$. Our conventions are summarized in Appendix \ref{app:conventions}.}
\begin{equation}\label{eq:ref_char_song}
Z_{\mathcal{V}}^{\mathsf{ref}}(\boldsymbol{z} ; q, t)=\sum_{r=0}^{\infty} \sum_{h\in\Z} \mr{ch}\left(\CV_h^{(r)} ; \boldsymbol{z}\right) q^{h-r} t^r~,
\end{equation}
where  
\begin{equation}
\mr{ch}\left(\CV_h^{(r)} ; \boldsymbol{z}\right)=\mr{Tr}_{\CV_h^{(r)}}\left(\prod_a \boldsymbol{z}_a^{\mr{F}_a}\right)    
\end{equation}
denotes the character of the (non-affine) symmetry $\mathfrak{g}$ on $\CV_h^{(r)}$. Here $\mr{F}_a$ are the Cartan generators of $\mathfrak{g}$. 
\begin{conj}\label{conj:song}\cite[Song]{Song:2016yfd}
The Macdonald index of a 4d $\CN=2$ \tsf{SCFT} with associated \tsf{VOA} $\CV$ is the refined index of $\CV$:
\begin{equation}\label{eq:song_conj}
    \mathcal{I}_{\mr{M}}(\boldsymbol{z};q,t)=Z_{\mathcal{V}}^{\mathsf{ref}}(\boldsymbol{z} ; q, t)~.
\end{equation}
\end{conj}
We remark that in defining the refined character, we assigned the bigrading 
\begin{equation}\label{eq:bigrade_Vn}
    \mr{wt}\left(u_{(-r_1)}^{1} u_{(-r_2)}^{2} \cdots u_{(-r_m)}^{m}|0\rangle\right)=(h-m,m)~,
\end{equation}
where $h$ is the conformal dimension of the state. In other words, the second grading was assigned as $w(u^i)=1$ to each of the strong generators. This is strictly not correct for all strong generators \cite{Song:2016yfd}. For generators of an affine vertex operator algebra and the stress tensor, it is 1 but in general it is not inherently defined in the 2d chiral algebra. If we could track the 4d origin of the strong generators, then there is a consistent way of assigning the weight. Finding a consistent 2d assignment of the second weight is called the $\mr{R}$-filtration problem. Up to these limitations, Song's conjecture gives us a prescription to compute the Macdonald index from the 2d chiral algebra.

\subsection{The Arc Space Of The Zhu's $C_2$-Algebra And Its Hilbert Series}
In this section, we recall the notion of Zhu's $C_2$-algebra, the Hilbert series of the arc space of a variety and prove that the Macdonald index of a 4d $\CN=2$ \tsf{SCFT} is given by the Hilbert series of the arc of space of the Zhu's $C_2$-algebra of the corresponding \tsf{VOA}. Arc spaces of vertex algebras have become an important tool in understanding the Schur sector of 4d $\CN=2$ \tsf{SCFT}s \cite{arakawa2021arc,ArakawaMoreau2021,arakawa2019singular}.
\subsubsection{The Zhu's $C_2$-Algebra Of A (Strongly Finitely Generated) \tsf{VOA}}
Let us begin by recalling the notion of Zhu's $C_2$-algebra for a \tsf{VOA}. 

Let $\CV$ be a \tsf{VOA}. Define the subspace 
\begin{equation}
    C_2(\CV):=\{a_{(-2)}\cdot b:a,b\in\CV\}~.
\end{equation}
Then the Zhu's $C_2$-algebra of $\CV$ is defined by 
\begin{equation}
    R_{\CV}:=\CV/C_2(\CV)~.
\end{equation}
It is well known that $R_{\CV}$ is a commutative associative algebra with the multiplication induced by the normal ordered product of vertex operators \cite{Zhu1995ModularIO}. The \tsf{VOA} $\CV$ is called $C_2$-\textit{cofinite} (or \textit{lisse}) if $R_{\CV}$ is finite dimensional. Note that $R_{\CV}$ inherits a $\Z$-grading from that of $\CV$. 

It is clear that $\CV$ is finitely strongly generated if and only if $R_{\CV}$ is finitely generated, see for example \cite[Corollary 2.6.2]{arakawa2012remark}. In particular, if $\CV$ is strongly generated by $U$ with basis 
\begin{equation}
    \mathcal{B}_U=\{u^1,\dots,u^N\}~,
\end{equation}
then 
\begin{equation}
    R_{\CV}\cong\C[x^{(1)},\dots,x^{(N)}]/I~,
\end{equation}
for some ideal $I$. We want to determine the ideal $I$. Roughly speaking, the ideal corresponds to images in Zhu's $C_2$-algebra of all the singular vector relations among the vertex operators of $u^1,\dots,u^N$. It was conjectured in \cite{Xie:2019zlb} that when $\CV$ is a $\mathcal{W}$-algebra then $I$ is generated by these singular vector relations. We now show that for any strongly finitely generated \tsf{VOA}, the ideal $I$ is generated by the images in Zhu's $C_2$-algebra of the singular vector relations, assuming that the conformal vector and vacuum are not singular vectors. 
Let $\CN$ be the ideal of 
\begin{equation}
\ov{\CV}:= \mathsf{Span}\left\{u^{i_1}_{(-r_1)}\cdots u^{i_k}_{(-r_k)}|0\rangle:r_i\geq 1,k\geq 0\right\}~,   
\end{equation}
generated by the singular vectors, so that $\CV=\ov{\CV}/\CN$. 
By \cite[Proposition 1.4.2]{frenkel1992vertex}, we have 
\begin{equation}
R_{\CV}=R_{\ov{\CV}}/R_{\CN}~,    
\end{equation}
where $R_{\CN}$ is the image of $\CN$ in $R_{\ov{\CV}}$ under the projection map. It is clear that the map
\begin{equation}
\begin{split}
    R_{\ov{\CV}}&\longrightarrow\C[x^{(1)},\dots,x^{(N)}] 
    \\
    [u^{i_1}_{(-1)}\cdots u^{i_k}_{(-1)}|0\rangle]&\longmapsto x^{(i_1)}\cdots x^{(i_k)},~~~~~k\geq 0~,
\end{split} 
\end{equation}
is an isomorphism of algebras. 
We claim that $R_{\CN}$ is generated by singular vector relations. Given a singular vector, it corresponds to a polynomial 
\begin{equation}
    F(Y(u^1,z),\partial Y(u^1,z),\dots;Y(u^2,z),\partial Y(u^2,z),\dots; Y(u^N,z),\partial Y(u^N,z),\dots)~,
\end{equation}
where $\dots$ mean higher order derivatives of the vertex operators, 
and the product being normal ordered product of vertex operators. In the Zhu's $C_2$-algebra, these polynomials are equivalent to  
\begin{equation}
    F'(Y(u^1,z),Y(u^2,z),\dots,Y(u^N,z))~,
\end{equation}
obtained from $F$ by setting all the monomials containing at least one derivative to zero. Since $\CN$ is generated by singular vectors,  $R_{\CN}$ is generated by the corresponding polynomials.
\subsubsection{Arc Spaces And Hilbert Series}
Let us now define the arc space of an algebraic variety
\begin{equation}\label{eq:variety_R}
    R=\C[x^{(1)},\dots,x^{(N)}]/\langle f_1,\dots,f_K\rangle~.
\end{equation}
We will follow \cite{bruschek2011arc} for our exposition. 
Suppose $R$ is graded with the grading defined by declaring 
\begin{equation}
    \mr{wt}(x^{(i)})=d_i\geq 1~.
\end{equation}
Without loss of generality, we can assume that the polynomials $f_1,\dots,f_K$ are homogeneous with respect to the above grading. We begin by writing the formal power series
\begin{equation}\label{eq:x(s)_exp}
    x^{(i)}(w)=\sum_{s=d_i-1}^\infty x^{(i)}_sw^s\in\left(\C[[w]]\right)^N~,\quad i=1,\dots,N
\end{equation}
Let us write the variables as a vector 
\begin{equation}
\vec{x}(w)=\left(x^{(1)}(w),\dots,x^{(N)}(w)\right)~.    
\end{equation}
For $s\in\mathds{Z}_{\geq 0}$, the $s$-th Jet space of $R$ is defined by
\begin{equation}
    R_s:=\C[x^{(1)}_{d_1-1},\dots,x^{(1)}_s,x^{(2)}_{d_2-1},\dots,x^{(2)}_s,\dots,x^{(N)}_{d_N-1},\dots,x^{(N)}_s]/I_s~,
\end{equation}
where 
\begin{equation}
    I_s:=\left\langle f_1(\vec{x}(w)),\dots,f_K(\vec{x}(w))\equiv 0\bmod w^{s+1}\right\rangle~.
\end{equation}
The $s$-th Jet scheme is defined as 
\begin{equation}
    X_s(R):=\mr{Spec}\,R_s~.
\end{equation}
The \textit{arc space} $X_\infty(R)$ of $R$ is defined to be 
\begin{equation}
X_\infty(R)=\lim_{\leftarrow}X_s(R)\cong\mr{Spec}\,R_\infty~,     
\end{equation}
where 
\begin{equation}
 R_\infty:=\C[x^{(1)}_{i_1},\dots,x^{(N)}_{i_N},i_k\geq d_{k}-1]/\left\langle f_1(\vec{x}(w)),\dots,f_K(\vec{x}(w))\right\rangle~.   
\end{equation}
$R_\infty$ admits a filtration
\begin{equation}
R_0\subset R_1\subset R_2\subset\dots    
\end{equation}
Assign bigrading to $R_\infty$ by declaring 
\begin{equation}\label{eq:wt_arc_space_var}
    \mr{wt}^{(2)}(x_j^{(i)}):=\left(\mr{wt}_1(x_j^{(i)}),\mr{wt}_2(x_j^{(i)})\right):=(j,1)~.
\end{equation}
The Hilbert series of $R_\infty$ is then defined by 
\begin{equation}\label{eq:Hilber_ser_def}
H_{R_\infty}(q,t)=\sum_{\ell,m\geq 0} \mathsf{dim}(R_\infty)_{\ell,m}q^\ell t^m~,    
\end{equation}
where $(R_\infty)_{\ell,m}\subset R_\infty$ is the homogeneous subspace with weight $(\ell,m)$.

When $R=R_{\CV}$ is the Zhu's $C_2$-algebra of a \tsf{VOA}, then the grading of the generators is simply the conformal dimension of the operators. 
\subsubsection{Classically free \tsf{VOA}}\label{sec:classical_free}
In this section, we recall the notion of classically free \tsf{VOA}s which will be useful.

Given any \tsf{VOA} $\CV$, for each $p\geq 0$, define the subspace 
\begin{eqsp}
    F^p\CV:=\mr{Span}\{a^1_{-r_1-1}\cdots a^r_{-r_s-1}|0\rangle:s\geq 0, a^i\in \CV,r_i\in\IZ_{\geq 0},r_1+\dots+r_s\geq p\}/\{\mr{null~ vectors}\}~.
\end{eqsp}
For $p<0$, define $F^p\CV=\CV$. 
Then there is a natural decreasing filtration on $\CV$:
\begin{eqsp}
    F^0\CV\supseteq F^1\CV\supseteq\cdots
\end{eqsp}
One can easily check that 
\begin{eqsp}
  F^0\CV=\CV~,\quad F^1\CV=C_2(\CV)~.  
\end{eqsp}
This is called the \textit{Li filtration} \cite{li2005abelianizing}. 
In particular, the Zhu's $C_2$-algebra is given by 
\begin{eqsp}
    R_{\CV}=\CV/C_2(\CV)=F^0\CV/F^1\CV~.
\end{eqsp}
Let 
\begin{eqsp}
    \mr{gr}_F(\CV):=\bigoplus_{p\geq 0}F^p\CV/F^{p+1}\CV~, 
\end{eqsp}
be the associated graded algebra. It was proved in \cite{li2005abelianizing} that $\mr{gr}_F(\CV)$ is generated by the Zhu's $C_2$-algebra $R_{\CV}$. By construction of the arc space $(R_{\CV})_\infty$, there is a surjective map \cite{li2005abelianizing,arakawa2012remark,arakawa2019singular}: 
\begin{eqsp}\label{eq:iso_vphi}
    \varphi:(R_{\CV})_\infty\twoheadrightarrow \mr{gr}_{F}(\CV)~.
\end{eqsp}
But in general, this map has a kernel. A \tsf{VOA} is said to be \textit{classically free} if this map is also injective, i.e., $\varphi$ is an isomorphism of (Poisson vertex) algebras \cite{van2024first}. 
For strongly finitely generated \tsf{VOA}, we can write the map  explicitly as\footnote{Note that because of our convention for the formal series \eqref{eq:x(s)_exp}, the modes have been shifted by the conformal dimension of the corresponding generator.}
\begin{eqsp}\label{eq:voa_arc_li}
  [x^{(i_1)}_{k_1}\cdots x^{(i_s)}_{k_s}] \mapsto [u^{i_1}_{(-k_1+d_{i_1}-2)}\cdots u^{i_s}_{(-k_s+d_{i_s}-2)}|0\rangle]~,
\end{eqsp}
where $[\cdot]$ indicates the equivalence class.
It is clear that this map is well defined and surjective, but it is not injective in general. There are known examples for which this map is not injective. For the $(p,p')$ Virasoro \tsf{VOA}s, this map is injective precisely when $p=2$ \cite{van2021chiral}. For example, for the $(3,4)$ minimal model (the Ising model), it was proved in \cite{andrews2023singular} that the kernel of the above map is nontrivial. Another example of a \tsf{VOA} which is not classically free was found in \cite{arakawa2019singular}. 
\begin{remark}
Incidentally, the $(p,p')$ Virasoro \tsf{VOA}s that appear as the Schur \tsf{VOA}s of 4d $\CN=2$ \tsf{SCFT}s is for $p=2, p'=2n+3,n\geq 1$. These are the Schur \tsf{VOA}s of the $(A_1,A_{2n})$ Argyres-Douglas theory \cite{beem2014infinite}.  
\end{remark}
It is useful to note that the map \eqref{eq:voa_arc_li} preserves the bigrading where the bigrading on the \tsf{VOA} side is defined by 
\begin{eqsp}\label{eq:bigrade_FpV}
\mr{wt}\left(u^{i_1}_{(-k_1)}\cdots u^{i_s}_{(-k_s)}|0\rangle\right)=\left(h-s,s\right)~,   
\end{eqsp}
where 
\begin{eqsp}
    h=\sum_{j=1}^s(k_j+d_{i_j}-1)~,
\end{eqsp}
is the conformal dimension of the state. 
Consequently, if we define the refined character of the \tsf{VOA} using this bigrading, then it will agree with the Hilbert series of the arc space of its Zhu's $C_2$-algebra if the \tsf{VOA} is classically free. 
\par
There is another filtration on $\CV$ which will be useful \cite{arakawa2012remark,arakawa2021arc}. For each $p\geq 0$, define the subspace 
\begin{eqsp}
    G_p\CV:=\mr{Span}\{u^{i_1}_{-r_1}\cdots u^{i_s}_{-r_s}|0\rangle:s\geq 0, r_i\in\IZ_{\geq 0},d_{i_1}+\dots+d_{i_s}\leq p\}/\{\mr{null~ vectors}\}~,
\end{eqsp}
where $d_{i}$ denotes the conformal dimension of the strong generator $u^i$. For $p<0$, put $G_p\CV=0$. This defines an increasing filtration on $\CV$:
\begin{eqsp}
    G_p\CV\subset G_{p+1}\CV~,\quad p\in\IZ~,
\end{eqsp}
and we can consider the associated graded space 
\begin{eqsp}
    \mr{gr}_{G}\CV:=\bigoplus_{p\geq 0}G_{p+1}\CV/G_p\CV~.
\end{eqsp}
Let $\CV^h$ denote the $L_0=h$-eigenspace of $\CV$. Then we can write 
\begin{eqsp}
    G_p\CV=\bigoplus_{h\in\IZ}G_p\CV^h~,\quad F^p\CV=\bigoplus_{h\in\IZ}F^p\CV^h~.
\end{eqsp}
It turns out that \cite{arakawa2012remark}
\begin{eqsp}\label{eq:GpFh-p}
    G_p\CV^h=F^{h-p}\CV^h~.
\end{eqsp}
This induces an isomorphism of the graded spaces
\begin{eqsp}\label{eq:iso_gr_F-G}
    \mr{gr}_F\CV\cong\mr{gr}_G\CV~.
\end{eqsp}
Note also that for affine \tsf{VOA}s, since the strong generators have conformal dimension 1, the filtrations $G_p\CV$ and $\CV_p$ defined in \eqref{eq:CV_k_def} agree:
\begin{eqsp}
    G_p\CV=\CV_p~.
\end{eqsp}
For affine \tsf{VOA}s then, \eqref{eq:GpFh-p} implies that 
\begin{eqsp}
    \CV_n^h:=\CV_n\cap\CV^h=F^{h-n}\CV^h~.
\end{eqsp}
Similarly, for Virasoro \tsf{VOA}s, since the conformal dimension of the generator is 2, we have 
\begin{eqsp}
    \CV_p=G_{2p}\CV~,
\end{eqsp}
and hence 
\begin{eqsp}
    \CV^h_n=F^{h-2n}\CV^h~.
\end{eqsp}
We conclude that 
\begin{eqsp}
    \CV^{(n)}_h=\begin{cases}
        F^{h-n}\CV_h/F^{h-n+1}\CV_h,&\CV\text{~~ affine \tsf{VOA}}~,
        \\
        F^{h-2n}\CV_h/F^{h-2n+2}\CV_h,&\CV\text{~~ Virasoro \tsf{VOA}}~.
    \end{cases}
\end{eqsp}
Note that $F^{h-2n+1}\CV_h=\varnothing$. This proves the following proposition.
\begin{prop}\label{prop:affine_vir}
Let $\CV$ be an affine $\mr{VOA}$ or Virasoro $\mr{VOA}$. Then the    refined character of $\mr{Gr}(\CV)$ and $\mr{gr}_F\CV$ defined using the bigrading \eqref{eq:bigrade_Vn} and \eqref{eq:bigrade_FpV} are identical.  
\end{prop} 
\subsubsection{Macdonald Index From Hilbert Series Of Arc Space}
In Section \ref{sec:ref_char_song}, we discussed Song's conjecture which states that the refined character of the bigraded algebra $\mr{Gr}(\CV)$ of a Schur \tsf{VOA} $\CV$ is the Macdonald index of the corresponding 4d $\CN=2$ \tsf{SCFT}. We also saw in Section \ref{sec:classical_free} that there is a natural surjection from the arc space of the Zhu's $C_2$-algebra of a \tsf{VOA} to the graded space $\mr{gr}_F(\CV)$, which is an isomorphism when the \tsf{VOA} is classically free. In this section, we will consider a similar surjective map between the arc space of the Zhu's $C_2$-algebra $(R_{\CV})_\infty$ and $\mr{Gr}(\CV)$. 

Let $u^1,\dots,u^N$ be a set of strong generators of $\CV$ with conformal dimension $d_1,\dots,d_N$ respectively, and $x^{(1)},\dots,x^{(N)}$ be the corresponding generators in the Zhu's $C_2$-algebra. By definition, $\CV^{(n)}_h$ is given by 
\begin{equation}
    \CV^{(r)}_h=\mr{Span}\left\{\left[u_{(-k_1)}^{i_1} u_{(-k_2)}^{i_2} \cdots u_{(-k_r)}^{i_r}|0\rangle\right]: k_1 \geq k_2 \geq \ldots \geq k_r\geq 1;\sum_{j=1}^r(d_{i_j}+k_j-1)=h \right\}/\{\mr{null~~ vectors}\}~,
\end{equation}
where $i_1,\dots,i_r\in\{1,\dots,N\}$.
Define the map
\begin{equation}\label{eq:voa_arc_space_zhu}
\begin{split}
\psi:(R_{\CV})_\infty&\longrightarrow \mr{Gr}(\CV)
    \\
    [x^{(i_1)}_{k_1+d_{i_1}-2}\cdots x^{(i_r)}_{k_r+d_{i_r}-2}]&\longmapsto [u_{(-k_1)}^{i_1} u_{(-k_2)}^{i_2} \cdots u_{(-k_r)}^{i_r}|0\rangle]~,
\end{split}    
\end{equation}
where $[\cdot]$ indicates the equivalence class. Note that this restricts to a map 
\begin{equation}\label{eq:iso_ref_char_Hilbert_ser}
\left((R_{\CV})_\infty\right)_{h-r,r}\longrightarrow  \CV^{(r)}_h~.    
\end{equation}
We first show that the map is well defined. Let us write
\begin{equation}
    R_{\CV}=\C[x^{(1)},\dots,x^{(N)}]/\langle f_1,\dots,f_K\rangle~,
\end{equation}
where $f_1,\dots,f_K$ are homogeneous polynomials with respect to the induced grading on $R_{\CV}$, corresponding to $K$ singular vectors in 
\begin{equation}
    \mathsf{Span}\left\{u^{i_1}_{(-r_1)}\cdots u^{i_k}_{(-r_k)}|0\rangle:r_i\geq 1,k\geq 1\right\}~.
\end{equation}
These polynomials correspond to polynomials 
\begin{equation}
    F_i(Y(u^1,z),\partial Y(u^1,z),\dots;Y(u^2,z),\partial Y(u^2,z),\dots; Y(u^N,z),\partial Y(u^N,z))~,\quad i=1,\dots,K~.
\end{equation}
with the product being normal ordered product of vertex operators obtained from singular vectors. In the Zhu's $C_2$-algebra, these polynomials are equivalent to  
\begin{equation}\label{eq:Fi'_sing_rel}
    F'_i(Y(u^1,z),Y(u^2,z),\dots,Y(u^N,z))~,\quad i=1,\dots,K~,
\end{equation}
obtained from $F_i$ by setting all the monomials containing at least one derivative to zero. In particular, note that $F_i'=f_i$. It is clear that the map is surjective. In general this map is not injective. For example, $(p,p')$ Virasoro \tsf{VOA}s with $p>2$ are known to be not classically free \cite{van2021chiral,andrews2023singular}, Proposition \ref{prop:affine_vir} implies that the above map is not injective. Proposition  \ref{prop:affine_vir} also implies that the map $\psi$ is injective (and hence an isomorphism) for 
\begin{eqsp}
    \CV=\mr{Vir}\left(c=-\frac{2n(5+6n)}{3+2n},0\right)~,
\end{eqsp} 
where $\mr{Vir}(c,0)$ denotes the \textit{simple} Virasoro \tsf{VOA} with central charge $c$, since these are precisely the Virasoro \tsf{VOA}s which are classically free \cite{van2021chiral}. The affine \tsf{VOA}s $\widehat{\mf{su}}(2)_k$ at admissible levels
\begin{eqsp}
    k=-2+\frac{p}{p'}~,\quad p,p'\geq 2, (p,p')=1~,
\end{eqsp}
is known not to be classically free for 
\begin{eqsp}
    p'>p>2~,
\end{eqsp}
and is expected to be classically free at boundary admissible values
\begin{eqsp}
    k=-2+\frac{2}{p'}~.
\end{eqsp}
Thus by Proposition \ref{prop:affine_vir}, we expect that map $\psi$ to be an isomorphism for $\widehat{\mf{su}}(2)_{-2+\frac{2}{p'}}$. Indeed, these \tsf{VOA}s for $p'=2n+1$ are the Schur \tsf{VOA}s of the $(A_1,D_{2n+1})$ Argyres-Douglas theories \cite{beem2014infinite} and we will give strong evidence that the map $\psi$ is an isomorphism (equivalently $\widehat{\mf{su}}(2)_{-2+\frac{2}{p'}}$ is classically free) by computing the Hilbert series of the arc space of its Zhu's $C_2$-algebra.   
\\\\
Noting that 
\begin{equation}
    \mr{ch}(\CV^{(r)}_h;1)=\mr{dim}~\CV^{(r)}_h~,
\end{equation} 
and assuming Song's conjecture,
we have proved the following theorem. 
\begin{thm}\label{thm:Hil_ser_mac}
Let $\CV$ be a strongly finitely generated $\mr{VOA}$ such that the map $\psi$ defined in \eqref{eq:voa_arc_space_zhu} is injective. Then the refined character \eqref{eq:ref_char_song} of $\CV$ with the fugacities for (non-affine) symmetry $\mf{g}$ set to $1$ is equal to the Hilbert series for the arc space of the Zhu's $C_2$-algebra of $\CV$. 
\end{thm}
It is believed that the Schur \tsf{VOA} of a 4d $\CN=2$ \tsf{SCFT} is strongly finitely generated (see Footnote $16$ of \cite{Beem:2017ooy}). We also know that the Schur VOAs of the $(A_1,A_{2n})$ and the $(A_1,D_{2n+3})$ SCFTs (which will be the subject of interest in this work) are classically free \cite{van2021chiral}. Conjecture \ref{conj:song} then combined with Theorem \ref{thm:Hil_ser_mac} motivates the following conjecture.
\begin{conj}
The Macdonald index with flavor fugacities set to 1 of a 4d $\CN=2~\mr{SCFT}$ whose Schur \tsf{VOA} is strongly finitely generated and classically free is given by the Hilbert series of the arc space of the Zhu's $C_2$-algebra of the corresponding Schur \tsf{VOA}:
\begin{equation}
\mathcal{I}_{\mr{M}}(1;q,t)= H_{(R_{\CV})_\infty}(q,t)~.   
\end{equation}    
\end{conj}
\subsection{Computation Of The Hilbert Series}\label{sec:comp_Hilb_ser}
In this section, we explain an algorithmic procedure for the computation of the Hilbert series of the arc space of a variety. The algorithm uses Gr\"{o}bner basis for a finitely generated ideal, relevant definitions and results have been reviewed in Appendix \ref{app:Grob}. 

Assume the notation of \eqref{eq:variety_R} -- \eqref{eq:Hilber_ser_def}. The Hilbert series of $R_n$ is defined as
\begin{equation}
H_{R_s}(q, t)=\sum_{\ell=0}^s \sum_{m \geq 0} \mr{dim}\left(R_s\right)_{\ell, m} q^{\ell} t^m~.
\end{equation}
Taking $s \rightarrow \infty$, we get the Hilbert series of  $R_{\infty}$:
\begin{equation}
H_{R_{\infty}}(q, t)=\sum_{\ell, m \geq 0} \mr{dim}\left(R_{\infty}\right)_{\ell, m} q^{\ell} t^m~.
\end{equation}
One way of computing the Hilbert series of an $s$-th jet space $R_s$ is by using Gr\"{o}bner basis of $I_s$.
Let us denote the Gr\"{o}bner basis of $I_s$ by $G_s$. Note that $G_s$ depends on the choice of ordering. Let us denote the set of leading monomials of the polynomials in $G_s$ by
\begin{equation}
\mr{LM}(G_s):=\{\mr{LM}(g)\mid g\in G_s\}~.
\end{equation}
We call the monomials in $\mr{LM}(G_s)$ as \textit{constraints}. 
Then, we can find a canonical basis of $R_s$ as explained in Appendix \ref{app:Grob}. Let us denote a monomial in $\C[x^{(1)}_{d_1-1},\dots,x^{(1)}_s,x^{(2)}_{d_2-1},\dots,x^{(2)}_s,\dots,x^{(N)}_{d_N-1},\dots,x^{(N)}_s]$ by
\begin{equation}
    X_1^{\vec{a}_1}X_2^{\vec{a}_2}\cdots X_N^{\vec{a}_N}~,\quad \vec{a}_k\in (\IZ_{\geq 0})^s
\end{equation}
where
\begin{equation}
    X_k^{\vec{a}_k}:=(x_{d_k-1}^{(k)})^{a_k^1}(x_{d_1}^{(k)})^{a_k^2}\cdots (x_s^{(k)})^{a_k^s}~.
\end{equation}
We call such monomials as \textit{words}. 
Then a canonical basis of $R_s/I_s$ is given by 
\begin{equation}\label{eq:basis_Rn/I_n}
    \CB_{R_s}:=\{\text{monomial}~~X_1^{\vec{a}_1}X_2^{\vec{a}_2}\cdots X_N^{\vec{a}_N}:g\nmid X_1^{\vec{a}_1}X_2^{\vec{a}_2}\cdots X_N^{\vec{a}_N}\text{ for all }g\in \mr{LM}(G_s)\}~,
\end{equation}
i.e., words modulo constraints. 
Let $(\CB_{R_s})_{\ell,m}\subseteq \CB_{R_s}$ be the subset of monomials of weight $(\ell,m)$. Then the Hilbert series for $R_s$ can be written as 
\begin{equation}
H_{R_s}(q, t)=\sum_{\ell=0}^s \sum_{m \geq 0} |(\CB_{R_s})_{\ell,m}|q^{\ell} t^m~.
\end{equation}
Taking $s \rightarrow \infty$, we get the Hilbert series of  $R_{\infty}$:
\begin{equation}
H_{R_{\infty}}(q, t)=\sum_{\ell, m \geq 0} |(\CB_{R_s})_{\ell,m}| q^{\ell} t^m~.
\end{equation}
In Appendix \ref{app:Hilb_ser_count}, we prove the following formula for the Hilbert series of $R_n$: 
\begin{equation}
\label{eq: hilbert_ser_Rn_gen}
H_{R_s}(q,t) = \mr{Ser}_{q=0,t=0}\frac{\sum_{S\subseteq \mr{LM}(G_s)}(-1)^{|S|}(q,t)^{\mr{wt}^{(2)}(\mr{LCM}(S))}}{\prod_{i=1}^N\prod_{j\geq 0}\left(1-q^{(d_i-1)+j}t\right)} \quad (\mr{mod} \;q^{s+1})~,
\end{equation}
where $\mr{Ser}_{q=0,t=0}(f(q,t))$ denotes the power series of $f(q,t)$ around $q=t=0$, the summation in the numerator is over all subsets of $\mr{LM}(G_s)$, $\mr{LCM}(S)$ is least common multiple of the monomials in $S$, and we have defined
\begin{equation}
(q,t)^{(\ell,m)} := q^{\ell} t^m ~.   
\end{equation}
The full Hilbert series is then given by taking $n\to \infty$:
\begin{equation}
\label{eq: hilbert_ser_Rinfty_gen}
H_{R_\infty}(q,t) = \mr{Ser}_{q=0,t=0}\frac{\sum_{S\subseteq \mr{LM}(G_\infty)}(-1)^{|S|}(q,t)^{\mr{wt}^{(2)}(\mr{LCM}(S))}}{\prod_{i=1}^N\prod_{j\geq 0}\left(1-q^{(d_i-1)+j}t\right)}~,
\end{equation}
where 
\begin{equation}
    \mr{LM}(G_\infty):=\lim _{\leftarrow } \mr{LM}(G_s)~.
\end{equation}
This formula will be primarily used in Section \ref{sec:Hilb_arc_MacD3} for computation of Macdonald index of $(A_1,D_3)$ theory. 
\subsubsection{Example: $(A_1,A_{2n})$ Argyres-Douglas Theory}
In this section, we work out the Macdonald index of $(A_1,A_{2n})$ theory using the Hilbert series of arc space of the Zhu's $C_2$-algebra of the corresponding \tsf{VOA} as a demonstration of Theorem \ref{thm:Hil_ser_mac}. \par 
The Macdonald index for $(A_1,A_{2n})$ theory is given by \cite{Buican:2015tda,Foda:2019guo}
\begin{equation}\label{eq:Mac_A_1A2n}
    \CI_{\mr{M}}^{(A_1,A_{2n})}(q,t)=\sum_{N_1\geq N_2\geq\dots\geq N_{n}\geq 0}t^{N_1+N_2+\dots+N_{n}}\frac{q^{N_1^2+\dots+N_{n}^2}}{(q;q)_{N_1-N_2}\dots (q;q)_{N_{n-1}-N_{n}}(q;q)_{N_{n}}}~.
\end{equation}
The associated \tsf{VOA} of the $(A_1,A_{2n})$ Argyres-Douglas theories is \cite{beem2014infinite,Beem:2017ooy}
\begin{equation}
    \CV_{(A_1,A_{2n})}=\mr{Vir}\left(c=-\frac{2 n(5+6 n)}{3+2 n},0\right)~.
\end{equation}
Note that these are precisely the $(2,2n+3)$ Virasoro \tsf{VOA}s which were proved to be classically free \cite{van2021chiral}. Thus, by Proposition \ref{prop:affine_vir}, we can expect under Song's conjecture that the Macdonald index of $(A_1,A_{2n})$ theory must be given by the Hilbert series of the arc space of its Zhu's $C_2$-algebra. The $C_2$-algebra of $\CV_{(A_1,A_{2n})}$ is given by \cite{Beem:2017ooy} 
\begin{equation}\label{eq:zhuA1A2}
R_{\CV_{(A_1,A_{2n})}}\cong \C[J]/\langle J^{n+1}\rangle~.    
\end{equation}
Here $J$ is the variable corresponding to the stress tensor in the \tsf{VOA}.
The Macdonald index for $(A_1,A_2)$ theory was reproduced from the Hilbert series for the arc space of $R_{\CV_{(A_1,A_{2})}}$ in \cite{Bhargava:2023hsc} (see also \cite{Foda:2019guo} for an alternate derivation), the physical reasoning was an exact map of operators of the theory to the free vector multiplet. Here we reproduce some of the calculations from \cite{Bhargava:2023hsc} but using the algorithm given in Section \ref{sec:comp_Hilb_ser} and show that it reproduces the Macdonald index in line with our general Theorem \ref{thm:Hil_ser_mac}. For $(A_1,A_4)$ theory, we compute the Hilbert series to $O(t^6)$ and show that it matches with the Macdonald index. For $n>1$, we have done numerical checks and confirmed that it reproduces the Macdonald index.  
\\\\
$\underline{\boldsymbol{n=1}.}$ The Hilbert series for the arc space of \eqref{eq:zhuA1A2} for $n=1$ was computed rigorously in  \cite{bruschek2011arc,bai2020quadratic}. Here we present a computation using our formula \eqref{eq: hilbert_ser_Rinfty_gen}.  We start with the formal series 
\begin{equation}
    J(x)=\sum_{m=1}^\infty J_mx^m~.
\end{equation}
We choose the ordering of variables to be
\begin{equation}
J_1>J_2>\dots
\end{equation}
With the degree lexicographic ordering of monomials, Mathematica calculations (checked up $s \leq 10$) give is the leading terms of the Gr\"{o}bner basis to be
\begin{equation}
\mr{LM}\left(G_s^{\left(A_1, A_2\right)}\right)=\bigcup_{m \geq 2} \bigcup_{0 \leq \ell \leq s} I_{\ell, m}^{\left(A_1, A_2\right)}~,
\end{equation}
where $I_{\ell, m}^{\left(A_1, A_2\right)}$ is given by
\begin{equation}
I_{\ell, m}^{\left(A_1, A_2\right)}:=\{J_{m-1} \}\cdot F_{m-1}(\ell-m+1, m-1)~,
\end{equation}
where for sets $A,B$ we have introduced the notation $A\cdot B:=\{ab:a\in A,b\in B\}$, $F_s(\ell, m)$ is the set of all monomials generated by $\left\{J_i: i \geq s\right\}$ of weight $(\ell, m)$.
The set $F_m(\ell, m)$ is thus the set of monomials of weight $(\ell, m)$. We now want to count the number of basis elements of $((R_{\CV_{(A_1,A_2)}})_n/I_n)_{\ell,m}$ obtained from Gr\"{o}bner basis using the formula \eqref{eq:basis_Rn/I_n}. The set $F(\ell,m)$ of all monomials of weight $(\ell,m)$ are given by 
\begin{equation}
    F(\ell,m):=\left\{J_{k_1}^{a_1}\cdots J_{k_s}^{a_s}:s\geq 0,k_1\leq k_2\leq\dots\leq k_s,\sum_{i=1}^sa_ik_i=\ell,\sum_{i=1}^sa_i=m\right\}~.
\end{equation}
From $F(\ell,m)$, we need to remove the monomials
\begin{equation}
\{J_k\}\cdot F_k(\ell-k,m-1)~,\quad k=1,\dots,m-1~,    
\end{equation}
because these monomials are part of the leading monomials of the Gr\"{o}bner basis. This means that that 
\begin{equation}
|((R_{\CV_{(A_1,A_2)}})_n/I_n)_{\ell,m}|=|F_m(\ell,m)|~.    
\end{equation}
The Hilbert series is then given by
\begin{equation}
H_{(R_{\CV_{(A_1,A_2)}})_\infty}(q,t)=\sum_{\ell \geq 0, m \geq 0}\left|F_m(\ell, m)\right| q^\ell t^m~.
\end{equation}
We now claim that 
\begin{equation}
    \sum_{\ell \geq 0, m \geq 0}\left|F_m(\ell, m)\right| q^\ell t^m=\sum_{m \geq 0} \frac{q^{m^2}}{(q)_m} t^m~,
\end{equation}
where 
\begin{equation}
    (q)_m:=(q;q)_m=\prod_{i=1}^m(1-q^i)~.
\end{equation}
To prove this, let us first define 
\begin{equation}
\label{def:Ism}
\mathcal{I}_m^s(q):=\sum_{\ell \geq 0}\left|F_s(\ell, m)\right| q^\ell~.
\end{equation}
Then we claim that
\begin{equation}\label{eq:Ism}
\mathcal{I}_m^s(q)=\frac{q^{sm}}{(q)_m}~.
\end{equation}
The generating function for $|F(\ell,m)|$ is given by 
\begin{equation}
    \sum_{\ell,m\geq 0}|F(\ell,m)|q^\ell t^m=\mr{Ser}_{t=0, q=0} \frac{1}{\prod_{i=1}^{\infty}\left(1-q^it\right)}~.
\end{equation}
It is also easy to see that 
\begin{equation}\label{eq:free_gen_ser_A1A4}
\mr{Ser}_{t=0, q=0} \frac{1}{\prod_{i=s}^{\infty}\left(1-q^it\right)}=\sum_{\ell,m\geq 0}|F_s(\ell,m)|q^\ell t^m=\sum_{m\geq 0}\CI_m^s(q)t^m~.    
\end{equation}
Using the formula 
\begin{equation}
    \frac{1}{(z;q)_\infty}=\sum_{m\geq 0}\frac{z^m}{(q;q)_m}~,
\end{equation}
for $z=tq$, we deduce that 
\begin{equation}\label{eq:I1m}
    \CI^1_{m}(q)=\frac{q^m}{(q)_m}~.
\end{equation}
It is easy to see that 
\begin{equation}\label{eq:Is0}
    \CI^s_0=1~.
\end{equation}
Next note that 
\begin{equation}
 \left(1-q^s\right) \sum_{m \geq 0} \mathcal{I}_m^s(q) t^m=\sum_{m \geq 0} \mathcal{I}_m^{s+1}(q) t^m ~.  
\end{equation}
This gives us the recursion relation
\begin{equation}
\mathcal{I}_m^{s+1}(q)=\mathcal{I}_m^s(q)-q^s \mathcal{I}_{m-1}^s(q)~,\quad m>0~.
\end{equation}
The solution to this recursion relation subject to initial conditions \eqref{eq:I1m} and \eqref{eq:Is0} is given by \eqref{eq:Ism}. Indeed
\begin{equation}
\mathcal{I}_m^{s+1}(q)=q^{s m}\left(\frac{1}{(q)_m}-\frac{1}{(q)_{m-1}}\right)=\frac{q^{(s+1)_m}}{(q)_m}~.
\end{equation}
This proves that 
\begin{equation}\label{eq:Hil_A1A2}
\sum_{\ell \geq 0, m \geq 0}\left|F_m(\ell, m)\right| q^\ell t^m=\sum_{m\geq 0}\CI_m^m(q)t^m=\sum_{m \geq 0} \frac{q^{m^2}}{(q)_m} t^m~.    
\end{equation}
Thus the Hilbert series \eqref{eq:Hil_A1A2} matches with the Macdonald index \eqref{eq:Mac_A_1A2n} of $(A_1,A_2)$ theory. 
\\\\
\underline{$\boldsymbol{n=2}$.} For $(A_1,A_4)$ theory, using numerical checks  we obtain the following expression for the Hilbert series of the arc space of $R_{\CV_{(A_1,A_4)}}$:
\begin{equation}
H_{(R_{\CV_{(A_1,A_4)}})_\infty}(q,t)=\sum_{m \geq 0}\CI^{(A_1,A_4)}_m(q) t^m~,   
\end{equation}
where 
\begin{equation}
    \mathcal{I}_m^{\left(A_1, A_4\right)}= \begin{cases}\mathcal{I}_m^{m / 2}+\sum_{i=1}^{(m-2) / 2} I_m^{2 i, m / 2-i}, & m \text { is even }, \\ \mathcal{I}_m^{(m+1) / 2}+\sum_{i=1}^{(m-1) / 2} I_m^{2 i-1,(m+1) / 2-i}, & m \text { is odd } ,\end{cases}
\end{equation}
and $I_m^{c,s}(q)$ can be obtained from the recursion relation:
\begin{equation}
I_m^{c, s}= \begin{cases}q^s \mathcal{I}_{m-1}^{s+1}, & c=1, \\ q^s \mathcal{I}_{m-1}^{s+2}, & c=2,3, \\ q^s \mathcal{I}_{m-1}^{s+c+1}+q^s \sum_{i=1}^{c / 2-1} I_{m-1}^{2 i-1, s+c / 2+1-i}, & c>3 \text { and even }, \\ q^s \mathcal{I}_{m-1}^{s+(c+1) / 2}+q^s \sum_{i=1}^{(c-3) / 2} I_{m-1}^{2 i, s+(c+1) / 2-i}, & c>3 \text { and odd~. }\end{cases}    
\end{equation}
See Appendix \ref{app:Hilb_A1A4} for details of the calculation. 
This formula matches the Macdonald index \eqref{eq:Mac_A_1A2n} for the theory $(A_1,A_4)$ up to $O(t^{20})$. More concretely, we can change variables to $N_1+N_2=m,N_2=k$ and write the Macdonald index as
\begin{equation}
    \CI_{\mr{M}}^{(A_1,A_4)}(q,t)=\sum_{m\geq 0}t^m\left(\sum_{k=0}^{\lfloor\frac{m}{2}\rfloor}\frac{q^{(m-k)^2+k^2}}{(q;q)_{m-2k}(q;q)_k}\right)~.
\end{equation}
Then, our formula for Macdonald index using the Hilbert series predicts the following identity:
\begin{equation}
\mathcal{I}_m^{\left(A_1, A_4\right)}(q)=\sum_{k=0}^{\lfloor\frac{m}{2}\rfloor}\frac{q^{(m-k)^2+k^2}}{(q;q)_{m-2k}(q;q)_k}~.    
\end{equation}
This was checked numerically for large values of $m\leq 20$. 
We conjecture that this is true for all $m$. 
\section{Macdonald Index For $(A_1,D_{3})$ Theory}\label{sec:A1D3_Mac_conj}
In this section, we conjecture a very simple expression for the Macdonald index of $(A_1,D_3)$ and provide powerful checks of our formula. Our starting point is the observation (whose existence we conjecture is based upon the Higgs branch \tsf{RG} flow explained in Section \ref{sec:RG_flow}) that the Macdonald index for the $(A_1,D_3)$ theory must be related to the Macdonald index for the (projected) free hypermultiplet (see Remark \ref{rem:proj_free_hyper}  for a remark about the term ``projected''.). 
\subsection{Relation To Free Hypermultiplet Index}\label{sec:D3_free_hyper_rel}
We begin by reviewing the Macdonald index for (projected) free hypermultiplet.
\subsubsection{Free Hypermultiplet}
The 4d $\mathcal{N}=2$ hypermultiplet consists of two 4d $\mathcal{N}=1$ chiral (and anti-chiral) multiplets $Q$ $(q,q^\dagger,\psi,\bar\psi)$ and 
$\widetilde{Q}$ $(\tilde{q},\tilde{q}^\dagger, \chi,\bar\chi)$.
$(q,\tilde{q}^\dagger)$ and 
$(- \tilde{q},q^\dagger)$
transform as doublets of $\mr{SU(2)_R}$. We choose the $\mr{U(1)}_r$ charge of $q$ and 
$\widetilde{q}$ to be zero. 
There is a global $\mr{U(1)_B}$ baryonic symmetry that acts differently on each half-hypermultiplet, 
\begin{equation}
 Q \to e^{i\varphi } Q,\quad \widetilde{Q}\to e^{-i\varphi} \widetilde Q~.   
\end{equation}
The letters contributing the index of the free 
$\mathcal{N}=2$ (half) hypermultiplet  are shown in Table \ref{tab:half_hyper_charge}.
\begin{table}[H]
    \centering
    \begin{tabular}{|c|c|c|c|c|c|c|}
    \hline
        \text{Letters} & 
        $E$ &
        $j_1$& $j_2$ & $R$ & $r$ &
        $\mathcal{I}(p,q,t)$\\
        \hline
         $q$& $1$ & $0$ & $0$ & $\frac{1}{2}$ &  $0$ & $t^{1/2}$\\
         $\bar\psi_{\dot +}$& 
         $\frac{3}{2}$
         & $0$ & $\frac{1}{2}$ & $0$ &$-\frac{1}{2}$  &
         $-pqt^{-1/2} $
         \\
         $\partial_{\pm \dplus}$
         & $1$ & $\pm \frac{1}{2}$ &$\frac{1}{2}$  & $0$ & $0$ & $p,q$\\
         \hline
    \end{tabular}
    \caption{Charge assignments of various operators in free (half) hypermultiplet.}
    \label{tab:half_hyper_charge}
\end{table}
\noindent We have the following single particle indices for free hypermultiplet
\begin{equation}
    \mathcal{I}^{\mr{HM}}_{s . p .}=t^{\frac{1}{2}} \frac{1-\frac{p q}{t}}{(1-p)(1-q)}\left(a+a^{-1}\right) \chi_{\Lambda}(\boldsymbol{x})~,
\end{equation}
where $a$ is the $U(1)_B$ fugacity and 
$\chi_\Lambda(\boldsymbol{x})$ is the character of the representation of some global symmetry.
In this paper, we will ignore flavor fugacities and consequently take $\boldsymbol{x} = 1,a=1$. We will take $a\to 1$ at the end. 
We get the superconformal index of a free hypermultiplet 
\begin{equation}
\begin{aligned}
    \mathcal{I}_{\mr{SC}}^{\mr{HM}}(p,q,t;a ) &= 
    \mathsf{PE}\left[
    t^{\frac{1}{2}} \frac{1-\frac{p q}{t}}{(1-p)(1-q)}\left(a+a^{-1}\right)
    \right]\\
    & = 
    \Gamma(t^{\frac{1}{2}} a; p,q)
    \Gamma(t^{\frac{1}{2}} a^{-1}; p,q)~,
\end{aligned}
\end{equation}
where the plethystic exponential is given by 
\begin{equation}
    \mr{PE}[f(x, y, \cdots)]:=\exp \left[\sum_{\ell=1}^{\infty} \frac{1}{\ell} f\left(x^{\ell}, y^{\ell}, \cdots\right)\right]~,
\end{equation}
and the elliptic gamma function is defined by 
\begin{equation}
    \Gamma(z ; p, q):=\mathsf{PE}\left[\frac{z-\frac{p q}{z}}{(1-p)(1-q)}\right]=\prod_{i, j=0}^{\infty} \frac{1-p^{i+1} q^{j+1} z^{-1}}{1-p^i q^j z}~.
\end{equation}
The Macdonald index of a free hypermultiplet is obtained by taking $p\to 0$ in the superconformal index: 
\begin{equation}
    \mathcal{I}_{\mr{SC}}^{\mr{HM}}( p\to 0, q,t; a) = 
    \prod_{r=0}^\infty 
    \frac{1}{\left(1-q^r t^{\frac{1}{2}} a\right)}
    \frac{1}{\left(1-q^r t^{\frac{1}{2}} a^{-1}\right)}.
\end{equation}
Taking $a \to 1$, we can draw a coefficient diagram for Macdonald index of a free hypermultiplet as in Table \ref{tab:coeff_mac_FH}.
\begin{table}[h]
    \centering
    \begin{tabular}{|c|ccccccccccc|}
    \hline
   Variables & $t^0$ & $t^{1/2}$ 
    & $t^{1}$ & $t^{3/2}$
    & $t^{2}$ & $t^{5/2}$
    & $t^{3}$ & $t^{7/2}$
    & $t^{4}$ & $t^{9/2}$
    & $t^{5}$ \\
    \hline
 $q^{0}$ &1 & 2 & 3 & 4 & 5 & 6 & 7 & 8 & 9 & 10 & 11 \\
 $q^{1}$ &0 & 2 & 4 & 6 & 8 & 10 & 12 & 14 & 16 & 18 & 20 \\
 $q^{2}$ &0 & 2 & 7 & 12 & 17 & 22 & 27 & 32 & 37 & 42 & 47 \\
 $q^{3}$ &0 & 2 & 8 & 18 & 28 & 38 & 48 & 58 & 68 & 78 & 88 \\
 $q^{4}$ &0 & 2 & 11 & 26 & 46 & 66 & 86 & 106 & 126 & 146 & 166 \\
 $q^{5}$ &0 & 2 & 12 & 34 & 64 & 100 & 136 & 172 & 208 & 244 & 280 \\
 $q^{6}$ &0 & 2 & 15 & 46 & 94 & 152 & 217 & 282 & 347 & 412 & 477 \\
 $q^{7}$ &0 & 2 & 16 & 56 & 124 & 214 & 316 & 426 & 536 & 646 & 756 \\
 $q^{8}$ &0 & 2 & 19 & 70 & 167 & 302 & 464 & 640 & 825 & 1010 & 1195 \\
 $q^{9}$ &0 & 2 & 20 & 84 & 212 & 406 & 648 & 922 & 1212 & 1512 & 1812 \\
 $q^{10}$ &0 & 2 & 23 & 100 & 271 & 542 & 899 & 1314 & 1766 & 2236 & 2717 \\\hline
\end{tabular}
\caption{Table of first few coefficients of the Macdonald index for a free hypermultiplet.}
\label{tab:coeff_mac_FH}
\end{table}
The unflavored index can also be written as an infinite sum
\begin{equation}
    \CI_{\mr{M}}^{\mr{HM}}(q,t)=\sum_{2 r=0}^{\infty} t^r \sum_{k=0}^{2 r} \frac{1}{(q ; q)_k(q ; q)_{2 r-k}}~.
\end{equation}
The projected free hypermultiplet is obtained from the free hypermultiplet by projecting out the operators with half-integral $\mr{U(1)}_r$-charge. 
\begin{remark}\label{rem:proj_free_hyper}
The free hypermultiplet has a global $\mathds{Z}_2\subset \mr{U(1)}_r$ symmetry. The projected free hypermultiplet that we will use throughout this work is obtained by projecting onto the local operators invariant under this symmetry. However, we note that we only work with the Schur ring of these theories, so for our purposes, this is not a projection (or $\mathds{Z}_2$-gauging) at the level of the full theory, but merely a convenient way to set up the operator map between the Schur rings of the \tsf{UV} and \tsf{IR} theories. We do not set up enough structure to treat the projected free hypermultiplet as a gauged version of the free hypermultiplet.
\end{remark}
\par We need to remove the columns with $t^{n+\frac{1}{2}},~n\in\Z_{\geq 0}$ from Table \ref{tab:coeff_mac_FH} to obtain the coefficient list of projected free hypermultiplet, we reproduce it here for later reference in Table \ref{tab:coeff_mac_proj_FH}.
\begin{table}[h]
    \centering
    \begin{tabular}{|c|cccccc|}
    \hline
   Variables & $t^0$ & $t^{1}$ & $t^{2}$ & $t^{3}$ & $t^{4}$ & $t^{5}$ \\
    \hline
    $q^{0}$  & 1 & 3 & 5 & 7 & 9 & 11 \\
    $q^{1}$  & 0 & 4 & 8 & 12 & 16 & 20 \\
    $q^{2}$  & 0 & 7 & 17 & 27 & 37 & 47 \\
    $q^{3}$  & 0 & 8 & 28 & 48 & 68 & 88 \\
    $q^{4}$  & 0 & 11 & 46 & 86 & 126 & 166 \\
    $q^{5}$  & 0 & 12 & 64 & 136 & 208 & 280 \\
    $q^{6}$  & 0 & 15 & 94 & 217 & 347 & 477 \\
    $q^{7}$  & 0 & 16 & 124 & 316 & 536 & 756 \\
    $q^{8}$  & 0 & 19 & 167 & 464 & 825 & 1195 \\
    $q^{9}$  & 0 & 20 & 212 & 648 & 1212 & 1812 \\
    $q^{10}$ & 0 & 23 & 271 & 899 & 1766 & 2717 \\
    \hline
    \end{tabular}
    \caption{Table of selected coefficients of the Macdonald index for a projected free hypermultiplet.}
    \label{tab:coeff_mac_proj_FH}
\end{table}
The resulting Macdonald index can be written as
\begin{equation}
    \CI_{\mr{M}}^{\mr{HM}/\mathds{Z}_2}(q,t)=\sum_{\ell=0}^{\infty}\frac{t^\ell}{(q;q)_{2\ell}}\sum_{k=0}^{2\ell}{2\ell\choose k}_q~,
\end{equation}
where 
\begin{equation}
\binom{2m}{k}_q=\frac{(q;q)_{2m}}{(q;q)_{2m-k}(q;q)_k}~,    
\end{equation}
is the $q$-binomial coefficient.

\subsubsection{A Closed Formula For The Macdonald Index}
Let us start with the closed form expression for the Macdonald index for $(A_1,D_3)$ theory obtained from Class $\CS$ construction. The formula for $(A_1,D_{2N-3})$ is recorded in \eqref{eq:Mac_class_S-1} -- \eqref{eq:Mac_class_S-7}. Simplifying it for $N=3$, we obtain Table \ref{tab:coeff_mac_A1D3} of coefficients of the Macdonald index for $(A_1,D_3)$ theory. 
\begin{table}[]
    \centering
    \begin{tabular}{|c|ccccccccc|} 
    \hline
Variables & $t^0$ & $t^1$ & $t^2$ & $t^3$ & $t^4$ & $t^5$ & $t^6$ & $t^7$ & $t^8$ \\\hline
$q^0$ & 1 & 3 & 5 & 7 & 9 & 11 & 13 & 15 & 17 \\
$q^1$ & 0 & 4 & 8 & 12 & 16 & 20 & 24 & 28 & 32 \\
$q^2$ & 0 & 4 & 17 & 27 & 37 & 47 & 57 & 67 & 77 \\
$q^3$ & 0 & 4 & 23 & 48 & 68 & 88 & 108 & 128 & 148 \\
$q^4$ & 0 & 4 & 33 & 79 & 126 & 166 & 206 & 246 & 286 \\
$q^5$ & 0 & 4 & 39 & 117 & 199 & 280 & 352 & 424 & 496 \\
$q^6$ & 0 & 4 & 49 & 171 & 322 & 466 & 607 & 737 & 867
\\
\hline
\end{tabular}
    \caption{Table of selected coefficients of the Macdonald index for $(A_1,D_3)$ theory.}
    \label{tab:coeff_mac_A1D3}
\end{table}

The crucial observation is that the coefficients in Table \ref{tab:coeff_mac_A1D3} and Table \ref{tab:coeff_mac_proj_FH} above the diagonal are in perfect agreement. This suggests the  naive guess that each column in the Macdonald index for $(A_1,D_3)$ may be expressible in terms of of the corresponding column in $\CI_{\mr{M}}^{\mr{HM}/\mathds{Z}_2}(q,t)$ corrected by some other factor. By checking the first few columns, we find that such a reformulation is indeed possible, and we conjecture a general form of the correction factor that we then check for higher columns to find complete agreement. 

The contribution of the second column to the Macdonald index is 
\begin{equation}
t\left(3+4 q+4 q^2+4 q^3+\cdots\right)~,    
\end{equation}
which can be written as 
\begin{equation}
\begin{split}
    &t\left(1-q^2\right)\left(3+4 q+7 q^2+8 q^3+11 q^4+12 q^5+\cdots\right)\\&=t\left(1-q^2\right)\left(\frac{1}{(q ; q)_0(q ; q)_2}+\frac{1}{(q ; q)_1(q ; q)_1}+\frac{1}{(q ; q)_2(q ; q)_0}\right)~.
\end{split}
\end{equation}
This has the desired form as it involves the contribution of only the second column of the projected free hypermultiplet and no other columns in the projected free hypermultiplet. At order $t^2$, we can the write the contribution as 
\begin{equation}
\begin{split}
    &t^2\left(5+8 q+17 q^2+23 q^3+33 q^4+39 q^5+49 q^6+\cdots\right)\\=~&t^2\left(1-q^3\right)\left(1-q^4\right)\left(5+8q+17q^2+28q^3+46q^4+64q^5+\cdots\right)\\    =~&t^2\left(1-q^3\right)\left(1-q^4\right)\left(\frac{1}{(q ; q)_0(q ; q)_4}+\frac{1}{(q ; q)_1(q ; q)_3}+\frac{1}{(q ; q)_2(q ; q)_2}+\frac{1}{(q ; q)_3(q ; q)_1}+\frac{1}{(q ; q)_4(q ; q)_0}\right)~,
\end{split}
\end{equation}
which again has the desired form. We further notice that 
\begin{equation}
\begin{aligned}
1&=(q ; q)_0 \\
\left(1-q^2\right)&=\left(q^2 ; q\right)_1 \\
\left(1-q^3\right)\left(1-q^4\right)&=\left(q^3 ; q\right)_2~.
\end{aligned}
\end{equation}
This pattern motivates our formula for the Macdonald index:
\begin{equation}
    \CI_{\mr{M}}^{(A_1,D_3)}(q, t)=\sum_{r=0}^{\infty} t^r\left(\sum_{k=0}^{2 r} \frac{1}{(q ; q)_k(q ; q)_{2 r-k}}\right)\left(q^{r+1} ; q\right)_n~.
\end{equation}
Using the formula
\begin{equation}
(a ; q)_{r+k}=(a ; q)_r\left(a q^r ; q\right)_k~,    
\end{equation}
we obtain our final proposal for the Macdonald index:
\begin{conj}
The Macdonald index for the $(A_1,D_3)$ Argyres-Douglas theory is given by 
\begin{equation}\label{eq:Mac_ind_A1D3}
\CI_{\mr{M}}^{(A_1,D_3)}=\sum_{m=0}^\infty 
    \frac{t^m}{(q;q)_m} \sum_{k=0}^{2m}\binom{2m}{k}_q~.    
\end{equation}
\end{conj}
This formula reproduces Table \ref{tab:coeff_mac_A1D3}. 
\subsection{Checks Of Our Formula}
In this section, we show that the formula \eqref{eq:Mac_ind_A1D3} satisfies several expected properties. 
\subsubsection{Matching The Schur Limit}
The first check is the matching of the Schur index obtained by taking the Schur limit. The Schur index of $(A_1,D_{2n+1})$ is given by (for details, we refer the reader to \cite{Cordova:2015nma, Song:2016yfd}, and to \cite{Buican:2015ina, Buican:2015hsa} for the specific case of the $(A_1,D_3)$ theory)
\begin{equation}
    \CI_{\mathsf{S}}^{(A_1,D_{2n+1})}(q,a)=\mr{PE}\left[\frac{q-q^{2 n+1}}{(1-q)\left(1-q^{2 n+1}\right)} \chi_{\mr{adj}(a)}\right]~,
\end{equation}
where 
\begin{equation}
    \chi_{\mathsf{adj}}(a)=a^2+\frac{1}{a^2}+1~, 
\end{equation}
is the character of the adjoint representation of the flavor symmetry $\mf{su}(2)$. We claim that 
\begin{equation}\label{eq:schur_D2n+1_pochh_form}
\CI_{\mathsf{S}}^{(A_1,D_{2n+1})}(q,1)=\left(\frac{\left(q^{2n+1} ; q^{2n+1}\right)_{\infty}}{(q ; q)_{\infty}}\right)^3~.    
\end{equation}
We have 
\begin{equation}
 \CI_{\mathsf{S}}^{(A_1,D_{2n+1})}(q,1)=\mr{PE}\left[3\frac{q-q^{2 n+1}}{(1-q)\left(1-q^{2 n+1}\right)} \right]   
\end{equation}
Notice that
\begin{equation}
\begin{split}
\frac{3(q-q^{2 n+1})}{(1-q)(1-q^{2 n+1})}&=\left(\frac{3 q}{1-q}-\frac{3 q^{2 n+1}}{1-q^{2 n+1}}\right) ~.
\end{split}
\end{equation}
Using the formula 
\begin{equation}
    \mr{PE}\left[\frac{kq}{1-q}\right]=\frac{1}{(q;q)^k_\infty} ~,\quad \PE[f-g]=\frac{\PE[f]}{\PE[g]}~,
\end{equation}
we get
\begin{equation}
    \CI_{\mathsf{S}}^{(A_1,D_{2n+1})}(q,1)=\mr{PE}\left[3\frac{q-q^{2 n+1}}{(1-q)\left(1-q^{2 n+1}\right)} \right]  =\left(\frac{\left(q^{2n+1} ; q^{2n+1}\right)_{\infty}}{(q ; q)_{\infty}}\right)^3~.
\end{equation}
Taking $t\to q$ in \eqref{eq:Mac_ind_A1D3}, we get 
\begin{equation}
    \CI_{\mr{S}}^{(A_1,D_3)}(q)=\sum_{m=0}^\infty 
    \frac{q^m}{(q;q)_m} \sum_{k=0}^{2m}\binom{2m}{k}_q~.    
\end{equation}
Thus, proving that our formula \eqref{eq:Mac_ind_A1D3} satisfies the correct Schur limit amounts to proving the following product-sum identity:
\begin{equation}\label{eq:new_idn1}
\sum_{m=0}^\infty 
    \frac{q^m}{(q;q)_m}\sum_{k=0}^{2m}\binom{2m}{k}_q=\left(\frac{\left(q^3 ; q^3\right)_{\infty}}{(q ; q)_{\infty}}\right)^3~.    
\end{equation}
We now prove a more general identity which implies \eqref{eq:new_idn1}.
\begin{thm}\label{thm:main_id}
We have the identity

\begin{equation}\label{eq:main_id}
\frac{\left(z ; q^3\right)_{\infty}\left(q^3 / z ; q^3\right)_{\infty}\left(q^3 ; q^3\right)_{\infty}}{(z ; q)_{\infty}(q / z ; q)_{\infty}(q ; q)_{\infty}}=\sum_{r=0}^{\infty} \frac{q^r}{(q ; q)_r} \sum_{k=-r}^r{
2 r \choose
r+k}_q z^k~.
\end{equation}    
\end{thm}
Before proving this theorem, we note that plugging $z=1$ in \eqref{eq:main_id} gives \eqref{eq:new_idn1}. On the \tsf{LHS} of \eqref{eq:main_id}, we need to use 
\begin{equation}
    \lim_{z\to 1}\frac{\left(z ; q^3\right)_{\infty}}{\left(z ; q\right)_{\infty}}=\frac{\left(q^3 ; q^3\right)_{\infty}}{\left(q ; q\right)_{\infty}}~.
\end{equation}
\begin{proof}
We start with the following form of the  Jacobi's triple product identity
\begin{equation}
\sum_{r \in \mathds{Z}} q^{\frac{r(r+1)}{2}} z^r=(q ; q)_{\infty}\left(-1/z ; q\right)_{\infty}(-z q ; q)_{\infty}
\end{equation}
which can be obtained from \cite[Theorem 2.8, Page 21]{andrews1998theory} upon using the definition of $q$-Pochhammer symbol. Substituting $z \mapsto-z / q$, we obtain 
\begin{align}
\sum_{r \in \mathds{Z}}(-1)^r q^{r\left(\frac{r+1}{2}-1\right)} z^r&=(q ; q)_{\infty},(q / z ; q)_{\infty}(z ; q)_{\infty}~,
\\\implies
\sum_{r \in \mathds{Z}}(-1)^r q^{\binom{r}{2}} z^r&=(z ; q)_{\infty}(q / z ; q)_{\infty}(q ; q)_{\infty} .
\end{align}
So we may equivalently prove
\begin{equation}\label{eq:intermediate_main_id}
\begin{aligned}
& \sum_{N=-\infty}^{\infty}(-1)^N q^{3\binom{N}{2}} z^N=\sum_{m=-\infty}^{\infty}(-1)^m q^{\binom{m}{2}} z^m \sum_{r=0}^{\infty} \frac{q^r}{(q ; q)_r} \sum_{k=-r}^r{
2 n \choose
r+k}_q z^k .
\end{aligned}
\end{equation}
To prove this, we will show that the coefficient of $z^N$ of both sides is equal for $N \geq 0$. Then for $N < 0$, we will use the $z \mapsto z^{-1}$ symmetry on both sides of \eqref{eq:intermediate_main_id}.
The coefficient of $z^N$ on LHS of \eqref{eq:intermediate_main_id} is
\begin{equation}\label{eq:LHS_zN_coeff}
(-1)^N q^{3{N\choose 2}}~.
\end{equation}
The coefficient of $z^N$ on RHS of \eqref{eq:intermediate_main_id} is
\begin{equation}\label{eq:RHS_zN_coeff}
\begin{aligned}
& \sum_{r=0}^{\infty} \frac{q^r}{(q ; q)_r} \sum_{k=-r}^r{
2 n \choose
r+k}_q(-1)^{N-k} q^{\binom{N-k}{2}} \\
& \quad=(-1)^N q^{\binom{N}{2}} \sum_{r=0}^{\infty} \frac{q^r}{(q ; q)_r} \sum_{k=-r}^r{
2 r \choose
r+k}_q(-1)^k q^{\binom{k}{2}+k(1-N)} \\
& \quad=(-1)^N q^{\binom{N}{2}} \sum_{r=0}^{\infty} \frac{q^r}{(q ; q)_r}\left(q^{1-N} ; q\right)_r\left(q^N ; q\right)_r~,
\end{aligned}
\end{equation}
where we used
\begin{equation}
\begin{aligned}
\binom{N}{2}+\binom{k}{2}+k(1-N) & =\frac{N(N-1)}{2}+\frac{k(k-1)}{2}+k(1-N) \\
& =\frac{N(N-1)}{2}+k \frac{(k-1+2-2 N)}{2} \\
& =\frac{N^2-N+k^2+k-2 N k}{2} \\
& =\frac{(N-k)(N-k-1)}{2} \\
& =\binom{N-k}{2}
\end{aligned}
\end{equation}
and \cite[Page 49, Example 1]{andrews1998theory}
\begin{equation}\label{eq:identity_qx}
\sum_{k=-r}^r q^{k^2} x^k{
2 r \choose
r+k}_{q^2}=\left(-x^{-1} q ; q^2\right)_r\left(-x q ; q^2\right)_r~,
\end{equation}
with $q \rightarrow q^{1 / 2}, x=-q^{1-N-\frac{1}{2}}$ so that LHS of \eqref{eq:identity_qx} becomes
\begin{equation}
\sum_{k=-r}^r q^{\frac{k^2}{2}+k\left(1-N-\frac{1}{2}\right)}(-1)^k{
2 r \choose
r+k}_q=\sum_{k=-r}^r q^{\binom{k}{2}+k(1-N)}(-1)^k{
2 r \choose
r+k}_q~,
\end{equation}
and the RHS of \eqref{eq:identity_qx} becomes
\begin{equation}
\left(q^{N-1+\frac{1}{2}+\frac{1}{2}} ; q\right)_r\left(q^{1-N-\frac{1}{2}+\frac{1}{2}} ; q\right)_r=\left(q^N ; q\right)_r\left(q^{1-N} ; q\right)_r~.
\end{equation}
Now for $N=0$, all terms in the sum in the last line of \eqref{eq:RHS_zN_coeff}, except $r=0$, vanish since $(1 ; q)_r=0$. The $r=0$ term is 1 which is same as \eqref{eq:LHS_zN_coeff} for $N=0$.
For $N>0,1-N \leq 0$ and hence we can use the identity \cite[Appendix II, Eq. (II.6)]{gasper2004basic}
\begin{equation}\label{eq:phi_identity}
{ }_2 \phi_1\left(a, q^{-r} ; c ; q, q\right)=\frac{(c / a ; q)_r}{(c ; q)_r} a^r~,\quad r\geq 0~,
\end{equation}
where \cite[Chapter 1, Eq.(1.2.22)]{gasper2004basic}
\begin{equation}
{ }_2 \phi_1\left(a_1, a_2 ; c ; q, z\right)=\sum_{r=0}^{\infty} \frac{\left(a_1 ; q\right)_r\left(a_2 ; q\right)_r}{(q ; q)_r(c ; q)_r} z^r~.
\end{equation}
Using \eqref{eq:phi_identity} for $a=q^N, c=0,z=q$, we get,
\begin{equation}\label{eq:S_1_sum}
{ }_2 \phi_1\left(q^N, q^{1-N} ; 0 ; q , q\right)=\sum_{r=0}^{\infty} \frac{\left(q^N ; q\right)_r\left(q^{1-N} ; q\right)_r}{(q ; q)_r} q^r=q^{N(N-1)}=q^{N^2-N}~.
\end{equation}
Thus the last line of \eqref{eq:RHS_zN_coeff} evaluates to
\begin{align}
(-1)^N q^{\frac{N(N-1)}{2}} q^{N^2-N}=(-1)^N q^{3\binom{N}{2}}~,
\end{align}
which is same as \eqref{eq:LHS_zN_coeff}.
\end{proof}
\paragraph{An interesting identity.} While proving \eqref{eq:new_idn1}, we noticed that we can prove another interesting $q$-series identity which might be of independent interest.
\begin{thm}
The following identity holds:
\begin{equation}
\left(\frac{\left(q^{2} ; q^{2}\right)_{\infty}}{(q ; q)_{\infty}}\right)^3= \frac{1}{(q;q)_\infty(q^2;q^2)_{\infty}}\sum_{r=0}^\infty(-1)^rq^{r^2+r}\left(\frac{1+q^{2r+1}}{1-q^{2r+1}}\right)~.    
\end{equation}
\begin{proof}
Let us start by noting that 
\begin{equation}
    \frac{(z,q^2/z,q^2;q^2)_\infty}{(z,q/z,q;q)_\infty}=\frac{1}{(zq,q/z,q;q^2)_\infty}~,
\end{equation}
where 
\begin{equation}
    (a_1,a_2,\dots,a_r;q)_\infty:=(a_1;q)_\infty(a_2;q)_\infty\cdots (a_r;q)_\infty~.
\end{equation}
We used the fact that 
\begin{equation}
    \frac{(z;q^2)_\infty}{(z;q)_\infty}=\frac{1}{(zq;q^2)_\infty},\quad \frac{(q^2/z;q^2)_\infty}{(q/z;q)_\infty}=\frac{1}{(q/z;q^2)_\infty}~,\quad \frac{(q^2;q^2)_\infty}{(q;q)_\infty}=\frac{1}{(q;q^2)_\infty}~,
\end{equation}
which follows from the definition. Next, we have 
\begin{equation}
\frac{1}{(zq,q/z,q;q^2)_\infty}=\frac{1}{(q;q^2)_\infty}\sum_{m,r=0}^\infty \frac{z^{m-r}q^{m+r}}{(q^2;q^2)_m(q^2;q^2)_r}~,    
\end{equation}
where we used 
\begin{equation}
    \frac{1}{(z;q)_\infty}=\sum_{m=0}^\infty\frac{z^m}{(q;q)_\infty}~.
\end{equation}
Changing variable to $m-r=k$, we get 
\begin{equation}
\begin{split}
\frac{1}{(zq,q/z,q;q^2)_\infty}&=\frac{1}{(q;q^2)_\infty}\sum_{k\in\Z}z^k\sum_{r=0}^\infty \frac{q^{2r+|k|}}{(q^2;q^2)_{r+|k|}(q^2;q^2)_r}
\\&=\frac{1}{(q;q^2)_\infty}\sum_{k\in\Z}z^k\frac{q^{|k|}}{(q^2;q^2)_{|k|}}{ }_2 \phi_1\left[\begin{array}{c}0,0\\ q^{2|k|+2}\end{array} ; q^2, q^2\right]~,
\end{split}
\end{equation}
where we used the definition in \eqref{eq:phirs_def}. Using Heine's transform \eqref{eq:Heine_exp}, we get 
\begin{equation}
\begin{split}
\frac{1}{(zq,q/z,q;q^2)_\infty}&=\frac{1}{(q;q^2)_\infty}\sum_{k\in\Z}z^k\frac{q^{|k|}}{(q^2;q^2)^2_{\infty}}\sum_{r=0}^\infty(-1)^rq^{r^2+r+2|k|r}
\\&=\frac{1}{(q;q)_\infty(q^2;q^2)_{\infty}}\sum_{k\in\Z}z^k\sum_{r=0}^\infty(-1)^rq^{r^2+r+|k|(2r+1)}
\\&=\frac{1}{(q;q)_\infty(q^2;q^2)_{\infty}}\sum_{r=0}^\infty(-1)^rq^{r^2+r}\sum_{k\in\Z}z^kq^{|k|(2r+1)}
\\&=\frac{1}{(q;q)_\infty(q^2;q^2)_{\infty}}\sum_{r=0}^\infty(-1)^rq^{r^2+r}\left(\sum_{k=0}^\infty\left((zq^{(2r+1)})^k+(z^{-1}q^{(2r+1)})^k\right)-1\right)
\\&=\frac{1}{(q;q)_\infty(q^2;q^2)_{\infty}}\sum_{r=0}^\infty(-1)^rq^{r^2+r}\left(\frac{1}{1-zq^{2r+1}}+\frac{1}{1-z^{-1}q^{2r+1}}-1\right)
\\&=\frac{1}{(q;q)_\infty(q^2;q^2)_{\infty}}\sum_{r=0}^\infty(-1)^rq^{r^2+r}\frac{1-q^{4r+2}}{(1-zq^{2r+1})(1-z^{-1}q^{2r+1})}~.
\end{split}    
\end{equation}
We then have 
\begin{equation}
\begin{split}
    \left(\frac{\left(q^{2} ; q^{2}\right)_{\infty}}{(q ; q)_{\infty}}\right)^3&=\lim_{z\to 1}\frac{(z,q^2/z,q^2;q^2)_\infty}{(z,q/z,q;q)_\infty}=\lim_{z\to 1}\frac{1}{(zq,q/z,q;q^2)_\infty}
    \\&=\frac{1}{(q;q)_\infty(q^2;q^2)_{\infty}}\sum_{r=0}^\infty(-1)^rq^{r^2+r}\frac{1-q^{4r+2}}{(1-q^{2r+1})^2}
    \\&=\frac{1}{(q;q)_\infty(q^2;q^2)_{\infty}}\sum_{r=0}^\infty(-1)^rq^{r^2+r}\left(\frac{1+q^{2r+1}}{1-q^{2r+1}}\right)~.
\end{split}    
\end{equation}
\end{proof}
\end{thm}
\subsubsection{Macdonald Index From Hilbert Series Of Arc Space}\label{sec:Hilb_arc_MacD3}
In this section, we check our formula for the Macdonald index by computing the first three terms in the $t$-expansion of the Hilbert series of the arc space of the Zhu's $C_2$-algebra of the chiral algebra associated to the $(A_1,D_3)$ Argyres-Douglas theory. This supports our formula following the discussion in Section \ref{sec:Mac_Hil_ser_arc}. 

The associated chiral algebra of $(A_1,D_{2n+1})$ theory is \cite{beem2014infinite,Beem:2017ooy} 
\begin{equation}
    \CV_{(A_1,D_{2n+1})}=\mf{su}(2)_{-\frac{4n}{2n+1}}~.
\end{equation}
It is an affine \tsf{VOA} at admissible level $k=-2+2/p$ with $p=2n+1$. Such affine \tsf{VOA}s are believed to be classically free and hence by Proposition \ref{prop:affine_vir} we can expect that the Hilbert series of the arc space of their Zhu's $C_2$-algebra should be equal to their refined character and hence to the Macdonald index of the corresponding theory by Song's conjecture. 
\par
To get the $C_2$-algebra of $\mf{su}(2)_{-\frac{4n}{2n+1}}$, we note that this \tsf{VOA} has three strong generators, namely, the $\mf{su}(2)$ currents $j^0,j^\pm$. The stress tensor is given by the Sugawara construction 
\begin{equation}
    T\sim j^+j^++j^-j^-+j^0j^0~.
\end{equation}
For $n=1,2$, there are four null vectors \cite{Beem:2017ooy} given by 
\begin{align}
j^aT&=0,\quad n=1~,
\\
j^aT^2&=0,\quad n=2~,
\end{align}
and 
\begin{align}
T^2&=0~,\quad n=1
\\
T^3&=0~,\quad n=2~.
\end{align}
We will change the basis for $\mf{su}(2)$ and take the generators to be 
\begin{equation}\label{eq:basis_change_su(2)}
    j^x:=j^++ij^-,\quad j^z:=j^+-ij^-,\quad j^y:=ij^0~.
\end{equation}
Mapping 
\begin{equation}\label{eq:var_su(2)_zhu}
    j^x\mapsto X,\quad j^y\mapsto Y,\quad j^z\mapsto Z,\quad T\mapsto J~,
\end{equation}
we see that the $C_2$-algebra of $\mf{su}(2)_{-\frac{4}{3}}$ is given by 
\begin{equation}
    R:=R_{\CV_{(A_1,D_{3})}}\cong\frac{\C[X,Y,Z,J]}{\langle f_1,f_2,f_3,f_4,f_5\rangle}~,
\end{equation}
where 
\begin{equation}
    f_1:=XZ-Y^2-J,\quad f_2:=JX,\quad f_3:=JY,\quad f_4:=JZ,\quad f_5:=J^2~.
\end{equation}
The grading on the Zhu's $C_2$-algebra is given by 
\begin{equation}
    \mr{wt}(X,Y,Z)=1,\quad \mr{wt}(J)=2~.
\end{equation}
To compute the Hilbert series of the arc space of $R$, we start with power series 
\begin{equation}
    X(x)=\sum_{s=0}^\infty X_sx^s~,\quad Y(x)=\sum_{s=0}^\infty Y_sx^s~,\quad Z(x)=\sum_{s=0}^\infty Z_sx^s~,\quad J(x)=\sum_{s=1}^\infty J_sx^s~.
\end{equation}
Note that the power series for $J$ starts from $s=1$ owing to the fact that the $J$ represents the stress tensor in the physical theory. 
Define 
\begin{equation}
    R_s:=\C[X_0,\dots,X_s,Y_0,\dots,Y_s,Z_0,\dots,Z_s,J_1,\dots,J_s]/I_s~,
\end{equation}
where 
\begin{equation}
    I_s:=\left\langle f_1(x),\dots,f_5(x)\equiv 0 \bmod x^{s+1}\right\rangle~,
\end{equation}
and
\begin{equation}
    f_i(x):=f_i(X(x),Y(x),Z(x),J(x))~,\quad i=1,\dots,5~.
\end{equation}
For $n=0,1$, we obtain
\begin{equation}
\begin{aligned}
R_0 & =\mathds{C}\left[X_0, Y_0, Z_0\right] /\left\langle X_0 Z_0-Y_0^2\right\rangle~, \\
R_1 & =\mathds{C}\left[X_0, Y_0, Z_0, J_1, X_1, Y_1, Z_1\right] / I_1~,
\end{aligned}
\end{equation}
where 
\begin{equation}
    I_1 = \langle X_0 Z_0 - Y_0^2, X_0 Z_1-2 Y_0 Y_1+Z_0 X_1-J_1, X_0 J_1,Y_0 J_1,Z_0J_1 \rangle~.
\end{equation}
The arc space of $R$ is defined by
\begin{equation}
    R_{\infty}:=\lim _{ \leftarrow } R_s~.
\end{equation}
We assign a bigrading to the $s$-jet space $R_s$ by defining
\begin{equation}
\begin{split}
\mathsf{wt}^{(2)}\left(X_i\right)&=\left(\mathsf{wt}_1\left(X_i\right),\mathsf{wt}_2\left(X_i\right)\right)=(i, 1)~,\quad i\geq 0~, 
\\
\mathsf{wt}^{(2)}\left(Y_i\right)&=\left(\mathsf{wt}_1\left(Y_i\right),\mathsf{wt}_2\left(Y_i\right)\right)=(i, 1)~,\quad i\geq 0~, 
\\
\mathsf{wt}^{(2)}\left(Z_i\right)&=\left(\mathsf{wt}_1\left(Z_i\right),\mathsf{wt}_2\left(Z_i\right)\right)=(i, 1)~,\quad i\geq 0~, 
\\
\mathsf{wt}^{(2)}\left(J_i\right)&=\left(\mathsf{wt}_1\left(J_i\right),\mathsf{wt}_2\left(J_i\right)\right)=(i, 1)~,\quad i\geq 1~, 
\\
\mathsf{wt}^{(2)}\left(1\right)&=\left(\mathsf{wt}_1\left(1\right),\mathsf{wt}_2\left(1\right)\right)=(0, 0)~.
\end{split}
\end{equation}
We denote $(R_s)_{\ell,m}$, the homogeneous subspace of weight $(\ell,m)$.
The Hilbert series of $R_s$ is defined as
\begin{equation}
H_{R_s}(q, t)=\sum_{\ell=0}^s \sum_{m \geq 0} \mr{dim}\left(R_s\right)_{\ell, m} q^{\ell} t^m~.
\end{equation}
Taking $n \rightarrow \infty$, we get the Hilbert series of  $R_{\infty}$
\begin{equation}
H_{R_{\infty}}(q, t)=\sum_{\ell, m \geq 0} \mr{dim}\left(R_{\infty}\right)_{\ell, m} q^{\ell} t^m~.
\end{equation}
From the general formula \eqref{eq: hilbert_ser_Rn_gen} (proved in Appendix \ref{app:Hilb_ser_count}), the Hilbert series of $R_s$ for the case at hand is given by
\begin{equation}
\label{eq: hilbert_ser_Rn}
H_{R_s}(q,t) = \mr{Ser}_{q=0,t=0}\frac{\sum_{S\subseteq \mr{LM}(G_s)}(-1)^{|S|}(q,t)^{\mr{wt}^{(2)}(\mr{LCM}(S))}}{(1-t)^3 \prod_{i=1}^{s}\left(1-q^i t\right)^4} \quad (\mr{mod} \;q^{s+1})~,
\end{equation}
where $\mr{Ser}_{q=0,t=0}(f(q,t))$ denotes the power series of $f(q,t)$ around $q=t=0$, the summation in the numerator is over all subsets of monomials in $\mr{LM}(G_s)$, $|S|$ denotes the number of monomials in $S$, $\mr{LCM}(S)$ is least common multiple of the monomials in $S$, and we have defined
\begin{equation}
(q,t)^{(\ell,m)} := q^{\ell} t^m ~.   
\end{equation}
The Hilbert series of $R_\infty$ can then be given by taking $s\to \infty$
\begin{equation}
\label{eq: hilbert_ser_Rinf}
H_{R_\infty}(q,t) = \mr{Ser}_{q=0,t=0}\frac{\sum_{S\subseteq \mr{LM}(G_\infty)}(-1)^{|S|}(q,t)^{\mr{wt}^{(2)}(\mr{LCM}(S))}}{(1-t)^3 \prod_{i=1}^{\infty}\left(1-q^i t\right)^4}~.
\end{equation}
To apply \eqref{eq: hilbert_ser_Rn}, we need a general form for Gr\"{o}bner basis.
Through Mathematica computations (up to $s \leq 4$), we have the following conjecture for $\mr{LM}(G_s)$:
\begin{conj}
Choose the ordering of variables to be
\begin{equation}\label{eq:var_ord}
J_1>J_2>\cdots>X_0>X_1>\cdots>Y_0>Y_1>\cdots>Z_0>Z_1>\ldots~.
\end{equation}
Then with the degree lexicographic (deglex) ordering of monomials (see Appendix \ref{app:Grob}), the leading terms of the Gr\"{o}bner basis of  is given by 
\begin{equation}
\label{eq:Gr\"{o}bner_basis_Gn}
\mr{LM}\left(G_s\right)=\bigcup_{\substack{0 \leq \ell \leq s \\ 2 \leq m \leq \ell+2}} I_{\ell, m}~,
\end{equation}
where $I_{\ell,m}$ is a homogeneous set of weight $(\ell,m)$ given by
\begin{equation}
I_{\ell,m} := \left(\{X_{m-2}\}\cdot F_z(\ell-m+2, m-1)\right)\bigcup \left(\{J_{m-1}\}\cdot F(\ell-m+1, m-1)\right)~,
\end{equation}
where 
\begin{enumerate}
\item for sets $A,B$, we have introduced the notation
\begin{equation}
    A\cdot B:=\{ab:a\in A,b\in B\}~,
\end{equation}
\item for $\ell,m\in\IZ_{\geq 0}$, we have defined
\begin{equation}
F_z(\ell,m) := \{Z_{i_1}\dots Z_{i_m}: i_1+\dots +i_m = \ell\}~,
\end{equation}
\item for $\ell,m\in\IZ_{\geq 0}$, $F(\ell,m)$ is the set of all degree $m$ monomials that can be formed using the variables $X_0,\dots,X_\ell,Y_0\dots,Y_\ell,Z_0,\dots,Z_\ell,J_1,\dots,J_\ell$. More explicitly, 
\begin{equation}
\begin{aligned}
F(0,0) &= \{1\}~, \\
F(\ell,0) &= \emptyset, \quad \text{for}\quad \ell\neq 0~,\\
\end{aligned}
\end{equation}
and for $m>0$,
\begin{equation}
F(\ell,m):=\left\{X_0^{n_{X_0}}Y_0^{n_{Y_0}}Z_0^{n_{Z_0}} \prod_{i=1}^\ell X_i^{n_{X_i}}Y_i^{n_{Y_i}}Z_i^{n_{Z_i}}J_i^{n_{J_i}}: \sum_{i=0}^\ell (n_{X_i}+n_{Y_i}+n_{Z_i})+\sum_{i=1}^\ell n_{J_i}=m\right\}~.
\end{equation} 
\end{enumerate}
\end{conj}
\par
We now work order by order in $t$. At order $t^m$, the only terms in the numerator that contribute are the constraints $S\subset\mr{LM}(G_\infty)$ which have weight $(\ell,m)$.
\paragraph{Order $t^0$.} At order $t^0$, we only have one generator 1. Thus, the series starts with 1. 
\paragraph{Order $t^1$.} From \eqref{eq:Gr\"{o}bner_basis_Gn}, we see there is no constraint $S\subset\mr{LM}(G_\infty)$  weight $(\ell,1)$. Then from \eqref{eq: hilbert_ser_Rinf}, the contribution for the first order in $t$ can be given by
\begin{equation}\label{eq:t1_comp_D3}
    \begin{aligned}
H_{R_\infty}(q,t) &= \mr{Ser}_{q=0,t=0}\frac{1}{(1-t)^3 \prod_{i=1}^{\infty}\left(1-q^i t\right)^4}\quad (\mr{mod} \;t^2)\\
& = 1 + t\bigg(3+4\sum_{i=1}^\infty q^i\bigg)\\
& = 1+t\left(\frac{3+q}{1-q}\right) \\
& = 1+\frac{t}{(q ; q)_1} \sum_{k=0}^2\binom{2}{k}_q~,
\end{aligned}
\end{equation}
where the last equality follows from a straightforward calculation. 
\paragraph{Order $t^2$.} The constraints $S\subset\mr{LM}(G_\infty)$ of weight $(\ell,m=2)$ are given by
\begin{equation}
I_{\ell,m=2} = \{X_{0}\}\times F_z(\ell, 1)\bigcup \{J_{1}\}\times F(\ell-1, 1)~.
\end{equation}
More explicitly, these are 
\begin{equation}
\begin{aligned}
I_{\ell=0, m = 2} &=\{X_0Z_0\}~,\\
I_{\ell=1, m= 2}  &= \{X_0Z_1, J_1 X_0,J_1Y_0,J_1Z_0\}~,\\
I_{\ell, m=2}  &= \{X_0Z_\ell, J_1 X_\ell,J_1Y_\ell,J_1Z_\ell, J_1J_\ell\}, \quad \ell>1~.
\end{aligned}
\end{equation}
The Hilbert series to the second order in $t$ can thus be obtained from \eqref{eq: hilbert_ser_Rn}:
\begin{equation}
\label{eq:Hil_ser_t2}
\begin{aligned}
H_{R_\infty}(q,t) &= \mr{Ser}_{q=0,t=0}\frac{1-t^2(1+4q + 5\sum_{i=2}^\infty q^i)}{(1-t)^3 \prod_{i=1}^{\infty}\left(1-q^i t\right)^4}\quad (\mr{mod} \;t^3).
\end{aligned}
\end{equation}
Focusing on the $t^2$ term,
\begin{equation}
H_{R_{\infty}}(q, t)|_{t^2}=\sum_{\ell \geq 0} \mr{dim}\left(R_{\infty}\right)_{\ell, 2} q^\ell~,
\end{equation}
since the denominator is the generating function for all monomials, the coefficients of $t^2$ term in \eqref{eq:Hil_ser_t2} is given by
\begin{equation}\label{eq:dimRl2}
\mr{dim}\left(R_{\infty}\right)_{\ell, 2} = |F(\ell,2)|-|F(\ell,0)|-4|F(\ell-1,0)|-5\sum_{i=2}^\infty|F(\ell-i,0)|~.
\end{equation}
Note that 
\begin{equation}
\begin{aligned}
|F(\ell,0)| = \delta_{\ell,0}~.
\end{aligned}
\end{equation}
Our goal is to compute $|F(\ell,2)|$. 
The first few set of free monomials is given  by
\begin{equation}
\begin{aligned}
F(0,2) &= \{X_0^2, Y_0^2,Z_0^2,X_0Y_0,X_0Z_0,Y_0Z_0\}~,\\
F(1,2) &= \{X_0, Y_0, Z_0\}\cdot \{X_1,Y_1,Z_1,J_1\}~,\\
F(2,2) & = \{X_0,Y_0,Z_0\}\cdot \{X_2,Y_2,Z_2,J_2\}\cup\{X_1^2, Y_1^2,Z_1^2,J_1^2,X_1Y_1,\dots\}~,\\
F(3,2) & = \{X_0,Y_0,Z_0\}\cdot \{X_3,Y_3,Z_3,J_3\}\cup\{X_1,Y_1,Z_1,J_1\}\cdot \{X_2,Y_2,Z_2,J_2\}~,\\
&\dots
\end{aligned}
\end{equation}
We anticipate the formula
\begin{equation}
\label{eq:num_Fl2}
|F(\ell,2)|=\begin{cases}
8\ell + 6, & \text{if $\ell$ is even}\\
8\ell + 4, & \text{if $\ell$ is odd}~.
\end{cases}
\end{equation}
We combine two cases 
\begin{equation}\label{eq:free_words_t2}
|F(\ell,2)| = 8\ell + 5 + (-1)^\ell~.
\end{equation}
Let us prove this. Note that
a monomial of total weight $(\ell, 2)$ must satisfy:
\begin{equation}
\sum (\text{indices} \times \text{exponents}) = \ell \quad \text{and} \quad \sum (\text{exponents}) = 2.
\end{equation}
We analyze two cases:
\\
\textbf{Case 1: Squared Variables.}
A monomial of the form $V_i^2$ contributes weight $(2i, 2)$. For this to equal $(\ell, 2)$, we require $2i = \ell$, i.e., $i = \ell/2$. 
\begin{itemize}
    \item If $\ell$ is even, $i = \ell/2$ is an integer. The variables available are $X_{\ell/2}, Y_{\ell/2}, Z_{\ell/2}$ (3 options) and $J_{\ell/2}$ (if $\ell/2 \geq 1$, since $J_0$ doesn't exist). 
    \begin{equation}
    \text{Count} = 
    \begin{cases} 
    4 & \text{if } \ell \geq 2 \text{ even}, \\
    3 & \text{if } \ell = 0.
    \end{cases}
    \end{equation}
    \item If $\ell$ is odd, no such $i$ exists. Count $= 0$.
\end{itemize}
\textbf{Case 2: Product of Two Distinct Variables.}
A monomial $V_i W_j$ with $i \neq j$ contributes weight $(i + j, 2)$. For $i + j = \ell$, we count pairs $(V_i, W_j)$:
\begin{itemize}
    \item For $i = 0$ or $j=0$: Variables $X_0, Y_0, Z_0$ (3 choices). Partner $W_j$ must satisfy $j = \ell$. For $j \geq 1$, there are 4 variables ($X_j, Y_j, Z_j, J_j$).
    \item For $i \geq 1$: Both $V_i$ and $W_j$ have 4 choices each ($X, Y, Z, J$).
\end{itemize}
Total count for this case:
\begin{equation}
\text{Count} = 
\begin{cases} 
\frac{1}{2}(3 \times 4) + \frac{n}{2}(4 \times 4) & \text{if } \ell \text{ even}, \\
(3 \times 4) + (4 \times 4) \times \left(\frac{\ell - 1}{2}\right) & \text{if } \ell \text{ odd}.
\end{cases}
\end{equation}
Simplifying, we get
\begin{equation}
\text{Count} = 
\begin{cases} 
8\ell + 6 & \text{if } \ell \text{ even}, \\
8\ell + 4 & \text{if } \ell \text{ odd}~.
\end{cases}
\end{equation}
This completes the proof. 
Using \eqref{eq:num_Fl2} and \eqref{eq:dimRl2}, we have 
\begin{equation}
H_{R_\infty}(q,t)|_{t^2} = t^2\left(5+8 q+\sum_{r \geq 2}\left(8 r+(-1)^r\right) q^r\right)~.
\end{equation}
Noting that
\begin{equation}
\sum_{r \geq 1} r q^r=\frac{q}{(1-q)^2}, \quad \sum_{r \geq 1}(-1)^r q^r=-\frac{q}{1+q}~.
\end{equation}
We have
\begin{equation}\label{eq:I_2-1_calc}
\begin{aligned}
5+8 q+\sum_{r \geq 2}\left(8 r+(-1)^r\right) q^r & =5+8 q+\left(\frac{q}{(1-q)^2}-q\right)+\left(-\frac{q}{1+q}+q\right) \\
& =4+q+\frac{8 q}{(1-q)^2}+\frac{1}{1+q}~.
\end{aligned}
\end{equation}
Using straightforward  calculation, we find that 
\begin{equation}\label{eq:I_2-2_calc}
H_{R_\infty}(q,t)|_{t^2} = \frac{1}{(q ; q)_2} \sum_{k=0}^4\binom{4}{k}_q.
\end{equation}
Because of the technicality of the discussion, we present the evaluation of the Hilbert series at order $t^3$ in Appendix \ref{app:Hilb_ser_t3}. 

\section{Macdonald Index For $(A_1,D_{2n+1})$ Theory}\label{sec:A1D2n+1_Mac_conj}
In this section, we conjecture a generalization of formula \eqref{eq:Mac_ind_A1D3} for the Macdonald index for the general $(A_1,D_{2n+1})$ theory. We motivate our formula using some observations from the Hilbert series of arc space for the Zhu's $C_2$-algebra of the chiral algebra of the $(A_1,D_{5})$ theory and its relation to the Macdonald index of the projected free hypermultiplet. We then check the consistency of our conjecture with the Schur limit. This requires us to prove a generalization of the product-sum formula \eqref{eq:new_idn1}, see \eqref{eq:new_id_gen_n} below. 
Let us first state our conjecture.
\begin{conj}
The Macdonald index for the $(A_1, D_{2n+1})$ Argyres-Douglas theory is given by
\begin{equation}\label{eq:Mac_ind_gen_n}
    \CI^{(A_1,D_{2n+1})}_{\mr{M}} = 1 + \sum_{m=1}^\infty 
    t^m \sum_{l=1}^m
    \frac{f^n_{m,l}(q)}{(q;q)_l}\sum_{k=0}^{2l}\binom{2l}{k}_q~,
\end{equation}
where $f^n_{m,l}(q)$ is a class of polynomial functions in $q$ which satisfy the hypergeometric relation 
\begin{equation}\label{eq:hyper_fmnl_rel}
    1 + \sum_{m=1}^\infty q^m \sum_{l=1}^m 
    \frac{f^n_{m,l}(q)}{(q;q)_l} (q^N; q)_l 
    (q^{1-N};q)_l = q^{nN(N-1)}~.
\end{equation}
\end{conj}
As we explain in Section \ref{sec:schur_gen_n}, the condition \eqref{eq:hyper_fmnl_rel} on $f^n_{m,l}$ makes sure that our proposed formula satisfies the Schur limit. 
In Appendix \ref{app:new_id}, we prove that the functions
\begin{equation}\label{eq:fnml_closed_form_sec5}
f^k_{m,l}(q)=\sum_{i_k,i_{k-1},\dots,i_3\geq 0}\frac{q^{Q(m-(k-1)i_k-(k-2)i_{k-1}-\dots-2i_3-l,i_3,\dots,i_k)}(q;q)_{l}}{(q;q)_{i_k}\dots (q;q)_{i_3}(q;q)_{m-(k-1)i_k-(k-2)i_{k-1}-\dots-2i_3-l}(q;q)_{2l+i_3+2i_4+\dots+(k-2)i_k-m}}~,    
\end{equation}
where 
\begin{equation}\label{eq:Q_def_sec5}
    Q(i_2,i_3,\dots,i_k):=i_k^2+(i_k+i_{k-1})^2+\dots+(i_k+\dots+i_3+i_2)^2=\sum_{\ell=2}^{k}\bigg(\sum_{j=\ell}^{k}i_{j}\bigg)^{2}~,
\end{equation}
satisfy \eqref{eq:hyper_fmnl_rel}. 
For $k=1,2$, the sum in \eqref{eq:fnml_closed_form_sec5} is empty, so we need to define it separately. For $k=2$, we get from \eqref{eq:fnml_closed_form_sec5}: 
\begin{eqsp}\label{eq:f2ml_def}
    f^2_{m,l}(q)=q^{(m-l)^2}\frac{(q;q)_{l}}{(q;q)_{m-l}(q;q)_{2l-m}}~.
\end{eqsp}
For $k=1$, we define 
\begin{eqsp}\label{eq:f1ml_def}
   f^{1}_{m,l}(q)=\delta_{m,l}~. 
\end{eqsp}
From the definition \eqref{eq:fnml_closed_form_sec5}, we see that 
\begin{eqsp}
    f^n_{m,m}(q)=1~,\quad n\geq 1,~~m\geq 0~.
\end{eqsp}
It was checked using Mathematica that with $f^k_{m,l}$ given above, the proposed Macdonald index \eqref{eq:Mac_ind_gen_n} reproduces the known coefficients in Tables \ref{tab:coeff_mac_A1D5} -- \ref{tab:coeff_mac_A1D9}.
\subsection{Motivating The Formula Using Hilbert Series}
As discussed in Section \ref{sec:RG_flow} and illustrated in Section \ref{sec:D3_free_hyper_rel}, the Macdonald index for $(A_1,D_{2n+1})$ theory must be related to that of the (projected) free hypermultiplet. 
\begin{table}[]
    \centering
    \begin{tabular}{|l|lllllllll|}
\hline Variables & $t^0$ & $t^1$ & $t^2$ & $t^3$ & $t^4$ & $t^5$ & $t^6$ & $t^7$ & $t^8$ \\
\hline
$q^0$ & 1 & 3 & 5 & 7 & 9 & 11 & 13 & 15 & 17 \\
$q^1$ & 0 & 4 & 11 & 17 & 23 & 29 & 35 & 41 & 47 \\
$q^2$ & 0 & 4 & 21 & 40 & 56 & 72 & 88 & 104 & 120 \\
$q^3$ & 0 & 4 & 27 & 73 & 114 & 150 & 186 & 222 & 258 \\
$q^4$ & 0 & 4 & 37 & 119 & 218 & 303 & 381 & 459 & 537 \\
$q^5$ & 0 & 4 & 43 & 173 & 361 & 546 & 709 & 863 & 1017 \\
$q^6$ & 0 & 4 & 53 & 243 & 583 & 942 & 1279 & 1586 & 1882 \\
\hline
\end{tabular}
    \caption{Table of selected coefficients of the Macdonald index for $(A_1,D_5)$ theory.}
    \label{tab:coeff_mac_A1D5}
\end{table}
\begin{table}[]
    \centering
    \begin{tabular}{|l|lllllllll|}
\hline Variables & $t^0$ & $t^1$ & $t^2$ & $t^3$ & $t^4$ & $t^5$ & $t^6$ & $t^7$ & $t^8$ \\
\hline $q^0$ & 1 & 3 & 5 & 7 & 9 & 11 & 13 & 15 & 17 \\
$q^1$ & 0 & 4 & 11 & 17 & 23 & 29 & 35 & 41 & 47 \\
$q^2$ & 0 & 4 & 21 & 43 & 61 & 79 & 97 & 115 & 133 \\
$q^3$ & 0 & 4 & 27 & 77 & 127 & 169 & 211 & 253 & 295 \\
$q^4$ & 0 & 4 & 37 & 123 & 243 & 349 & 443 & 537 & 631 \\
$q^5$ & 0 & 4 & 43 & 177 & 401 & 638 & 846 & 1038 & 1230 \\
$q^6$ & 0 & 4 & 53 & 247 & 639 & 1109 & 1552 & 1952 & 2332 \\
\hline
\end{tabular}
    \caption{Table of selected coefficients of the Macdonald index for $(A_1,D_7)$ theory.}
    \label{tab:coeff_mac_A1D7}
\end{table}
\begin{table}[]
    \centering
    \begin{tabular}{|l|lllllllll|}
\hline Variables & $t^0$ & $t^1$ & $t^2$ & $t^3$ & $t^4$ & $t^5$ & $t^6$ & $t^7$ & $t^8$ \\
\hline $q^0$ & 1 & 3 & 5 & 7 & 9 & 11 & 13 & 15 & 17 \\
$q^1$ & 0 & 4 & 11 & 17 & 23 & 29 & 35 & 41 & 47 \\
$q^2$ & 0 & 4 & 21 & 43 & 61 & 79 & 97 & 115 & 133 \\
$q^3$ & 0 & 4 & 27 & 77 & 130 & 174 & 218 & 262 & 306 \\
$q^4$ & 0 & 4 & 37 & 123 & 247 & 362 & 462 & 562 & 662 \\
$q^5$ & 0 & 4 & 43 & 177 & 405 & 663 & 892 & 1100 & 1308 \\
$q^6$ & 0 & 4 & 53 & 247 & 643 & 1149 & 1644 & 2089 & 2507 \\
\hline
\end{tabular}
    \caption{Table of selected coefficients of the Macdonald index for $(A_1,D_9)$ theory.}
    \label{tab:coeff_mac_A1D9}
\end{table}
Let us write down the Macdonald index for $(A_1, D_3)$ as 
\begin{equation}
\mathcal{I}_{\mathrm{M}}^{\left(A_1, D_3\right)}=\sum_{m=0}^{\infty} \frac{t^m}{(q ; q)_m} \sum_{k=0}^{2 m}\binom{2 m}{k}_q =: \sum_{m=0}^\infty t^m \mathcal{I}_m(q)~,
\end{equation}
where 
\begin{equation}
\mathcal{I}_m(q):= \frac{1}{(q ; q)_m} \sum_{k=0}^{2 m}\binom{2 m}{k}_q~.   
\end{equation}
We anticipate that the relation of the Macdonald index of $(A_1,D_{2n+1})$ theory to that of the (projected) free hypermultiplet is through the functions $\CI_m$. So, we want to express the Macdonald index of $(A_1, D_{2n+1})$ in terms of $\mathcal{I}_m$.

We start from $(A_1, D_5)$ theory, where the Zhu's $C_2$-algebra is known for its corresponding chiral algebra $\mathcal{V}_{(A_1, D_5)}=\mathfrak{s u}(2)_{-\frac{8}{5}}$ \cite{Beem:2017ooy}\footnote{Again, we are using the basis of $\mf{su}(2)$ given in \eqref{eq:basis_change_su(2)} and \eqref{eq:var_su(2)_zhu}.}:
\begin{equation}
R^{(A_1, D_5)}:=R_{\mathcal{V}_{\left(A_1, D_5\right)}} \cong \frac{\mathbb{C}[X, Y, Z, J]}{\left\langle f_1, f_2, f_3, f_4, f_5\right\rangle}~,
\end{equation}
where 
\begin{equation}\label{eq:const_D4}
f_1:= XZ - Y^2 - J, \quad  f_2:= J^2 X, \quad f_3:= J^2Y, \quad f_4:= J^2 Z, \quad f_5:= J^3~.
\end{equation}
By Theorem \ref{thm:Hil_ser_mac}, the Macdonald index for $(A_1,D_5)$ theory is equal to the Hilbert series of the arc space of the associated Zhu's $C_2$-algebra:
\begin{equation}
\mathcal{I}^{(A_1, D_5)}_{\mr{M}}(q,t) = H_{R^{(A_1, D_5)}_{\infty}}(q, t)~.
\end{equation}
Let us compute the Hilbert series of $R^{(A_1, D_5)}_\infty$ till $O(t^3)$. We will again use the deglex monomial ordering with the ordering on variables given by \eqref{eq:var_ord} along with the general formula \eqref{eq: hilbert_ser_Rinfty_gen}.

\paragraph{Order $t^0$.} At order $t^0$, we only have one generator 1. Thus, the series starts with 1. 
\paragraph{Order $t^1$.} From \eqref{eq:const_D4}, we see that there are no elements of the Gr\"{o}bner basis of $I_n$ for any $n$, which has weight $(\ell,1)$. Thus we have 
\begin{equation}
\begin{aligned}
H_{R^{(A_1, D_5)}_{\infty}}(q, t)&=\mr{Ser}_{q=0, t=0} \frac{1}{(1-t)^3 \prod_{i=1}^{\infty}\left(1-q^i t\right)^4} \quad\left(\mr{mod} ~t^2\right)\\
& = 1+ t \mathcal{I}_1~,
\end{aligned}
\end{equation}
where we used the computations of \eqref{eq:t1_comp_D3}.
\paragraph{Order $t^2$.} From Mathematica computations (checked till $G_4$), we find that the leading monomials of a Gr\"{o}bner basis $G_\infty$ with weight  $(\ell,m=2)$ are given by 
\begin{equation}\label{eq:Grob_D5t2}
\{X_0 Z_\ell:\ell\geq 0\}~.    
\end{equation}
Then using \eqref{eq: hilbert_ser_Rn_gen}, we get 
\begin{equation}
\begin{aligned}
\left.H_{R^{(A_1, D_5)}_{\infty}}(q, t)\right|_{t^2}&=\left.\mr{Ser}_{q=0, t=0} \frac{1-t^2 \sum_{\ell=0}^{\infty}q^\ell}{(1-t)^3 \prod_{i=1}^{\infty}\left(1-q^i t\right)^4}\right|_{t^2}~.
\end{aligned}
\end{equation}
Note that 
\begin{equation}
\sum_{\ell=0}^{\infty}q^\ell=1+q+\sum_{\ell=2}^\infty q^\ell=1+4q+5\sum_{\ell=2}^{\infty}q^\ell-q\left(3+4\sum_{\ell=1}^{\infty}q^\ell\right)=1+4q+5\sum_{\ell=2}^{\infty}q^\ell-q\CI_1~.    
\end{equation}
Thus using \eqref{eq:free_words_t2}, we get 
\begin{equation}
\begin{split}
\left.H_{R^{(A_1, D_5)}_{\infty}}(q, t)\right|_{t^2}&=\left. \mr{Ser}_{q=0, t=0} \frac{1-t^2 (1+ 4 q + 5\sum_{\ell=2}^{\infty}q^\ell)+t^2q\mathcal{I}_1}{(1-t)^3 \prod_{i=1}^{\infty}\left(1-q^i t\right)^4}\right|_{t^2}
\\
&=t^2\left(\sum_{\ell=0}^\infty (8\ell+5+(-1)^\ell)q^\ell-1- 4 q - 5\sum_{\ell=2}^{\infty}q^\ell+q\CI_1\right)
\\
&=t^2\left(5+8q+\sum_{\ell=2}^\infty (8\ell+(-1)^\ell)q^\ell+q\CI_1\right)
\\
& = t^2\left(\mathcal{I}_2+q\mathcal{I}_1\right)~,
\end{split}    
\end{equation}
where we used the calculation from \eqref{eq:I_2-1_calc} -- \eqref{eq:I_2-2_calc}.
\paragraph{Order $t^3$.} From Mathematica computations (checked till $G_4$), we find that the leading monomials of a Gr\"{o}bner basis $G_\infty$ with weight  $(\ell,m=3)$ are given by 
\begin{align}
&\{J_1^2X_0,J_1^2 Y_0,J^2_1 Z_0\},\quad \ell=2~,
\\
&\{J_1^2 X_{\ell-2},J_1^2 Y_{\ell-2},J^2_1 Z_{\ell-2}, J_1^2 J_{\ell-2}\},\quad \ell>2~.
\end{align}
Combined with \eqref{eq:Grob_D5t2} and using \eqref{eq: hilbert_ser_Rinfty_gen}, we have
\begin{equation}
\begin{aligned}
H_{R^{(A_1, D_5)}_{\infty}}(q, t)|_{t^3} & = \mr{Ser}_{q=0, t=0} \frac{1-t^2 \sum_{\ell=0}^{\infty}q^\ell -t^3(3 q^2 +4\sum_{\ell=3}^\infty q^\ell)}{(1-t)^3 \prod_{i=1}^{\infty}\left(1-q^i t\right)^4}\bigg|_{t^3}~.
\end{aligned}
\end{equation}
From \eqref{eq:Hil_ser_modt4} and \eqref{eq:I-3_ser} we see that 
\begin{equation}
    \CI_3=\left.\mr{Ser}_{q=0,t=0}\frac{1-t^2(1+4q + 5\sum_{\ell=2}^\infty q^\ell)+t^3(4q + 8 q^2 + 11 q^3 + 12 \sum_{\ell\geq 4}q^\ell)}{(1-t)^3 \prod_{i=1}^{\infty}\left(1-q^i t\right)^4}\right|_{t^3}~.
\end{equation}
Noting that 
\begin{equation}
\begin{split}
    1-t^2 \sum_{\ell=0}^{\infty}q^\ell -t^3\left(3 q^2 +4\sum_{\ell=3}^\infty q^\ell\right)&=t^2 q \mathcal{I}_1 - t^3(3q \mathcal{I}_1 +q^3 \mathcal{I}_1-5 q -q^2)+1-t^2\left(1+4q + 5\sum_{\ell=2}^\infty q^\ell\right)\\&+t^3\left(4q + 8 q^2 + 11 q^3 + 12 \sum_{\ell\geq 4}q^\ell\right)~,
\end{split}    
\end{equation}
we get 
\begin{equation}
\begin{split}
 H_{R^{(A_1, D_5)}_{\infty}}(q, t)|_{t^3} & = \mathcal{I}_3 +\mr{Ser}_{q=0, t=0} \frac{t^2 q \mathcal{I}_1 - t^3(3q \mathcal{I}_1 +q^3 \mathcal{I}_1-5 q -q^2)}{(1-t)^3 \prod_{i=1}^{\infty}\left(1-q^i t\right)^4}\bigg|_{t^3}
 \\
 &=\mathcal{I}_3 + q \mathcal{I}_1\mr{Ser}_{q=0, t=0} \frac{1}{(1-t)^3 \prod_{i=1}^{\infty}\left(1-q^i t\right)^4}\bigg|_{t} -\mr{Ser}_{q=0, t=0} \frac{t^3(3q \mathcal{I}_1 +q^3 \mathcal{I}_1-5 q -q^2)}{(1-t)^3 \prod_{i=1}^{\infty}\left(1-q^i t\right)^4}\bigg|_{t^3}
 \\
& = \mathcal{I}_3 + q \mathcal{I}_1^2 -3q \mathcal{I}_1 -q^3 \mathcal{I}_1+5 q +q^2~.   \end{split}    
\end{equation}
We can express $\mathcal{I}_1, \mathcal{I}_1^2$ by using $\mathcal{I}_2$ by
\begin{equation}
\mathcal{I}_1 = \frac{(1-q^2) \mathcal{I}_2 - 5+ q}{3 q + q^3}~, \quad 
\mathcal{I}_1^2 = \frac{1}{q}\left((1+q)\mathcal{I}_2 -5 - 4q - q^2\right)~. 
\end{equation}
Then
\begin{equation}
H_{R^{(A_1, D_5)}_{\infty}}(q, t)|_{t^3} = \mathcal{I}_3 + (q + q^2)\mathcal{I}_2~.
\end{equation}
The calculation above suggests the following form for the Macdonald index:
\begin{equation}\label{eq:Mac_A1D5_conj}
    \CI_{\mr{M}}^{(A_1,D_5)}(q,t)=1+\sum_{m=1}^\infty 
    t^m \sum_{l=1}^m
    \frac{f^2_{m,l}(q)}{(q;q)_l}\sum_{k=0}^{2l}\binom{2l}{k}_q~,
\end{equation}
where $f^2_{m,l}(q)$ is a polynomial which we found for some special values of $m,l$ above:
\begin{equation}
\begin{split}
    f^2_{1,1}&=1,
    \\
    f^2_{2,1}&=q, \quad f^2_{2,2}=1~,
    \\
    f^2_{3,1}&=0, \quad f^2_{3,2}=q+q^2, \quad f^2_{3,3}=1~. 
\end{split}    
\end{equation}
Using Mathematica computations (checked up to $O(t^{12})$), and matching the results to the Macdonald index \eqref{eq:Mac_class_S-1} -- \eqref{eq:Mac_class_S-7} obtained from class $\CS$ description, we were able to compute the functions $f^2_{m,l}(q)$ for high values of $m,l$. The result is recorded in Appendix \ref{app:data_fmnl} and are consistent with the closed formula \eqref{eq:fnml_closed_form}.
\par

Numerical evidence suggests the general form of Macdonald index for $(A_1,D_{2n+1})$ theory to be
\begin{equation}\label{eq:Mac_A1D2n+1_conj}
\CI_{\mr{M}}^{(A_1,D_{2n+1})}(q,t)=1+\sum_{m=1}^\infty 
    t^m \sum_{l=1}^m
    \frac{f^n_{m,l}(q)}{(q;q)_l}\sum_{k=0}^{2l}\binom{2l}{k}_q~,    
\end{equation}
for some polynomials $f^n_{m,l}(q)$. Again, matching with the Macdonald index \eqref{eq:Mac_class_S-1} -- \eqref{eq:Mac_class_S-7} (we checked up to $O(t^{12})$ for $n=2,3$), we were able to compute the polynomials $f_{m,l}^n(q)$, the results are recorded in Appendix \ref{app:data_fmnl} for $n=2,3$. 
\par 
In addition, we checked on Mathematica (up to $O(q^5)$ and  $O(t^6)$) that the Hilbert series of the arc space of the variety 
\begin{equation}\label{eq:variety_gen_n}
    R:=\frac{\C[X,Y,Z,J]}{\langle g_1,g_2,g_3,g_4,g_5\rangle}~,
\end{equation}
where 
\begin{equation}\label{eq:const_gen_n}
    g_{n,1}:=XZ-Y^2-J\quad g_{n,2}:=J^{n}X,\quad g_{n,3}:=J^{n}Y,\quad g_{n,4}:=J^{n}Z,\quad g_{n,5}:=J^{n+1}~,
\end{equation}
reproduces the Macdonald index \eqref{eq:Mac_ind_gen_n} to high orders in $t$. This suggests the following conjecture:
\begin{conj}
The Zhu's $C_2$-algebra of $\mf{su}(2)_{-\frac{4n}{2n+1}}$ is given by 
\begin{equation}
    R_{\mathfrak{su}(2)_{-\frac{4n}{2n+1}}}\cong R~,
\end{equation}
with $R$ given by \eqref{eq:variety_gen_n}.
\end{conj}

\subsection{Matching The Schur Limit}\label{sec:schur_gen_n}
Matching the Schur limit of the formula \eqref{eq:Mac_A1D2n+1_conj} amounts to proving the following product-sum identity:
\begin{equation}\label{eq:new_id_gen_n}
   \left(\frac{\left(q^{2n+1} ; q^{2n+1}\right)_{\infty}}{(q ; q)_{\infty}}\right)^3= \sum_{m=0}^\infty 
    q^m \sum_{l=0}^m
    \frac{f^n_{m,l}(q)}{(q;q)_l}\sum_{k=0}^{2l}\binom{2l}{k}_q~.
\end{equation}
Closely examining the proof of Theorem \ref{thm:main_id}, we notice that \eqref{eq:new_id_gen_n} holds if
$f^n_{m,l}$ satisfies the hypergeometric relation:
\begin{equation}\label{eq:fmnl_def}
    \sum_{m=0}^\infty q^m \sum_{l=0}^m 
    \frac{f^n_{m,l}(q)}{(q;q)_l} (q^N; q)_l 
    (q^{1-N};q)_l = q^{nN(N-1)}~.
\end{equation}
Indeed, we can prove the more general identity 
\begin{equation}\label{eq:main_id_gen}
\frac{\left(z ; q^{2n+1}\right)_{\infty}\left(q^{2n+1} / z ; q^{2n+1}\right)_{\infty}\left(q^{2n+1} ; q^{2n+1}\right)_{\infty}}{(z ; q)_{\infty}(q / z ; q)_{\infty}(q ; q)_{\infty}}=\sum_{m=0}^{\infty}q^m\sum_{l=0}^m \frac{f^n_{m,l}(q)}{(q ; q)_l} \sum_{k=-l}^l{
2 l \choose
l+k}_q z^k~.    
\end{equation}
The idea is to again compare the powers of $z^N$ on both sides of \eqref{eq:main_id_gen}:
\begin{equation}\label{eq:genid_intermediate}
\sum_{N=-\infty}^{\infty}(-1)^N q^{(2n+1)\binom{N}{2}} z^N=\sum_{r=-\infty}^{\infty}(-1)^r q^{\binom{r}{2}} z^r 
\sum_{m=0}^{\infty}q^m \sum_{l=0}^m \frac{f^n_{m,l}(q)}{(q ; q)_l} \sum_{k=-l}^l{
2 l \choose
l+k}_q z^k~.     
\end{equation}
Again, we will only show that coefficients of $z^N$ match on both sides for $N\geq 0$ and use the $z\to z^{-1}$ symmetry on both sides for $N<0$. The coefficient of $z^N$ in the \tsf{LHS} of \eqref{eq:genid_intermediate} is 
\begin{equation}
    (-1)^Nq^{(2n+1){N\choose 2}}
\end{equation}
The coefficient of $z^N$ on the \tsf{RHS} of \eqref{eq:genid_intermediate} is
\begin{equation}
\begin{aligned}
& \sum_{m=0}^{\infty}q^m\sum_{l=0}^m \frac{f^n_{m,l}(q)}{(q ; q)_l} \sum_{k=-l}^l{
2 l \choose
l+k}_q(-1)^{N-k} q^{\binom{N-k}{2}} \\
& \quad=(-1)^N q^{\binom{N}{2}} \sum_{m=0}^{\infty} \sum_{l=0}^m \frac{f^n_{m,l}(q)}{(q ; q)_l} \sum_{k=-l}^l{
2 l \choose
l+k}_q(-1)^k q^{\binom{k}{2}+k(1-N)}\\
& \quad=(-1)^N q^{\binom{N}{2}} \sum_{m=0}^\infty q^m \sum_{l=0}^m 
    \frac{f^n_{m,l}(q)}{(q;q)_l} (q^N; q)_l 
    (q^{1-N};q)_l~.
\end{aligned}
\end{equation}
Now using \eqref{eq:fmnl_def} we get 
\begin{equation}
\sum_{m=0}^{\infty}q^m\sum_{l=0}^m \frac{f^n_{m,l}(q)}{(q ; q)_l} \sum_{k=-l}^l{
2 l \choose
l+k}_q(-1)^{N-k} q^{\binom{N-k}{2}} =(-1)^N q^{\binom{N}{2}}q^{nN(N-1)}=(-1)^Nq^{(2n+1){N\choose 2}}~.    
\end{equation}
This proves the identity \eqref{eq:main_id_gen}. Taking $z\to 1$ in \eqref{eq:main_id_gen} gives \eqref{eq:new_id_gen_n}. 
\subsection{A Surprising Relation Between Indices Of $(A_1,D_{2n+1})$ And $(A_1,A_{2(n-1)})$}\label{sec:prod_A2n_fH_Mac}
From the discussion in Section \ref{sec:RG_flow}, recall that $(A_{1},D_{2n+1})$ flows to free hypermultiplet and $(A_1,A_{2(n-1)})$ theory. Although indices are not expected to be preserved under \tsf{RG} flow, we will show that the Macdonald index satisfies the surprising relation
\begin{equation}
    \CI^{(A_1,D_{2n+1})}_{\mr{M}}\sim \CI^{(A_1,A_{2(n-1)})}_{\mr{M}}\cdot \CI^{\mr{HM}/\Z_2}_{\mr{M}}~,
\end{equation}
where $\sim$ means that we need to include a correction factor. 
As an example, we compute the coefficients of the product of Macdonald indices of $(A_1,A_2)$ theory and the projected free hyper in Table \ref{tab:coeff_mac_A1A2_fH}. Comparing it with the coefficients of the Macdonald index for $(A_1,D_5)$ theory in Table \ref{tab:coeff_mac_A1D5}, we find that the coefficients above the diagonal match. 
\par
We now show that the product of the Macdonald indices of $(A_1,A_{2(n-1)})$ theory and the projected free hypermultiplet matches the Macdonald index of $(A_1,D_{2n+1})$ up to a correction factor (note that unlike in the case of $(A_1,D_3)$, the correction factor no longer restricts reassignments to a single column, but mixes different columns). We begin by showing this for $n=2$ where the analysis is cleaner. 
\begin{table}[]
    \centering
    \begin{tabular}{|l|lllllllll|}
\hline Variables & $t^0$ & $t^1$ & $t^2$ & $t^3$ & $t^4$ & $t^5$ & $t^6$ & $t^7$ & $t^8$ \\
\hline
$q^0$ & 1 & 3 & 5 & 7 & 9 & 11 & 13 & 15 & 17 \\
$q^1$ & 0 & 5 & 11 & 17 & 23 & 29 & 35 & 41 & 47 \\
$q^2$ & 0 & 8 & 24 & 40 & 56 & 72 & 88 & 104 & 120 \\
$q^3$ & 0 & 9 & 42 & 78 & 114 & 150 & 186 & 222 & 258 \\
$q^4$ & 0 & 12 & 69 & 147 & 225 & 303 & 381 & 459 & 537 \\
$q^5$ & 0 & 13 & 98 & 247 & 401 & 555 & 709 & 863 & 1017 \\
$q^6$ & 0 & 16 & 141 & 402 & 698 & 994 & 1290 & 1586 & 1882 \\
\hline
\end{tabular}
    \caption{Table of selected coefficients of the Macdonald index for $(A_1,A_2)\otimes\mr{HM}/\mathds{Z}_2$ theory. The coefficients above the diagonal matches with the index for $(A_1,D_5)$ in Table \ref{tab:coeff_mac_A1D5}.}
    \label{tab:coeff_mac_A1A2_fH}
\end{table}
\subsubsection{$(A_1,D_5)\to(A_1,A_2)\otimes\mr{HM}/\mathds{Z}_2$}
The Macdonald index of the $(A_1,A_2)$ theory is given by \cite{Bhargava:2023hsc}\footnote{Note again that because of our convention, we do no have a factor of $q^n$, compared to \cite{Bhargava:2023hsc}.}:
\begin{equation}
    \CI_{\mr{M}}^{(A_1,A_2)}(q,t)=\sum_{r=0}^{\infty}t^r\frac{q^{r^2}}{(q;q)_r}~.
\end{equation}
The Macdonald index of the $\mathds{Z}_2$-projected free hypermultiplet is,
\begin{equation}
    \CI_{\mr{M}}^{\mr{HM}/\mathds{Z}_2}(q,t)=\sum_{l=0}^{\infty}\frac{t^l}{(q;q)_{2l}}\sum_{k=0}^{2l}{2l\choose k}_q
\end{equation}
Taking a product of these two indices, we obtain,
\begin{equation}
    \sum_{r,l=0}^{\infty}t^{r+l}\frac{q^{r^2}}{(q;q)_{2l}(q;q)_r}\sum_{k=0}^{2l}{2l\choose k}_q
\end{equation}
Substituting $m=r+l$,
\begin{equation}
    \sum_{m=0}^{\infty}t^m\sum_{l=0}^m\frac{q^{(m-l)^2}}{(q;q)_{m-l}}\frac{\textcolor{red}{1}}{\textcolor{red}{(q;q)_{2l}}}\sum_{k=0}^{2l}{2l\choose k}_q
\end{equation}
Now, we turn to the $(A_1,D_5)$ index. Using the fact that 
\begin{equation}
    f^2_{m,l}(q)=q^{(m-l)^2}{l\choose m-l}_q=q^{(m-l)^2}\frac{(q;q)_{l}}{(q;q)_{m-l}(q;q)_{2l-m}}~,
\end{equation}
we have 
\begin{equation}
\begin{split}
    \CI_{\mr{M}}^{(A_1,D5)}(q,t)&=\sum_{m=0}^{\infty}t^m\sum_{l=0}^{m}\frac{q^{(m-l)^2}}{(q;q)_l}\frac{(q;q)_{l}}{(q;q)_{m-l}(q;q)_{2l-m}}\sum_{k=0}^{2l}{2l\choose k}_q
    \\
    &=\sum_{m=0}^{\infty}t^m\sum_{l=0}^{m}\frac{q^{(m-l)^2}}{(q;q)_{m-l}}\frac{1}{\textcolor{red}{(q;q)_{2l-m}}}\sum_{k=0}^{2l}{2l\choose k}_q
\end{split}
\end{equation}
The correction factor is given by 
\begin{equation}
    \textcolor{red}{\frac{(q;q)_{2l}}{(q;q)_{2l-m}}}~.
\end{equation}
\subsubsection{The General $n$ Case}
We now show this for general $n$. The Macdonald index of the $(A_1,A_{2(n-1)})$ theory is given by \cite{Buican:2015tda}
\begin{equation}
    \CI_{\mr{M}}^{(A_1,A_{2(n-1)})}(q,t)=\sum_{N_1\geq N_2\geq\dots\geq N_{n-1}\geq 0}t^{N_1+N_2+\dots+N_{n-1}}\frac{q^{N_1^2+\dots+N_{n-1}^2}}{(q;q)_{N_1-N_2}\dots (q;q)_{N_{n-2}-N_{n-1}}(q;q)_{N_{n-1}}}~.
\end{equation}
The Macdonald index of the $\mathds{Z}_2$-projected free hypermultiplet is,
\begin{equation}
    \CI_{\mr{M}}^{\mr{HM}_2/\mathds{Z}_2}(q,t)=\sum_{l=0}^{\infty}\frac{t^l}{(q;q)_{2l}}\sum_{k=0}^{2l}{2l\choose k}_q~.
\end{equation}
Taking a product of these two indices, we obtain,
\begin{equation}\label{eq:prod_A1D2n+1HMZ2}
    \sum_{\substack{N_1\geq N_2\geq\dots\geq N_{n-1}\geq 0\\l\geq 0}}t^{N_1+\dots+N_{n-1}+l}\frac{q^{N_1^2+\dots+N_{n-1}^2}}{(q;q)_{N_1-N_2}\dots (q;q)_{N_{n-2}-N_{n-1}}(q;q)_{N_{n-1}}(q;q)_{2l}}\sum_{k=0}^{2l}{2l\choose k}_q~.
\end{equation} 
We claim that the correction factor is 
\begin{equation}
    \textcolor{red}{q^{(n-1)^2N_{n-1}^2-(n-1)N_{n-1}^2-2N_{n-1}(N_1+\dots+N_{n-2})}\frac{(q;q)_{N_{n-1}}(q;q)_{2l}}{(q;q)_{(n-2)N_{n-1}-N_1}(q;q)_{l-(n-1)N_{n-1}}}}~.
\end{equation}
To prove this, let us insert the correction factor in the expression \eqref{eq:prod_A1D2n+1HMZ2}. We get 
\begin{equation}
\begin{split}
\sum_{\substack{N_1\geq N_2\geq\dots\geq N_{n-1}\geq 0\\l\geq 0}}t^{N_1+\dots+N_{n-1}+l}\frac{q^{N_1^2+\dots+N_{n-1}^2+(n-1)^2N_{n-1}^2-(n-1)N_{n-1}-2N_{n-1}(N_1+\dots+N_{n-2})}}{(q;q)_{N_1-N_2}\dots (q;q)_{N_{n-2}-N_{n-1}}(q;q)_{(n-2)N_{n-1}-N_1}(q;q)_{l-(n-1)N_{n-1}}}\\\times\sum_{k=0}^{2l}{2l\choose k}_q~. 
\end{split}
\end{equation}
Note that 
\begin{equation}
\begin{split}
    N_1^2+\dots+N_{n-1}^2+(n-1)^2N_{n-1}^2-(n-1)N_{n-1}^2-2N_{n-1}(N_1+\dots+N_{n-2})\\=(N_1-N_{n-1})^2+\dots+(N_{n-2}-N_{n-1})^2+(n-1)^2N_{n-1}^2~.
\end{split}
\end{equation}
Put 
\begin{equation}
    m:=N_1+\dots+N_{n-1}+l,\quad i_{j+2}:=N_{j}-N_{j+1}~,\quad j=1,\dots,n-2~.
\end{equation}
Then we have 
\begin{equation}
    N_1+\dots N_{n-2}-(n-2)N_{n-1}=i_3+2i_4+\dots+(n-2)i_{n}~.
\end{equation}
Using this, we have 
\begin{equation}
    (n-1)N_{n-1}=m-i_3-2i_4-\dots-(n-2)i_{n}~.
\end{equation}
Using the fact that 
\begin{equation}
    N_1-N_{n-1}=i_3+\dots+i_n~,
\end{equation}
we get 
\begin{equation}
    (n-2)N_{n-1}-N_1=m-2i_3-3i_4-\dots-(n-1)i_n~.
\end{equation}
Using all these relations, the sum \eqref{eq:prod_A1D2n+1HMZ2} becomes 
\begin{equation}
\begin{split}
\sum_{m=0}^\infty t^m\sum_{l=0}^m\frac{1}{(q;q)_l}\sum_{\substack{i_3, i_4,\dots,i_n\geq 0}}&\frac{q^{Q(m-(n-1)i_n-(n-2)i_{n-1}-\dots-2i_3-l,i_3,\dots,i_n)}(q;q)_{l}}{(q;q)_{i_n}\dots (q;q)_{i_3}(q;q)_{m-(n-1)i_n-(n-2)i_{n-1}-\dots-2i_3-l}(q;q)_{2l+i_3+2i_4+\dots+(n-2)i_n-m}}\\&\times\sum_{k=0}^{2l}{2l\choose k}_q
\\&=\sum_{m=0}^\infty t^m\sum_{l=0}^m\frac{f^n_{m,l}(q)}{(q;q)_l}\sum_{k=0}^{2l}{2l\choose k}_q
\\&=\mathcal{I}^{(A_1,D_{2n+1})}_{\mr{M}}(q,t)~,
\end{split}
\end{equation}
where 
\begin{equation}
    Q(i_2,i_3,\dots,i_k):=i_k^2+(i_k+i_{k-1})^2+\dots+(i_k+\dots+i_3+i_2)^2~,
\end{equation}
as defined in \eqref{eq:Q_def}.
\\\\
\textbf{Acknowledgments.} We thank Matthew Buican for collaboration at the initial stages of the project. A.B. would like to thank Philip Argyres for helpful discussions. R.K.S. would like to thank Ken Ono for some useful correspondence and Anirudh Deb and Spencer Stubbs for discussions. R.T. would like to thank Heeyeon Kim for pointing out some interesting directions for future work. R.K.S. would like to thank ICTS, Bangalore for hospitality where this work was completed. The work of A.B. is supported by the US Department of Energy under grant DE-SC1019775. The work of C.B. was partially supported by funds from the Queen Mary University of London. The work of R.K.S. and R.T. is supported by the US
Department of Energy under grant DE-SC0010008. 

\begin{appendix}
\section{Some Special Functions}\label{app:special_func}
We make use of several special functions in defining and working with the superconformal index and its limits. Here we compile a list of functions used in this work. 
\paragraph{Plethystic Exponential.} This function is defined as:
\begin{equation}
    \mr{PE}[f(t)]=\exp\bigg(\sum_{r=1}^{\infty}\frac{f(t^r)}{r}\bigg)~.
\end{equation}
For example
\begin{equation}
    \mr{PE}[t]=\exp\bigg(\sum_{r=1}^{\infty}\frac{t^r}{r}\bigg)=\exp\bigg(-\log(1-t)\bigg)=\frac{1}{1-t}~.
\end{equation}
Following are two other properties of this function that will be useful for us:
\begin{equation}
    \begin{split}
        &\mr{PE}[f + g]=\mr{PE}[f]\mr{PE}[g]~, \\
        &\mr{PE}[-f]=1/\mr{PE}[f]~. \\
    \end{split}
\end{equation}
\paragraph{Elliptic Gamma Function.} This function is defined as:
\begin{equation}
    \Gamma_e(t):=\Gamma_e(t;p,q)=\prod_{m,r=0}^{\infty}\frac{1-p^{m+1}q^{r+1}t^{-1}}{1-p^mq^rt}=\mr{PE}\bigg[\frac{t-
    \frac{pq}{t}}{(1-p)(1-q)}\bigg]~.
\end{equation}
\paragraph{$q$-Pochhammer Symbol.} These are $q$-analogues of the Pochhammer symbols. This function is defined as:
\begin{equation}
    (z;q)_r=\prod_{k=0}^{r-1}(1-zq^k)~.
\end{equation}
We can similarly define the infinite $q$-Pochhammer symbol:
\begin{equation}
    (z;q)_{\infty}=\prod_{k=0}^{\infty}(1-zq^k)~.
\end{equation}
We also define 
\begin{equation}
    (a_1,a_2,\dots,a_r;q)_\infty:=(a_1;q)_\infty(a_2;q)_\infty\cdots (a_r;q)_\infty~.
\end{equation}
Some results about $q$-Pochhamer symbols that we use are recorded below, see \cite{gasper2004basic,andrews1998theory} for proofs.
\begin{align}
    (a ; q)_r &=\frac{(a ; q)_{\infty}}{\left(a q^r ; q\right)_{\infty}}~.
    \\
    (a ; q)_{r+k}&=(a ; q)_r\left(a q^r ; q\right)_k~.
    \\
    \left(a q^k ; q\right)_{r-k}&=\frac{(a ; q)_r}{(a ; q)_k}~.
    \\
    \frac{1}{(z;q)_\infty}&=\sum_{m\geq 0}\frac{z^m}{(q;q)_m}~.
    \\
    \sum_{r \in \mathds{Z}} q^{\frac{r(r+1)}{2}} z^r&=(q ; q)_{\infty}\left(-1/z ; q\right)_{\infty}(-z q ; q)_{\infty}~.
\end{align}

\paragraph{Functions involving $q$-analogues.}
We first define the $q$-analogues of integers:
\begin{equation}
    [r]_q:=\frac{1-q^r}{1-q}=1+q+q^2+\cdots +q^{r-1}, \quad \lim_{q\to 1}[r]_q=r~.
\end{equation}
We can then use this to define the $q$-analogues of factorials:
\begin{equation}
    [r]_q!:=\prod_{k=1}^{r}\frac{(1-q^k)}{(1-q)}=\frac{(q;q)_r}{(1-q)^r}
\end{equation}
We can use this to define $q$-binomial coefficients as follows:
\begin{equation}
    {r\choose k}_q:=\frac{[r]_q!}{[k]_q![r-k]_q!}=\frac{(q;q)_r}{(q;q)_k(q;q)_{r-k}}~.
\end{equation}
\paragraph{Basic Hypergeometric Series.}
The original Gaussian hypergeometric series $_2F_1(a,b;c;z)$ is given as:
\begin{equation}
    _2F_1(a,b;c;z)=\sum_{r=0}^{\infty}\frac{z^r}{r!}\frac{\Gamma(a)\Gamma(b)}{\Gamma(c)}~,
\end{equation}
and is well-studied in the context of the hypergeometric differential equation and Sturm-Liouville problems. 
This can be generalized to involve $q$-analogues of the integers as defined above: for $r>s+1,q\neq 0$,
\begin{equation}\label{eq:phirs_def}
\begin{aligned}  { }_r \phi_s\left(a_1, a_2, \ldots, a_r ; b_1, \ldots, b_s ; q, z\right) &\equiv{ }_r \phi_s\left[\begin{array}{c}a_1, a_2, \ldots, a_r \\ b_1, \ldots, b_s\end{array} ; q, z\right] \\ & :=\sum_{k=0}^{\infty} \frac{\left(a_1 ; q\right)_k\left(a_2 ; q\right)_k \cdots\left(a_r ; q\right)_k}{(q ; q)_k\left(b_1 ; q\right)_k \cdots\left(b_s ; q\right)_k}\left[(-1)^k q^{\binom{k}{2}}\right]^{1+s-r} z^k ~, 
\end{aligned}
\end{equation}
We will use the following relations \cite{gasper2004basic}:
\begin{enumerate}
\item The $q$-Vandermonde ($q$-Chu-Vandermonde) sums:
\begin{equation}
{ }_2 \phi_1\left(a, q^{-k} ; c ; q, q\right)=\frac{(c / a ; q)_k}{(c ; q)_k} a^k~.
\end{equation}
\item Reversing the order of summation,
\begin{equation}
{ }_2 \phi_1\left(a, q^{-k} ; c ; q, c q^k / a\right)=\frac{(c / a ; q)_k}{(c ; q)_k}~ .
\end{equation}
\item Heine's transform:
\begin{equation}
    { }_2 \phi_1(a, b ; c ; q, z)=\frac{(b, a z ; q)_{\infty}}{(c, z ; q)_{\infty}}{ }_2 \phi_1(c / b, z ; a z ; q, b)~.
\end{equation}
\item Expanding the \tsf{RHS} of Heine's transform gives:
\begin{equation}\label{eq:Heine_exp}
    { }_2 \phi_1(a, b ; c ; q, z)=\frac{(a z ; q)_{\infty}}{(z ; q)_{\infty}} \sum_{r=0}^{\infty} \frac{(a, c / b ; q)_r}{(q, c, a z ; q)_r}(-b z)^r q^{\binom{r}{2}}~.
\end{equation}
\end{enumerate}
\section{Conventions For Superconformal Index And Various Limits}\label{app:conventions}
There exist multiple conventions for the superconformal index in the literature, so in this appendix, we lay out our choice of conventions. See Table 1 of \cite{Gadde:2011uv} for other choices. \\
The superconformal index evaluates the partition function of a 4d $\mathcal{N}=2$ \tsf{SCFT} on $S^3\times S^1_{\beta}$ with certain fugacities for background fields turned on. In the Hamiltonian perspective, this involves four commuting supercharges and takes the form,
\begin{equation}
    \CI_{\mr{SCI}}(\rho,\sigma,\tau)=\text{Tr}_{\mathcal{H}(S^3)}[(-1)^F\sigma^{\frac{1}{2}\tilde{\delta}^1_{\dot{-}}}\rho^{\frac{1}{2}\tilde{\delta}^1_{\dot{+}}}\tau^{\frac{1}{2}\delta^2_{+}}e^{-\beta\delta^1_{-}}]~,
\end{equation}
where,
\begin{equation}
    \begin{split}
        &\delta^1_{-}=E-2R+r-2j~, \\
        &\tilde{\delta}^1_{\dot{+}}=E-2R-r+2\bar{j} ~,\\
        &\tilde{\delta}^1_{\dot{-}}=E-2R-r-2\bar{j}~, \\
        &\delta^2_{+}=E+2R+r+2j~. \\
    \end{split}
\end{equation}
We can change to a different set of variables as follows,
\begin{equation}
    \begin{split}
        &t:=\tau^2~, \\
        &p:=\sigma\tau~, \\
        &q:=\rho\tau~. \\
    \end{split}
\end{equation}
In these new variables, we have,
\begin{equation}
    \CI_{\mr{SCI}}(p,q,t)=\text{Tr}_{\mathcal{H}(S^3)}[(-1)^{F}p^{\frac{1}{2}\tilde{\delta}^1_{\dot{-}}}q^{\frac{1}{2}\tilde{\delta}^1_{\dot{+}}}t^{\frac{1}{4}(\delta^2_+-\tilde{\delta}^1_{\dot{+}}-\tilde{\delta}^1_{\dot{-}})}e^{-\beta\delta^1_-}]~.
\end{equation}
The index is automatically independent of $\beta$, i.e., it counts $1/8$-BPS operators that satisfy,
\begin{equation}
    E=2R-r+2j~.
\end{equation}
Manually substituting this condition into the index, we can write it in a simplified form,
\begin{equation}
    \CI_{\mr{SCI}}(p,q,t)=\text{Tr}_{\mathcal{H}(S^3)}[(-1)^Fp^{j-\bar{j}-r}q^{j+\bar{j}-r}t^{R+r}]~.
\end{equation}
Now, the superconformal index admits various special limits that give more refined counts of BPS operators. We describe some of these limits below, 
\paragraph{Macdonald limit ($p\to 0$).} This implements the additional condition that $j-\bar{j}-r=0$, which leads to the following conditions,
\begin{equation}
    \begin{split}
        &E-2R=j+\bar{j}~, \\
        &r=j-\bar{j}~. \\
    \end{split}
\end{equation}
Only $\hat{\mathcal{B}}_R$, $\bar{\mathcal{D}}_{R(j,0)}$ and $\hat{\mathcal{C}}_{R(j,\bar{j})}$ multiplets \footnote{We are using the notation of \cite{Dolan:2002zh}. See \cite{Cordova:2016emh} for another popular choice.} contribute to the Macdonald index $\CI_{\mr{M}}(q,t)$. We call this the Schur ring of operators, for reasons discussed next.

\paragraph{Schur limit ($t\to q$).} In this limit, the index automatically becomes independent of $p$, so it counts operators satisfying $E=2R-r+2j$ and $r=j-\bar{j}$. Written differently, the constraints are,
\begin{equation}
    \begin{split}
        &E-2R=j+\bar{j}~, \\
        &r=j-\bar{j}~. \\
    \end{split}
\end{equation}
so we see that we count the same operators as we did in the Macdonald limit, but the count is less refined. The contribution of these operators simplifies to,
\begin{equation}
    \CI_{\mr{S}}(q)=\text{Tr}_{\mathcal{H}_{\text{Schur}}}[(-1)^Fq^{E-R}]~.
\end{equation}
Similarly, the Macdonald index may be expressed as,
\begin{equation}
    \CI_{\mr{M}}(q,t)=\text{Tr}_{\mathcal{H}_{\text{Schur}}}[(-1)^Fq^{E-2R-r}t^{R+r}]=\text{Tr}_{\mathcal{H}_{\text{Schur}}}[(-1)^Fq^{2\bar{j}}t^{R+j-\bar{j}}]~.
\end{equation}
This last form is particularly convenient in studying the $(A_1,D_{2n+3})$ and $(A_1,A_{2n})$ theories, where all the Schur operators have $j=\bar{j}$, so that $r=0$. 
\paragraph{Hall-Littlewood limit ($p\to 0$, $q\to 0$).} This limit implements the additional conditions $j-\bar{j}-r=0$ and $j+\bar{j}-r=0$, which leads to the following conditions on quantum numbers,
\begin{equation}
    \begin{split}
        &E=2R+r~, \\
        &j=r~, \\
        &\bar{j}=0~. \\
    \end{split}
\end{equation}
Only $\hat{\mathcal{B}}_R$ and $\bar{\mathcal{D}}_{R(j,0)}$ contribute to this limit. 
The contribution of these operators may be expressed as,
\begin{equation}
    \CI_{\mr{HL}}(t)=\text{Tr}_{\mathcal{H}_{\mr{HL}}}[(-1)^Ft^{E-R}]~.
\end{equation}

\paragraph{Higgs branch Hilbert series.} This quantity is related to the Hall-Littlewood index, although it is not directly a limit of the $\mathcal{N}=2$ superconformal index. The operators that contribute to this quantity are $1/2$-BPS operators annihilated by $Q^1_{\alpha}$ and $\tilde{Q}^1_{\dot{\alpha}}$. This means that their quantum numbers satisfy,
\begin{equation}
    \begin{split}
        &j=\bar{j}=0~, \\
        &r=0~, \\
        &E=2R~.
    \end{split}
\end{equation}
Only the superconformal primaries of $\hat{\mathcal{B}}_R$ multiplets contribute to the Higgs branch Hilbert series. 
\section{Gr\"{o}bner Basis}\label{app:Grob}
In this appendix, we record some definitions from commutative algebra, specifically the Gr\"{o}bner basis of an ideal of a polynomial ring following \cite{loustaunau1994introduction}. This terminology and some standard results recorded here are used in Section \ref{sec:comp_Hilb_ser}.

Let $\mathds{F}$ be a field and $P:=\IF[x_1,\dots,x_n]$ be the polynomial ring in $n$-variables with coefficients in $\IF$. 
Recall that given $s$ polynomials $f_1,\dots,f_s$, the ideal generated by $I$ is given by 
\begin{equation}
    I=\langle f_1,\dots,f_s\rangle:=\left\{\sum_{i=1}^sg_if_i~:g_i\in P\right\}~.
\end{equation}
Our goal is to get a basis (or at least the dimension) of the vector space $P/I$. This is achieved by finding \textit{the} Gr\"{o}bner basis for $I$ which we now explain. The first step is to choose a \textit{monomial ordering} on the monomials in $P$. A monomial in $P$ is of the form 
\begin{equation}
    X^{\vec{a}}:=x_1^{a^1}\dots x_n^{a^r}~,\quad \vec{a}=(a^1,\dots,a^r)\in (\IZ_{\geq 0})^r~.
\end{equation}
A monomial ordering on $P$ is a total order $<$ on the set of monomials such that  
\begin{enumerate}[(1)]
    \item $X<Y\implies XZ<YZ$ for any monomials $X,Y,Z$.
    \item $1<X$ for any monomial $X$.
\end{enumerate}
There are several standard choices of monomial ordering, lexicographic (lex), degree lexicographic (deglex), degree reverse lexicographic (degrevlex) and so on. We will use the deglex ordering, which we now define.
\paragraph{Lexicographic ordering.} Choose an ordering on the variables. For example:
    \begin{equation}
        x_1>x_2>\dots x_r~.
    \end{equation}
Then to compare two monomials, lex first compares the exponent of $x_1$ (i.e., the largest variables), if equal then compares the exponent of $x_2$ (i.e., the second largest variables) and so forth. For example  \begin{equation}
    x_1^2>x_1 x_2>x_1>x_2^2>x_2>1~.
\end{equation}   
\paragraph{Degree lexicographic ordering.}
Deglex first compares the degree of the monomials. If the degrees are equal then, deglex compares the monomials using lex. For example
\begin{equation}
    x_1^2>x_1 x_2>x_2^2>x_1>x_2>1~.
\end{equation}
Once we have chosen a monomial ordering, we can rearrange any polynomial in the order of decreasing order. 
\begin{defn}
For any $0\neq f\in P$ of degree $r$ given by 
\begin{equation}
    f=\sum_{i=1}^rc_iX^{\vec{a}_i}~,
\end{equation}
such that $X^{\vec{a}_1}>X^{\vec{a}_2}>\dots>X^{\vec{a}_r}$, we define 
\begin{enumerate}[(1)]
    \item $\mathsf{LM}(f)=X^{a_1}$, the leading monomial of $f$,
\item $\mathsf{LC}(f)=c_1$, the leading coefficient of $f$,
\item $\mathsf{LT}(f)=c_1 X^{\alpha_1}$, the leading term of $f$.
\end{enumerate}
\end{defn}
\begin{defn}
\begin{enumerate}[(1)]
    \item Given polynomials $f,g,h\in P$, we say that $f$ is reducible to $h$ by $g$ if $\mr{LM}(g)$ divides a nonzero term $X$ of $f$ and 
    \begin{equation}
        h=f-\frac{X}{\mr{LM}(g)}g~.
    \end{equation}
    We symbolically write the reduction as 
    \begin{equation}
        f\overset{g}{\longrightarrow} h~.
    \end{equation}
    \item Given polynomials $f, h \in R$ and a set $G=\left\{g_1, \ldots, g_s\right\} \subset R$ of non-zero polynomials, we say that $f$ is reducible to $h$ modulo $G$, denoted
\begin{equation}
f \xrightarrow{G} h~,
\end{equation}
if there exists a subset $\{g_{i_1}, \ldots, g_{i_k}\} \subseteq G$ and a sequence of polynomials $h_1, \ldots, h_{k-1}$ such that
\begin{equation}
f \xrightarrow{g_{i_1}} h_1 \xrightarrow{g_{i_2}} h_2 \xrightarrow{g_{i_3}} \cdots \xrightarrow{g_{i_{k-1}}} h_{t-1} \xrightarrow{g_{i_k}} h~ .
\end{equation}
\item A polynomial $f$ is called reduced with respect to $G$ if it cannot be reduced modulo $G$. That is, no term of $f$ is divisible by any one of the $\mr{LM}\left(g_i\right)$.
\item The remainder of a polynomial $f$ with respect to $G$ is a polynomial $h$ such that $f$ is reducible to $h$ modulo $G$ and $h$ is reduced with respect to $G$. 
\end{enumerate}    
\end{defn}
\begin{defn}
A set of non-zero polynomials $G=\left\{g_1, \ldots, g_s\right\}$ is called a Gr\"{o}bner basis for an ideal $I\subset P$ if for all nonzero $f \in I$, $\mr{LM}\left(g_i\right)$ divides $\mr{LM}(f)$ for some $g_i \in G$.
\end{defn} 
\begin{thm}\emph{\cite[Theorem 1.6.7]{loustaunau1994introduction}}\label{prop:remain_grob}
$G$ is a Gr\"{o}bner basis if and only if for all $f \in P$ the remainder of reduction of $f$ by $G$ is unique.    
\end{thm}
Let $r_G(f)\in P$ denote the remainder of the reduction of a polynomial $f$ by $G$. Then Theorem \ref{prop:remain_grob} implies that $f\equiv g\bmod I$ if and only if $g_G(f)=r_G(g)$. This gives us a basis for the $\IF$-vector space $P/I$:
\begin{equation}
    \CB_{P/I}:=\{\text{monomial}~~X^{\vec{a}}:\mr{LM}(g_i)\nmid X^{\vec{a}}\text{ for all }g_i\in G\}~.
\end{equation}
We now discuss an algorithm to construct a Gr\"{o}bner basis for an ideal generated by a set of polynomials. 
\begin{defn}
Given non-zero $f, g \in P$, the $\CS$-polynomial of $f$ and $g$ is defined as
\begin{equation}
\CS(f, g)=\frac{L}{\mathsf{LT}(f)} f-\frac{L}{\mathsf{LT}(g)} g~,
\end{equation}
where $L=\mr{LCM}(\mr{LM}(f), \mr{LM}(g))$.    
\end{defn}
The following theorem is due to Buchberger \cite{Buchberger1985}:
\begin{thm}\emph{\cite[Theorem 1.7.4]{loustaunau1994introduction}}\label{thm:S-pol_grob}
A set of non-zero polynomials $G=\left\{g_1, \ldots, g_s\right\}$ is a Gr\"{o}bner basis for an ideal $I$ if and only if for all $i \neq j$,
\begin{equation}
\CS\left(g_i, g_j\right) \xrightarrow{G} 0~ .
\end{equation}    
\end{thm}
Theorem \ref{thm:S-pol_grob} gives us an algorithm to construct a Gr\"{o}bner basis for $I$. 
\paragraph{Buchberger's algorithm for Gr\"{o}bner basis.} Let $I=\langle f_1,\dots,f_s\rangle\subset P$ where each $f_i\neq 0$. To construct a Gr\"{o}bner basis for $I$, we follow the steps below:
\begin{enumerate}[(1)]
    \item Start with $G:=\{f_1,\dots,f_s\}$. If $s=1$, then $G$ is a Gr\"{o}bner basis. If $s>1$, define $A:=\{(f_i,f_j):i<j\}$.
    \item  For $(f,g)\in A$, calculate $\CS(f,g)$ and reduce modulo $G$: $\CS\left(f, g\right) \xrightarrow{G} r$. If $r\neq 0$ then add $f_{s+1}:=r$ to the set $G$ and the pairs $(f_i,f_{s+1})$ to $A$.
    \item Repeat steps (1) and (2) with the new $G$ and $A$ until $\CS\left(f_i, f_j\right) \xrightarrow{G} 0$ for all $(f_i,f_j)\in A$.
\end{enumerate}
It is a consequence of the Hilbert basis theorem that this algorithm terminates for polynomial rings. We refer the reader to \cite{loustaunau1994introduction} for further details on this point.
\section{Proof Of Formula \eqref{eq: hilbert_ser_Rinfty_gen} For Hilbert Series}\label{app:Hilb_ser_count}
In this appendix, we prove the formula \eqref{eq: hilbert_ser_Rn_gen} and \eqref{eq: hilbert_ser_Rinfty_gen} using explicit counting of monomials.  

Recall that given a Gr\"{o}bner basis of $I_s$, a canonical basis for $R_s/I_s$ is given by 
\begin{equation}
    \CB_{R_s}:=\{\text{monomial}~~X_1^{\vec{a}_1}X_2^{\vec{a}_2}\cdots X_N^{\vec{a}_N}:g\nmid X_1^{\vec{a}_1}X_2^{\vec{a}_2}\cdots X_N^{\vec{a}_N}\text{ for all }g\in \mr{LM}(G_s)\}~.
\end{equation}
This gives us the following algorithm for the computation of $\mr{dim}\left(R_s\right)_{\ell, m}$:
\begin{enumerate}
    \item Count the number of monomials in 
    \begin{equation}
        \C[x^{(1)}_{d_1-1},\dots,x^{(1)}_s,x^{(2)}_{d_2-1},\dots,x^{(2)}_s,\dots,x^{(N)}_{d_N-1},\dots,x^{(N)}_s]
    \end{equation}
    of weight $(\ell,m)$. 
    \item Subtract the number of monomials of weight $(\ell,m)$ that can be obtained by multiplying another monomial (necessarily of weight $(\ell',m')$ with $\ell'<\ell$ and $m'<m$) to any monomial in $G_s$.
\end{enumerate}
Let $M[\ell,m]$ be the number of monomials in variables $\{x^{(1)}_{i_1},\dots,x^{(N)}_{i_N}:i_k\geq d_{k}-1\}$ of weight $(\ell,m)$. We also set 
\begin{eqsp}
    M[\ell,m]=0\quad\text{for}~~\ell<0~\text{or}~m<0~.
\end{eqsp}
We call such monomials \textit{free words}. Then the generating function for $M[\ell,m]$ is given by
\begin{equation}\label{eq:gen_free_words}
    \sum_{\ell,m\geq0}M[\ell,m]q^\ell t^m=\prod_{i=1}^N\prod_{j\geq 0}\left(1-q^{(d_i-1)+j}t\right)^{-1}~.
\end{equation}
Then we can write 
\begin{equation}
\mr{dim}\,(R_s)_{\ell,m}=M[\ell,m]-C[\ell,m]~,    
\end{equation}
where $C[\ell,m]$ is the number of monomials we get from step 2. 
To implement step 2, we count all possible ways of constructing a monomial of weight $(\ell,m)$ by multiplying free words of weight $(\ell',m')~,\ell'<\ell,m'<m$ with monomials in $\mr{LM}(G_n)$ of weight $(\ell-\ell',m-m')$. Let us define 
\begin{equation}
\mr{LM}(G_s)_{\ell,m}:=\{g\in\mr{LM}(G_s):\mr{wt}_1(g)\leq \ell,\mr{wt}_2(g)\leq m\}~.    
\end{equation}
For $g\in \mr{LM}(G_n)_{\ell,m}$, define  
\begin{equation}
C_g[\ell,m]:=\{g\cdot X_1^{\vec{a}_1}X_2^{\vec{a}_2}\cdots X_N^{\vec{a}_N}:\mr{wt}^{(2)}(g\cdot X_1^{\vec{a}_1}X_2^{\vec{a}_2}\cdots X_N^{\vec{a}_N})=(\ell,m) \}~.
\end{equation}
If $\mr{wt}^{(2)}(g)=(\ell',m')$, then clearly 
\begin{equation}
    |C_g[\ell,m]|=M[\ell-\ell',m-m']~.
\end{equation}
Then we can write 
\begin{equation}
    C[\ell,m]=\left|\bigcup_{g\in\mr{LM}(G_s)}C_g[\ell,m]\right|
\end{equation}
Note that 
\begin{equation}
    \bigcap_{i=1}^kC_{g_i}[\ell,m]\neq \emptyset\iff \mr{wt}^{(2)}(\mr{LCM}(g_1,\dots,g_k))=(\ell',m')\text{ with }\ell'\leq\ell,m'\leq m~.
\end{equation}
In this case
\begin{equation}
\left|\bigcap_{i=1}^kC_{g_i}[\ell,m]\right|=M[\ell-\mr{wt}_1(\mr{LCM}(g_1,\dots,g_k)),m-\mr{wt}_2(\mr{LCM}(g_1,\dots,g_k))]~.    
\end{equation}
Then using the fact that for sets $A_1,\dots,A_M$, we get
\begin{equation}
    \left|\bigcup_{i=1}^MA_i\right|=\sum_{k=1}^M(-1)^{k+1} \sum_{1 \leq i_1<\cdots<i_k \leq M}\left|A_{i_1} \cap \cdots \cap A_{i_k}\right|~,
\end{equation}
we get 
\begin{equation}
\begin{split}
C[\ell,m]&= \sum_{S\subseteq \mr{LM}(G_s)_{\ell,m}}(-1)^{|S|+1}\left|\bigcap_{g\in S}C_{g}[\ell,m]\right|
\\
&=\sum_{S\subseteq \mr{LM}(G_s)_{\ell,m}}(-1)^{|S|+1}M[\ell-\mr{wt}_1(\mr{LCM}(S)),m-\mr{wt}_2(\mr{LCM}(S))]~,
\end{split}
\end{equation}
where $|S|$ denotes the number of monomials in $S$, $\mr{LCM}(S)$ is least common multiple of the monomials in $S$. This gives us the final formula for $\mr{dim}\,(R_s)_{\ell,m}$:
\begin{equation}
    \mr{dim}\,(R_s)_{\ell,m}=M[\ell,m]+\sum_{S\subseteq \mr{LM}(G_s)_{\ell,m}}(-1)^{|S|}M[\ell-\mr{wt}_1(\mr{LCM}(S)),m-\mr{wt}_2(\mr{LCM}(S))]~.
\end{equation}
Substituting this in the formula
\begin{equation}
H_{R_s}(q, t)=\sum_{\ell=0}^s \sum_{m \geq 0} \mr{dim}\left(R_s\right)_{\ell, m} q^{\ell} t^m~,
\end{equation}
and using \eqref{eq:gen_free_words}, we obtain the formula 
\begin{equation}
\label{eq: hilbert_ser_Rn'}
H_{R_s}(q,t) = \mr{Ser}_{q=0,t=0}\frac{\sum_{S\subseteq \mr{LM}(G_s)}(-1)^{|S|}(q,t)^{\mr{wt}^{(2)}(\mr{LCM}(S))}}{\prod_{i=1}^N\prod_{j\geq 0}\left(1-q^{(d_i-1)+j}t\right)} \quad (\mr{mod} \;q^{s+1})~.
\end{equation}
The full Hilbert series is then given by taking $s\to \infty$:
\begin{equation}
\label{eq: hilbert_ser_Rinfty'}
H_{R_\infty}(q,t) = \mr{Ser}_{q=0,t=0}\frac{\sum_{S\subseteq \mr{LM}(G_\infty)}(-1)^{|S|}(q,t)^{\mr{wt}^{(2)}(\mr{LCM}(S))}}{\prod_{i=1}^N\prod_{j\geq 0}\left(1-q^{(d_i-1)+j}t\right)}~,
\end{equation}
where 
\begin{equation}
    \mr{LM}(G_\infty):=\lim _{\leftarrow } \mr{LM}(G_s)~.
\end{equation}
\section{Hilbert Series For The Arc Space Of $A(\CV_{(A_1,A_4)})$}\label{app:Hilb_A1A4}
The Zhu's $C_2$-algebra for the chiral algebra of $(A_1,A_4)$ Argyres-Douglas theory is given by \eqref{eq:zhuA1A2}:
\begin{equation}
R_{(A_1,A_4)}:=A\left(\mathcal{V}_{\left(A_1, A_4\right)}\right) \cong \mathds{C}[J]/\langle J^3\rangle~.
\end{equation}
We compute the Hilbert series of its arc space 
$H_{(A(\mathcal{V}_{(A_1,A_4)})_\infty}(q,t)$ using technique discussed in Section \ref{sec:comp_Hilb_ser}. Let us  order the variables in the deglex ordering:
\begin{equation}
    J_1>J_2>\cdots
\end{equation}
Then from Mathematica computations (checked up to $m \leq 10$) of Gr\"{o}bner basis, we have the following conjecture: the
set of monomials of weight $(\ell,m)$ constituting a Gr\"obner basis for $\langle J(x)^2\rangle$ is given by
\begin{equation}
    \mathcal{G} (\ell,m) := \begin{cases}
\bigcup_{i = 1}^{(m-2)/2} G^{2i, m/2 - i}(\ell,m), & m \text{ is even},\\
\bigcup_{i = 1}^{(m-1)/2} G^{2i -1 , (m+1)/2 -i}(\ell,m), & m \text{ is odd},
\end{cases}
\end{equation}
where $G^{c,s}(\ell,m)$ is the set of monomials of weight $(\ell,m)$ given recursively by
\begin{equation}\label{eq:Grob_recusive_A1A4}
    G^{c,s}(\ell,m) = \begin{cases}
\{J_s^2\}\cdot F_s(\ell-2s, m -2), & c = 1, \\
\{J_s J_{s+1}\}\cdot F_{s+1}(\ell -2s -1, m -2), & c= 2,3,\\
\bigcup_{i = 1}^{c/2 -1} \{J_s\}\cdot G^{2i -1, s + c/2 +1 -i}(\ell-s , m-1), & c>3 \text{ and even}, \\
\bigcup_{i = 1}^{(c-3)/2}\{J_s\}\cdot G^{2i , s + (c+1)/2 -i}(\ell-s, m-1), & c>3 \text{ and odd,}
\end{cases}
\end{equation}
where $F_s(\ell,m)$ is defined as the set of monomials of weight $(\ell,m)$ freely generated by $\{J_i: \text{for } i\geq s \}$.
Let us write the Hilbert series in the form:
\begin{equation}
    H_{(A(\mathcal{V}_{(A_1,A_4)})_\infty}(q,t)
    \sum_{m = 0}\mathcal{I}_m^{(A_1, A_4)}(q)t^m.
\end{equation}
Let us analyze the Hilbert series to the first orders in $t$:
\begin{enumerate}
    \item $m<3$:
    Since there is no constraints till $O(t^3)$, words are freely generated by $j_i$ for $i \geq 1.$ Then using \eqref{eq:free_gen_ser_A1A4}, the Hilbert series up to $O(t^3)$ is thus given by 
    \begin{equation}
        H_{(A(\mathcal{V}_{(A_1,A_4)})_\infty}(q,t)=\mr{Ser}_{q=0,t=0}\prod_{j=1}^\infty\frac{1}{(1-q^it)}~~(\mr{mod}~t^3)= 1 + \mathcal{I}_1^1(q)t + \mathcal{I}^1_2 (q)t^2 + \mathcal{O}(t^3)~,
    \end{equation}
    where $\mathcal{I}_m^s$ is defined in \eqref{def:Ism} with its explicit expression given in \eqref{eq:Ism}.
    \item $m=3$:
    The set of constraints from Gr\"{o}bner basis is given by:
    \begin{equation}
        \mathcal{G}(\ell,3) =G^{1,1}(\ell,3) = \{j_1^2\}\cdot F_1(\ell-2,1)~.
    \end{equation}
    Thus the monomials which are freely generated by $\{j_\ell: \ell\geq1\}$ and not divisible by elements inside $\mathcal{G}(\ell,3)$ are
    \begin{equation}
        \{j_1\}\cdot F_2(\ell-1,2) \bigcup F_2(\ell,3)~,
    \end{equation}
   This set generates the homogeneous subspace in $(A(\mathcal{V}_{(A_1,A_4)})_\infty$ of weight $(\ell,m =3)$. Therefore, we have
    \begin{equation}
        \mathcal{I}^{(A_1,A_4)}_3 = q \mathcal{I}_{2}^{2} + \mathcal{I}_3^2~.
    \end{equation}
    \item $m =4$:
    The set of constraints from Gr\"{o}bner basis reads
    \begin{equation}
        \mathcal{G}(\ell,4) =G^{2,1}(\ell,4) = \{j_1j_2\}\cdot F_2(\ell-3,2)~.
    \end{equation}
    Thus the words which generate the homogeneous subspace in $(A(\mathcal{V}_{(A_1,A_4)})_\infty$ of weight $(\ell,m =4)$ are given by 
    \begin{equation}
        \{j_1\}\cdot F_3(\ell-1,3) \bigcup F_2(\ell,4)~.
    \end{equation}
   This gives us
\begin{equation}
\mathcal{I}^{(A_1,A_4)}_4 = q \mathcal{I}_{3}^{3} + \mathcal{I}_4^2~.
\end{equation}
\item
The set of constraints from Gr\"{o}bner basis reads
\begin{equation}
    \mathcal{G}(\ell,5) =G^{1,2}(\ell,5)\bigcup G^{3,1}(\ell,5) = \{j_2^2\}\cdot F_2(\ell-4,3)\bigcup\{j_1j_2\}\cdot F_2(\ell-3,3)~.
\end{equation}
The words which generates the homogeneous subspace in $(A(\mathcal{V}_{(A_1,A_2)})_\infty$ of weight $(\ell,m =5)$ are thus given by
\begin{equation}
    \{j_2\}\cdot F_3(\ell-2,4)\bigcup \{j_1\}\cdot F_3(\ell-1,4) \bigcup F_3(\ell,5)~.
\end{equation}
Thus we have
\begin{equation}
    \mathcal{I}_5^{(A_1,A_4)} = q^2 \mathcal{I}_4^3 + q \mathcal{I}_4^3 + \mathcal{I}_5^3~.
\end{equation}
\end{enumerate}
For $m\geq 6$, similar analysis can be performed using the general recursive form of the Gr\"{o}bner basis given in \eqref{eq:Grob_recusive_A1A4}. The result is that we can write
\begin{equation}
    \mathcal{I}_m^{(A_1,A_4)} = \begin{cases}\mathcal{I}_m^{m/2}+
\sum_{i=1}^{(m-2)/2}I^{2i, m/2 -i}_m, &m\text{  is even,}\\
\mathcal{I}_m^{(m+1)/2} +\sum_{i =1}^{(m-1)/2} I_m^{2i -1,(m+1)/2 -i}, & m \text{ is odd},
\end{cases}
\end{equation}
where  $I^{c,s}_m$ is determined by the following recursion relation:
\begin{equation}
I_m^{c, s}= \begin{cases}q^s \mathcal{I}_{m-1}^{s+1}, & c=1, \\ q^s \mathcal{I}_{m-1}^{s+2}, & c=2,3, \\ q^s \mathcal{I}_{m-1}^{s+c+1}+q^s \sum_{i=1}^{c / 2-1} I_{m-1}^{2 i-1, s+c / 2+1-i}, & c>3 \text { and even }, \\ q^s \mathcal{I}_{m-1}^{s+(c+1) / 2}+q^s \sum_{i=1}^{(c-3) / 2} I_{m-1}^{2 i, s+(c+1) / 2-i}, & c>3 \text { and odd~. }\end{cases}    
\end{equation}
\section{Evaluation Of Hilbert Series \eqref{eq: hilbert_ser_Rinf} At $t^3$}\label{app:Hilb_ser_t3}
In this appendix, we give the detailed calculation of the Hilbert series \eqref{eq: hilbert_ser_Rinf} at $t^3$. We want to compute the Hilbert series 
\begin{equation}
\label{eq: hilbert_ser_Rinf_app}
H_{R_\infty}(q,t) = \mr{Ser}_{q=0,t=0}\frac{1+\sum_{S\subseteq \mr{LM}(G_\infty)}(-1)^{|S|}(q,t)^{\mr{wt}^{(2)}(\mr{LCM}(S))}}{(1-t)^3 \prod_{i=1}^{\infty}\left(1-q^i t\right)^4}~,
\end{equation}
where $\mr{LM}(G_s)$ is given by
\begin{equation}
\label{eq:Gr\"{o}bner_basis_Gn_app}
\mr{LM}\left(G_s\right)=\bigcup_{\substack{0 \leq \ell \leq s \\ 2 \leq m \leq \ell+2}} I_{\ell, m}~,
\end{equation}
where $I_{\ell,m}$ is a homogeneous set of weight $(\ell,m)$ given by
\begin{equation}
I_{\ell,m} := \left(\{X_{m-2}\}\cdot F_z(\ell-m+2, m-1)\right)\bigcup \left(\{J_{m-1}\}\cdot F(\ell-m+1, m-1)\right)~.
\end{equation}
Let us write the numerator of \eqref{eq: hilbert_ser_Rinf_app} as
\begin{equation}
\sum_{S \subseteq \mr{L M}\left(G_{\infty}\right)}(-1)^{|S|}(q, t)^{\mr{wt}^{(2)}(\mr{LCM}(S))}=: \sum_{\substack{\ell \geq 0\\ m\geq2}} C[\ell,m]q^\ell t^m~.
\end{equation}
there are several non-trivial contributions to $C[\ell, m]$. The first contribution comes from the subsets $S\subset I_{\ell,m}\subseteq \mr{LM}(G_\ell)$. Note that the sets $S\subset I_{\ell,m}$ contributing to $C[\ell,m]$ are necessarily singletons since \tsf{LCM} of any two elements of $I_{\ell,m}$ has weight $(\ell',m')$ with $m'>m$. For this reason, the constraints obtained from $S\subset I_{\ell,m}$ are called \textit{first-depth constraints}. Other contributions come from subsets of 
\begin{equation}
\bigcup_{\substack{\ell'\leq \ell\\m'<m}}I_{\ell',m'}~.    
\end{equation}
If $|S|=k$, then the constraints obtained from $S$ are called \textit{depth-k constraints}. A good way to organize these constributions is as follows: given $k$ sets $S_1,S_2,\dots S_k$ of monomials, define
\begin{equation}
    \mr{LCM}(S_1,S_2,\dots S_k) := \sum_{\substack{s_1 \in S_1,\dots, s_k\in S_k\\
    s_i \neq s_j}} 
    \mr{LCM}(s_1,s_2,\dots, s_k)~,
\end{equation}
and define $\mr{LCM}^{(\ell,m)}(S_1,S_2,\dots S_k)$ to be the maximal homogeneous polynomial which is a summand in $\mr{LCM}(S_1,S_2,\dots S_k)$ of weight $(\ell,m)$. Then the number of depth-$k$ contributions to $C[\ell,m]$ is given by 
\begin{equation}
    \sum_{S_1,\dots,S_k\subseteq \mr{LM}(G_\ell)}(-1)^k|\mr{LCM}^{(\ell,m)}(S_1,S_2,\dots, S_k)|~,
\end{equation}
where for a homogeneous polynomial $f$, $|f|$ is the number of terms in the polynomial counted with multiplicity\footnote{This means that we write the polynomial such that each monomial has coefficient 1 and then count the number of terms.}. Thus, we can write 
\begin{equation}
    C[\ell,m]=\sum_{k=1}^{|\mr{LM}(G_\ell)|}\sum_{S_1,\dots,S_k\subseteq \mr{LM}(G_\ell)}(-1)^k|\mr{LCM}^{(\ell,m)}(S_1,S_2,\dots S_k)|~.
\end{equation}
\par
We need to obtain all the coefficients $C[\ell,m=3]$ for $\ell\geq 0$. From \eqref{eq:Gr\"{o}bner_basis_Gn_app} with $s=0$, it is not hard to see  that
\begin{equation}
    C[\ell=0, m =3] = 0~.
\end{equation}
For $\ell>0$, we list the constraints contributing to $C[\ell, m =3]$ below. 

\begin{enumerate}
\item $\boldsymbol{C[\ell=1, m=3].}$

From \eqref{eq:Gr\"{o}bner_basis_Gn_app} with $s=1$, the depth-1 constraint is given by
\begin{equation}
    \mr{LCM}^{(1,3)}(\mr{LM}(G_1)) = X_1 Z_0^2~.
\end{equation}
The non-trivial contributions for the depth-2 constraints
are given by 
\begin{equation}
\begin{aligned}
    \mr{LCM}^{(1,3)}(\{X_0 Z_0\}, \{X_0 Z_1\}) &= X_0 Z_0 Z_1~,\\
    \mr{LCM}^{(1,3)}(\{X_0 Z_0\}, \{J_1\}\cdot F(0,1))&= 2 J_1 X_0 Z_0~, \\
    \mr{LCM}^{(1,3)}(\{J_1\}\cdot F(0,1), \{J_1\} \cdot F(0,1)) &= 
    J_1(X_0Y_0+Y_0Z_0+X_0Z_0)~.
\end{aligned}
\end{equation}
Depth-3 contribution are given by
\begin{equation}
    \mr{LCM}^{(1,3)}(\{X_0Z_0\},\{J_1\}\cdot F(0,1), \{J_1\}\cdot F(0,1)) = 
    J_1X_0Z_0~.
\end{equation}
Then we have
\begin{equation}
    C[\ell=1,m=3]=-1+(1+2+3)-1=4~.
\end{equation}
\item $\boldsymbol{C[\ell=2, m = 3].}$

The depth-1 contribution comes from
\begin{equation}
\mr{LCM}^{(2,3)}\left(\mr{L M}\left(G_2\right)\right)=X_1\overline{ F_z(1,2)}+J_2 \overline{F(0,2)}~,
\end{equation}
where given a set $S$, we have defined 
\begin{equation}
    \overline{S}:=\sum_{s\in S}s~.
\end{equation}

The depth-2 constraints are given by:
\begin{equation}
\begin{aligned}
\mr{LCM}^{(2,3)}(\{X_0 Z_0\},\{X_0 Z_2\}) &= X_0Z_0 Z_2~,\\
\mr{LCM}^{(2,3)}(\{X_0 Z_0\},\{J_2\}\cdot F(0,2)) &= J_2 X_0 Z_0~,\\
\mr{LCM}^{(2,3)}(\{X_0 Z_1\},\{J_1\}\cdot F(0,1)) &= J_2 X_0 Z_0~,\\
\mr{LCM}^{(2,3)}(\{X_0 Z_1\},\{J_1\}\cdot F(1,1)) &= J_1 X_0 Z_1~,\\
\mr{LCM}^{(2,3)}(J_1F(0,1),\{J_1\}\cdot F(1,1)) &=J_1 \overline{F(1,2)}~.
    \end{aligned}
\end{equation}
The depth-3 constraints are given by:
\begin{equation}
    \mr{LCM}^{(2,3)}
    (\{X_0 Z_1\}, \{J_1\}\cdot F(0,1), \{J_1\}\cdot F(1,1)) =
    X_0 Z_1 J_1~.
\end{equation}
Then we have 
\begin{equation}
    C[\ell=2, m=3] = -(|F_z(1,2)|+ |F(0,2)|)+ (4 + |F(1,2)|) -1 = 8~, 
\end{equation}
where we used  \eqref{eq:num_Fl2} and 
\begin{equation}
\label{eq:num_Fzl2}
    |F_z(\ell,2)| =\begin{cases}
        \ell/2 + 1, & \ell \text{ even},\\
        (\ell+1)/2, 
        & \ell \text{ odd}~.
    \end{cases}
\end{equation}
\item $\boldsymbol{C[\ell, m=3]}$ \textbf{for} $\boldsymbol{\ell>2.}$

The depth-1 constraints are given by 
\begin{equation}
    \mr{LCM}^{(\ell,3)}(\mr{LM}(G_\ell)) = X_1 \overline{F_z(\ell-1,2)}+ J_2 \overline{F(\ell-2,2)}~.
\end{equation}
The depth-2 constraints are given by
\begin{equation}
\begin{aligned}
\mr{LCM}^{(\ell,3)}(\{X_0 Z_{\ell-1}\},\{J_1 Z_{\ell-1}\}) &= J_1 Z_{\ell-1}X_0~, \\
    \mr{LCM}^{(\ell,3)}(\{X_0 Z_{\ell-1}\}, \{J_1 X_0\}) &= J_1 Z_{\ell-1} X_0~, \\
    \mr{LCM}^{(\ell,3)}(\{X_0 Z_{\ell-2}\}, \{J_2 X_0Z_{\ell-2}\}) &= J_2 Z_{\ell-2}X_0~ ,\\
    \sum_{i=0}^{\lfloor \ell/2\rfloor}\mr{LCM}^{(\ell,3)}(\{X_0 Z_i\},\{X_0 Z_{\ell-i}\}) &= \sum_{i=0}^{\lfloor \ell/2\rfloor} X_0 Z_i Z_{\ell-i}~,\\
    \sum_{i=0}^{\lfloor \ell/2\rfloor} \mr{LCM}^{(\ell,3)}(J_1 F(i,1), J_1F(\ell-1-i,1)) &= 
    \begin{cases}
    J_1 \overline{F(\ell-1,2)}, & \ell \text{ even}~,\\
    J_1 \overline{F(\ell-1,2)}-J_1\big(X^2_{\frac{\ell-1}{2}}+Y^2_{\frac{\ell-1}{2}}+Z^2_{\frac{\ell-1}{2}}+ J^2_{\frac{\ell-1}{2}}\big), & \ell \text{ odd}~,
    \end{cases}\\
    \mr{LCM}^{(\ell,3)}(J_1 F(\ell-3,1), J_2 F(\ell-2,2)) &= J_1 J_2 \overline{F(\ell-3,1)}~,\\
    \mr{LCM}^{(\ell,3)}(J_1 F(\ell-1,1), J_2 F(\ell-2,2)) &= J_1 J_2 \overline{F(\ell-3,1)}~.
\end{aligned}
\end{equation}
The depth-3 constraints are given by:
\begin{equation}
    \begin{aligned}
        \mr{LCM}^{(\ell,3)}(\{X_0 Z_{\ell-1}\},\{J_1 Z_{\ell-1}\},\{J_1 X_0\}) &= J_1 Z_{\ell-1}X_0~,\\
        \mr{LCM}^{(\ell,3)}(J_1 F(\ell-1,1), J_2 F(\ell-2,2), J_1 F(\ell-3,1)) &= J_1 J_2 \overline{F(\ell-3,1)}~.
    \end{aligned}
\end{equation}
Then 
\begin{equation}
    \begin{aligned}C[\ell,m=3]&=\begin{cases}
        -|F_z(\ell-1,2)|-|F(\ell-2,2)| + |F(\ell-3,1)| & l \text{ even},\\
        \quad+|F_z(\ell,2)|+|F(\ell-1,2)| + 1,\\
        -|F_z(\ell-1,2)|-|F(\ell-2,2)| + |F(\ell-3,1)| & l \text{ odd},\\
        \quad+|F_z(\ell,2)|+|F(\ell-1,2)| -2,
    \end{cases}\\
    &= \begin{cases}
        11, & \ell = 3~,\\
        12, & \ell > 3~,
    \end{cases}
\end{aligned}
\end{equation}
where we used \eqref{eq:num_Fl2} and 
\eqref{eq:num_Fzl2}.
\end{enumerate}
Thus we get
\begin{equation}
\sum_{\substack{S \subseteq \mr{L M}\left(G_{\infty}\right) \\ \mr{wt}_2(\mr{LCM}(S))=3}}(-1)^{|S|}(q, t)^{\mathsf{wt}^{(2)}(\mr{LCM}(S))}=t^3\left(4 q+8 q^2+11 q^3+12 \sum_{\ell \geq 4} q^\ell\right)~.
\end{equation}
The Hilbert series up to $t^3$ is then given by
\begin{equation}\label{eq:Hil_ser_modt4}
\begin{aligned}
H_{R_\infty}(q,t) &=\mr{Ser}_{q=0,t=0} \frac{1-t^2(1+4q + 5\sum_{\ell=2}^\infty q^\ell)+t^3(4q + 8 q^2 + 11 q^3 + 12 \sum_{\ell\geq 4}q^\ell)}{(1-t)^3 \prod_{i=1}^{\infty}\left(1-q^i t\right)^4}\quad (\mr{mod} \;t^4)~.
\end{aligned}
\end{equation}
Let us focus on $t^3$ term
\begin{equation}
H_{R_\infty}(q,t)|_{t^3} = \sum_{\ell=0}^\infty \mr{dim}(R_\infty)_{\ell,3} q^\ell~.
\end{equation}
From \eqref{eq:Hil_ser_modt4}, we see that the coefficients are given by
\begin{equation}
\label{eq:coefficients_t3}
\begin{aligned}
\mr{dim}(R_\infty)_{0,3} &= |F(0,3)| - |F(0,1)|~, \\
\mr{dim}(R_\infty)_{1,3} &= |F(1,3)| - |F(1,1)| - 4 |F(0,1)|+4|F(0,0)|~,\\
\mr{dim}(R_\infty)_{2,3} &= |F(2,3)|-|F(2,1)|-4|F(1,1)|-5|F(0,1)|+8 |F(1,1)|~,\\
\mr{dim}(R_\infty)_{3,3} & = |F(3,3)|-|F(3,1)|-4 |F(2,1)|-5 |F(1,1)|-5|F(0,1)|+ 11 |F(0,0)|~,\\
\mr{dim}(R_\infty)_{\ell,3} &= |F(\ell,3)|-|F(\ell,1)|-4|F(\ell-1,1)|-5\sum_{i=0}^{\ell -2} |F(i,1)|+ 12 |F(0,0)| \quad 
\ell >3~.
\end{aligned}
\end{equation}
Note that 
\begin{equation}
F(\ell,1) = \begin{cases}
\{X_0,Y_0,Z_0\} & \ell = 0~,\\
 \{X_\ell, Y_\ell,Z_\ell, J_\ell\} & \ell >0~.
\end{cases}    
\end{equation}
Thus, we have
\begin{equation}
\label{eq:Fl1}
|F(\ell,1)| = \begin{cases}
3 & \ell = 0~,\\
4 & \ell >0~.
\end{cases}
\end{equation}
Next, $|F(0,3)| = 10$, since the elements of $F(0,3)$ is constructed using the variables $\{X_0, Y_0, Z_0\}$.
For $\ell>0$, the number $|F(\ell,3)|$ is counted in six cases below:
\begin{itemize}
\item
The number of monomials constructed  using one variable from $\{X_\ell,Y_\ell,Z_\ell,J_\ell\}$ and two variables from $\{X_0,Y_0,Z_0\}$ is $4\times(3\times 2)=24$.
\item
The number of monomials constructed using one variable from $\{X_i, Y_i,Z_i,J_i\}$, one from $\{X_j, Y_j,Z_j,J_j\}$ and one from $\{X_0,Y_0,Z_0\}$, such that $i>j>0$ and $i+j = \ell$, is 
\begin{equation}
    4\times 4 \times 3 \bigg\lfloor\frac{\ell-1}{2}\bigg\rfloor~.
\end{equation}
\item 
The number of monomials constructed using one variable from $\{X_i, Y_i,Z_i,J_i\}$ and two variables from $\{X_j, Y_j,Z_j,J_j\}$, such that $i >0, j>0,i\neq j, i + 2j = \ell$, is  
\begin{equation}
4 \times 10 \left(\bigg\lfloor \frac{\ell-1}{2}\bigg\rfloor+\bigg\lfloor \frac{\ell-1}{3}\bigg\rfloor - \bigg\lfloor \frac{\ell}{3}\bigg\rfloor\right)~.
\end{equation}
\item 
The number of monomials constructed using one variable from $\{X_0,Y_0,Z_0\}$ and two variables from  $\{X_i, Y_i,Z_i,J_i\}$ , such that $i>0, 2i = \ell$, is 
\begin{equation}
    3 \times 10\left(\bigg\lfloor \frac{\ell}{2}\bigg\rfloor - \bigg\lfloor \frac{\ell-1}{2}\bigg\rfloor \right)~.
\end{equation}
\item 
The number of monomials constructed using three variables from  $\{X_i, Y_i,Z_i,J_i\}$ , such that $i>0, 3i = \ell$, is 
\begin{equation}
    20\left(\bigg\lfloor \frac{\ell}{3}\bigg\rfloor - \bigg\lfloor \frac{\ell-1}{3}\bigg\rfloor \right)~.
\end{equation}
\item 
The number of monomials constructed using three variables from three different sets $\{X_i,Y_i,Z_i,J_i\}$, $\{X_j,Y_j,Z_j,J_j\}$ and $\{X_k, Y_k, Z_k, J_k\}$, such that $i>j>k>0, i+j+k = \ell$, is
\begin{equation}
4\times 4 \times 4 \bigg\lfloor \frac{(\ell-3)^2 + 6}{12}\bigg\rfloor~.    
\end{equation}
\end{itemize}
Therefore, for $\ell > 0$, we get 
\begin{equation}
\label{eq:Fl3}
|F(\ell,3)| = 24 + 64\bigg\lfloor \frac{(\ell-3)^2 + 6}{12}\bigg\rfloor +20\bigg\lfloor \frac{\ell-1}{3}\bigg\rfloor + 58\bigg\lfloor \frac{\ell-1}{2}\bigg\rfloor -20 \bigg\lfloor \frac{\ell}{3}\bigg\rfloor+ 30\bigg\lfloor \frac{\ell}{2}\bigg\rfloor~.
\end{equation}
Combining \eqref{eq:coefficients_t3},\eqref{eq:Fl1} and \eqref{eq:Fl3}, then the coefficients are
\begin{equation}
\label{eq:coeff_t3}
\begin{aligned}
\mr{dim}(R_\infty)_{0,3} &= 7~, \\
\mr{dim}(R_\infty)_{1,3} &= 12~,\\
\mr{dim}(R_\infty)_{2,3} &= 27~,\\
\mr{dim}(R_\infty)_{3,3} & = 48~,\\
\mr{dim}(R_\infty)_{\ell,3} &= 41-20\ell + 64\bigg\lfloor \frac{(\ell-3)^2 + 6}{12}\bigg\rfloor +20\bigg\lfloor \frac{\ell-1}{3}\bigg\rfloor + 58\bigg\lfloor \frac{\ell-1}{2}\bigg\rfloor -20 \bigg\lfloor \frac{\ell}{3}\bigg\rfloor+ 30\bigg\lfloor \frac{\ell}{2}\bigg\rfloor~.
\end{aligned}
\end{equation}
We separate the above coefficients into 12 groups by taking $\ell=12 k+i$, where $i=4,5, \ldots 15,$ and $k\in\IZ_{\geq 0}$. All floor functions can then be removed:
\begin{equation}
\begin{aligned}
\mr{dim}(R_\infty)_{12k+4,3} & =768 k^2+416 k+79~ ,\\
\mr{dim}(R_\infty)_{12k+5,3} & =768 k^2+544 k+117~ ,\\
\mr{dim}(R_\infty)_{12k+6,3} & =768 k^2+672 k+171~, \\
\mr{dim}(R_\infty)_{12k+7,3} & =768 k^2+800 k+229~, \\
\mr{dim}(R_\infty)_{12k+8,3} & =768 k^2+928 k+303~, \\
\mr{dim}(R_\infty)_{12k+9,3} & =768 k^2+1056 k+385~, \\
\mr{dim}(R_\infty)_{12k+10,3} & =768 k^2+1184 k+479~, \\
\mr{dim}(R_\infty)_{12k+11,3} & =768 k^2+1312 k+581~, \\
\mr{dim}(R_\infty)_{12k+12,3}& =768 k^2+1440 k+699~, \\
\mr{dim}(R_\infty)_{12k+13,3} & =768 k^2+1568 k+821~, \\
\mr{dim}(R_\infty)_{12k+14,3} & =768 k^2+1696 k+959~, \\
\mr{dim}(R_\infty)_{12k+15,3} & =768 k^2+1824 k+1105~.
\end{aligned}
\end{equation}
Let us denote
\begin{equation}
\mr{dim}(R_\infty)_{12k+i,3}=a_i k^2+b_i k+c_i~.
\end{equation}
The $i$-th sector can be summed over as follows:
\begin{equation}
\sum_{k=0}^{\infty} \mr{dim}(R_\infty)_{12k+i,3} q^{12 k+i}=q^i \frac{a_i q^{12}\left(1+q^{12}\right)+b_i q^{12}\left(1-q^{12}\right)+c_i\left(1-q^{12}\right)^2}{\left(1-q^{12}\right)^3}~.
\end{equation}
Combing \eqref{eq:coeff_t3}, by straightforward computation, we can show that
\begin{equation}\label{eq:I-3_ser}
\begin{aligned}
\sum_{\ell=0}^\infty \mr{dim}(R_\infty)_{\ell,3} q^\ell&= 7+12 q+27 q^2+48 q^3+\sum_{i=4}^{15} q^i \frac{a_i q^{12}\left(1+q^{12}\right)+b_i q^{12}\left(1-q^{12}\right)+c_i\left(1-q^{12}\right)^2}{\left(1-q^{12}\right)^3} \\
& =\left(1-q^4\right)\left(1-q^5\right)\left(1-q^6\right) \sum_{k=0}^6 \frac{1}{(q ; q)_k(q ; q)_{6-k}}
\\
&=\frac{1}{(q;q)_3} \sum_{k=0}^{6}\binom{6}{k}_q~,
\end{aligned}
\end{equation}
where the second last equality follows from explicit calculation. 
This matches with the coefficient of $t^3$ in our formula \eqref{eq:Mac_ind_A1D3}.
\section{Numerical Computation Of $f^n_{m,l}(q)$}\label{app:data_fmnl}
In this section, we present the computation of the polynomials $f^n_{m,l}(q)$ for $n\geq 2$. These are consistent with the general formula \eqref{eq:fnml_closed_form_sec5}. 
In our computation, we noticed a curious recursion relation for 
 $f^n_{m,l}(q)$:
 \begin{equation}\label{eq:rec_fnml}
     f^{n}_{m,l} = \begin{cases}
         0, \quad m \leq 0\; \text{or}\; l \leq 0\;, \text{or}\; m > n l~,\\
         1,\quad m=l> 0~, \\
         q^{(n-1) l^2}, \quad m = n l~,\\
         q^l f^{n}_{m-1,l} + f^{n}_{m-1,l-1}-q^{2m-l-n-1} f^{n}_{m-n-1,l-1}, \quad m/n< l < m~.
     \end{cases}
 \end{equation}
We show in Appendix \ref{app:new_id} that the formula for $f^n_{m,l}$ satisfies this recursion relation for $n=2$ and leave the general proof for future work. We do not yet have a physical interpretation of these recursion relations. 
\begin{table}[H]
    \centering
    \begin{tabular}{|c|l|}
\hline
$(n,m)$ & $f^n_{m,l}(q)$ for $l=1,\dots,m$
\\
\hline
$(2,1)$ & 
$f^2_{1,1} = 1$ \\
\hline
$(2,2)$ & 
$f^2_{2,1} = q$ \\
        & $f^2_{2,2} = 1$ \\
\hline
$(2,3)$ & 
$f^2_{3,1} = 0$ \\
        & $f^2_{3,2} = q + q^2$ \\
        & $f^2_{3,3} = 1$ \\
\hline
$(2,4)$ & 
$f^2_{4,1} = 0$ \\
        & $f^2_{4,2} = q^4$ \\
        & $f^2_{4,3} = q + q^2 + q^3$ \\
        & $f^2_{4,4} = 1$ \\
\hline
$(2,5)$ & 
$f^2_{5,1} = f^2_{5,2} = 0$ \\
        & $f^2_{5,3} = q^4 + q^5 + q^6$ \\
        & $f^2_{5,4} = q + q^2 + q^3 + q^4$ \\
        & $f^2_{5,5} = 1$ \\
\hline
$(2,6)$ & 
$f^2_{6,1} = f^2_{6,2} = 0$ \\
        & $f^2_{6,3} = q^9$ \\
        & $f^2_{6,4} = q^4 + q^5 + 2q^6 + q^7 + q^8$ \\
        & $f^2_{6,5} = q + q^2 + q^3 + q^4 + q^5$ \\
        & $f^2_{6,6} = 1$ \\
\hline
$(2,7)$ & 
$f^2_{7,1} = f^2_{7,2} = f^2_{7,3} = 0$ \\
        & $f^2_{7,4} = q^9 + q^{10} + q^{11} + q^{12}$ \\
        & $f^2_{7,5} = q^4 + q^5 + 2q^6 + 2q^7 + 2q^8 + q^9 + q^{10}$ \\
        & $f^2_{7,6} = q + q^2 + q^3 + q^4 + q^5 + q^6$ \\
        & $f^2_{7,7} = 1$ \\
\hline
$(2,8)$ & 
$f^2_{8,1} = f^2_{8,2} = f^2_{8,3} = 0$ \\
        & $f^2_{8,4} = q^{16}$ \\
        & $f^2_{8,5} = q^9 + q^{10} + 2q^{11} + 2q^{12} + 2q^{13} + q^{14} + q^{15}$ \\
        & $f^2_{8,6} = q^4 + q^5 + 2q^6 + 2q^7 + 3q^8 + 2q^9 + 2q^{10} + q^{11} + q^{12}$ \\
        & $f^2_{8,7} = q + q^2 + q^3 + q^4 + q^5 + q^6 + q^7$ \\
        & $f^2_{8,8} = 1$ \\
\hline
$(2,9)$ & 
$f^2_{9,1} = f^2_{9,2} = f^2_{9,3} = f^2_{9,4} = 0$ \\
        & $f^2_{9,5} = q^{16} + q^{17} + q^{18} + q^{19} + q^{20}$ \\
        & $f^2_{9,6} = q^9 + q^{10} + 2q^{11} + 3q^{12} + 3q^{13} + 3q^{14} + 3q^{15} + 2q^{16} + q^{17} + q^{18}$ \\
        & $f^2_{9,7} = q^4 + q^5 + 2q^6 + 2q^7 + 3q^8 + 3q^9 + 3q^{10} + 2q^{11} + 2q^{12} + q^{13} + q^{14}$ \\
        & $f^2_{9,8} = q + q^2 + q^3 + q^4 + q^5 + q^6 + q^7 + q^8$ \\
        & $f^2_{9,9} = 1$ \\
\hline
\end{tabular}
    \caption{Mathematica computations of the polynomials $f^2_{m,l}$.}
    \label{tab:my_label}
\end{table}
\begin{table}[H]
    \centering
    \begin{tabular}{|c|l|}
    \hline
    $(n,m)$ & $f^n_{m,l}(q)$ for $l=1,\dots,m$
\\
    \hline
$(2,10)$ & 
$f^2_{10,1} = f^2_{10,2} = f^2_{10,3} = f^2_{10,4} = 0$ \\
         & $f^2_{10,5} = q^{25}$ \\
         & $f^2_{10,6} = q^{16} + q^{17} + 2q^{18} + 2q^{19} + 3q^{20} + 2q^{21} + 2q^{22} + q^{23} + q^{24}$ \\
         & $f^2_{10,7} = q^9 + q^{10} + 2q^{11} + 3q^{12} + 4q^{13} + 4q^{14} + 5q^{15} + 4q^{16} + 4q^{17} + 3q^{18} + 2q^{19} + q^{20} + q^{21}$ \\
         & $f^2_{10,8} = q^4 + q^5 + 2q^6 + 2q^7 + 3q^8 + 3q^9 + 4q^{10} + 3q^{11} + 3q^{12} + 2q^{13} + 2q^{14} + q^{15} + q^{16}$ \\
         & $f^2_{10,9} = q + q^2 + q^3 + q^4 + q^5 + q^6 + q^7 + q^8 + q^9$ \\
         & $f^2_{10,10} = 1$ \\
\hline
$(2,11)$ & 
$f^2_{11,1} = f^2_{11,2} = f^2_{11,3} = f^2_{11,4} = f^2_{11,5} = 0$ \\
         & $f^2_{11,6} = q^{25} + q^{26} + q^{27} + q^{28} + q^{29} + q^{30}$ \\
         & \begin{tabular}[t]{@{}l@{}}$f^2_{11,7} = q^{16} + q^{17} + 2q^{18} + 3q^{19} + 4q^{20} + 4q^{21} + 5q^{22}+ 4q^{23} + 4q^{24}+ 3q^{25} + 2q^{26} + q^{27} + q^{28} $\end{tabular} \\
         & \begin{tabular}[t]{@{}l@{}}$f^2_{11,8} = q^9 + q^{10} + 2q^{11} + 3q^{12} + 4q^{13} + 5q^{14} + 6q^{15}+ 6q^{16} + 6q^{17} + 6q^{18} + 5q^{19}+ 4q^{20} + 3q^{21}$ \\
         \quad\quad~$+ 2q^{22} + q^{23} + q^{24} $\end{tabular} \\
         & \begin{tabular}[t]{@{}l@{}}$f^2_{11,9} = q^4 + q^5 + 2q^6 + 2q^7 + 3q^8 + 3q^9 + 4q^{10} + 4q^{11}+ 4q^{12} + 3q^{13} + 3q^{14} + 2q^{15}+ 2q^{16}  $ \\
         \quad\quad~$+ q^{17} + q^{18}$\end{tabular} \\
         & $f^2_{11,10} = q + q^2 + q^3 + q^4 + q^5 + q^6 + q^7 + q^8 + q^9 + q^{10}$ \\
         & $f^2_{11,11} = 1$ \\
\hline
$(2,12)$ & 
$f^2_{12,1} = f^2_{12,2} = f^2_{12,3} = f^2_{12,4} = f^2_{12,5} = 0$ \\
         & $f^2_{12,6} = q^{36}$ \\
         & \begin{tabular}[t]{@{}l@{}}$f^2_{12,7} = q^{25} + q^{26} + 2q^{27} + 2q^{28} + 3q^{29} + 3q^{30} + 3q^{31}+ 2q^{32} + 2q^{33} + q^{34} + q^{35}$ \end{tabular} \\
         & \begin{tabular}[t]{@{}l@{}}$f^2_{12,8} = q^{16} + q^{17} + 2q^{18} + 3q^{19} + 5q^{20} + 5q^{21} + 7q^{22}+ 7q^{23} + 8q^{24} + 7q^{25} + 7q^{26}+ 5q^{27} + 5q^{28} $ \\
         \quad\quad~$+ 3q^{29} + 2q^{30} + q^{31} + q^{32}$\end{tabular} \\
         & \begin{tabular}[t]{@{}l@{}}$f^2_{12,9} = q^9 + q^{10} + 2q^{11} + 3q^{12} + 4q^{13} + 5q^{14} + 7q^{15}+ 7q^{16} + 8q^{17} + 8q^{18} + 8q^{19}+ 7q^{20} + 7q^{21} $ \\
         \quad\quad~$+ 5q^{22} + 4q^{23} + 3q^{24} + 2q^{25} + q^{26} + q^{27}$\end{tabular} \\
         & \begin{tabular}[t]{@{}l@{}}$f^2_{12,10} = q^4 + q^5 + 2q^6 + 2q^7 + 3q^8 + 3q^9 + 4q^{10} + 4q^{11}+ 5q^{12} + 4q^{13} + 4q^{14} + 3q^{15}+ 3q^{16}  $ \\
         \quad\quad ~~$+ 2q^{17} + 2q^{18}+ q^{19} + q^{20}$\end{tabular} \\
         & $f^2_{12,11} = q + q^2 + q^3 + q^4 + q^5 + q^6 + q^7 + q^8 + q^9 + q^{10} + q^{11}$ \\
         & $f^2_{12,12} = 1$ \\
\hline
\end{tabular}
    \caption{Mathematica computations of the polynomials $f^2_{m,l}$ continued.}
    \label{tab:my_label}
\end{table}
\begin{table}[H]
    \centering
    \begin{tabular}{|c|l|}
\hline
$(n,m)$ & $f^n_{m,l}(q)$ for $l=1,\dots,m$ \\
\hline
$(3,1)$ & $f^3_{1,1} = 1$ \\
\hline
$(3,2)$ & $f^3_{2,1} = q$ \\
        & $f^3_{2,2} = 1$ \\
\hline
$(3,3)$ & $f^3_{3,1} = q^2$ \\
        & $f^3_{3,2} = q + q^2$ \\
        & $f^3_{3,3} = 1$ \\
\hline
$(3,4)$ & $f^3_{4,1} = 0$ \\
        & $f^3_{4,2} = q^2 + q^3 + q^4$ \\
        & $f^3_{4,3} = q + q^2 + q^3$ \\
        & $f^3_{4,4} = 1$ \\
\hline
$(3,5)$ & $f^3_{5,1} = 0$ \\
        & $f^3_{5,2} = q^5 + q^6$ \\
        & $f^3_{5,3} = q^2 + q^3 + 2q^4 + q^5 + q^6$ \\
        & $f^3_{5,4} = q + q^2 + q^3 + q^4$ \\
        & $f^3_{5,5} = 1$ \\
\hline
$(3,6)$ & $f^3_{6,1} = 0$ \\
        & $f^3_{6,2} = q^8$ \\
        & $f^3_{6,3} = q^5 + 2q^6 + 2q^7 + q^8 + q^9$ \\
        & $f^3_{6,4} = q^2 + q^3 + 2q^4 + 2q^5 + 2q^6 + q^7 + q^8$ \\
        & $f^3_{6,5} = q + q^2 + q^3 + q^4 + q^5$ \\
        & $f^3_{6,6} = 1$ \\
\hline
$(3,7)$ & $f^3_{7,1} = f^3_{7,2} = 0$ \\
        & $f^3_{7,3} = q^8 + q^9 + 2q^{10} + q^{11} + q^{12}$ \\
        & $f^3_{7,4} = q^5 + 2q^6 + 3q^7 + 3q^8 + 3q^9 + 2q^{10} + q^{11} + q^{12}$ \\
        & $f^3_{7,5} = q^2 + q^3 + 2q^4 + 2q^5 + 3q^6 + 2q^7 + 2q^8 + q^9 + q^{10}$ \\
        & $f^3_{7,6} = q + q^2 + q^3 + q^4 + q^5 + q^6$ \\
        & $f^3_{7,7} = 1$ \\
\hline
$(3,8)$ & $f^3_{8,1} = f^3_{8,2} = 0$ \\
        & $f^3_{8,3} = q^{13} + q^{14} + q^{15}$ \\
        & $f^3_{8,4} = q^8 + q^9 + 3q^{10} + 3q^{11} + 4q^{12} + 3q^{13} + 2q^{14} + q^{15} + q^{16}$ \\
        & $f^3_{8,5} = q^5 + 2q^6 + 3q^7 + 4q^8 + 5q^9 + 4q^{10} + 4q^{11} + 3q^{12} + 2q^{13} + q^{14} + q^{15}$ \\
        & $f^3_{8,6} = q^2 + q^3 + 2q^4 + 2q^5 + 3q^6 + 3q^7 + 3q^8 + 2q^9 + 2q^{10} + q^{11} + q^{12}$ \\
        & $f^3_{8,7} = q + q^2 + q^3 + q^4 + q^5 + q^6 + q^7$ \\
        & $f^3_{8,8} = 1$ \\
\hline
\end{tabular}
    \caption{Mathematica computations of the polynomials $f^3_{m,l}$.}
    \label{tab:my_label}
\end{table}
\begin{table}[H]
    \centering
    \begin{tabular}{|c|l|}
    \hline
$(n,m)$ & $f^n_{m,l}(q)$ for $l=1,\dots,m$ \\
\hline
$(3,9)$ & $f^3_{9,1} = f^3_{9,2} = 0$ \\
&$f^3_{9,3} = q^{18}$
\\
       & $f^3_{9,4} = q^{13}+2q^{14}+3q^{15}+3q^{16}+3q^{17}+2q^{18}+q^{19}+q^{20}$ \\
       & \begin{tabular}[t]{@{}l@{}}$f^3_{9,5} = q^8+q^9+3q^{10}+4q^{11}+6q^{12}+6q^{13}+7q^{14}+5q^{15}+5q^{16}+3q^{17}+2q^{18}+q^{19}+q^{20}$\end{tabular} \\
       & \begin{tabular}[t]{@{}l@{}}$f^3_{9,6} = q^5 + 2q^6 + 3q^7 + 4q^8 + 6q^9 + 6q^{10} + 6q^{11}+ 6q^{12} + 5q^{13} + 4q^{14} + 3q^{15}+ 2q^{16} + q^{17}+q^{18} $\end{tabular} \\
       & \begin{tabular}[t]{@{}l@{}}$f^3_{9,7} = q^2 + q^3 + 2q^4 + 2q^5 + 3q^6 + 3q^7 + 4q^8+ 3q^9 + 3q^{10} + 2q^{11} + 2q^{12} + q^{13} + q^{14}$\end{tabular} \\
       & $f^3_{9,8} = q + q^2 + q^3 + q^4 + q^5 + q^6 + q^7 + q^8$ \\
       & $f^3_{9,9} = 1$ \\
\hline
$(3,10)$ & $f^3_{10,1} = f^3_{10,2} = f^3_{10,3} = 0$ \\
&$f^3_{10,4} =q^{18}+q^{19}+2q^{20}+2q^{21}+2q^{22}+q^{23}+q^{24}$
\\
        & $f^3_{10,5} = q^{13}+2q^{14}+4q^{15}+5q^{16}+7q^{17}+7q^{18}+7q^{19}+6q^{20}+5q^{21}+3q^{22}+2q^{23}+q^{24}+q^{25}$ \\
        & \begin{tabular}[t]{@{}l@{}}$f^3_{10,6} = q^8+q^9+3q^{10}+4q^{11}+7q^{12}+8q^{13}+10q^{14}+10q^{15}+11q^{16}+9q^{17}+8q^{18}+6q^{19}+5q^{20}$
\\\quad\quad~$+3q^{21}+2q^{22}+q^{23}+q^{24}$\end{tabular} \\
        & \begin{tabular}[t]{@{}l@{}}$f^3_{10,7} = q^5+2q^6+3q^7+4q^8+6q^9+7q^{10}+8q^{11}+8q^{12}+8q^{13}+7q^{14}+7q^{15}+5q^{16}+4q^{17}$
\\\quad\quad~$+3q^{18}+2q^{19}+q^{20}+q^{21}$
\end{tabular} \\
        & \begin{tabular}[t]{@{}l@{}}$f^3_{10,8} = q^2 + q^3 + 2q^4 + 2q^5 + 3q^6 + 3q^7 + 4q^8+ 4q^9 + 4q^{10} + 3q^{11} + 3q^{12} + 2q^{13} + 2q^{14} + q^{15} + q^{16}$\end{tabular} \\
        & $f^3_{10,9} = q + q^2 + q^3 + q^4 + q^5 + q^6 + q^7 + q^8 + q^9$ \\
        & $f^3_{10,10} = 1$ \\
\hline
$(3,11)$ & $f^3_{11,1} = f^3_{11,2} =f^3_{11,3} = 0$ \\
&$ f^3_{11,4} = q^{25}+q^{26}+q^{27}+q^{28}$
\\
& $f^3_{11,5}=q^{18}+q^{19}+3q^{20}+4q^{21}+6q^{22}+6q^{23}+7q^{24}+5q^{25}+5q^{26}+3q^{27}+2q^{28}+q^{29}+q^{30} $
\\
        &\begin{tabular}[t]{@{}l@{}} $f^3_{11,6} = q^{13}+2q^{14}+4q^{15}+6q^{16}+9q^{17}+11q^{18}+13q^{19}+14q^{20}+14q^{21}+13q^{22}+11q^{23}+9q^{24}$\\\quad\quad~$+7q^{25}+5q^{26}+3q^{27}+2q^{28}+q^{29}+q^{30}$\end{tabular} \\
        & \begin{tabular}[t]{@{}l@{}}$f^3_{11,7} = q^8+q^9+3q^{10}+4q^{11}+7q^{12}+9q^{13}+12q^{14}+13q^{15}+16q^{16}+15q^{17}+16q^{18}+14q^{19}$\\\quad\quad~$+13q^{20}+10q^{21}+9q^{22}+6q^{23}+5q^{24}+3q^{25}+2q^{26}+q^{27}+q^{28}$\end{tabular} \\
        & \begin{tabular}[t]{@{}l@{}}$f^3_{11,8} = q^5+2q^6+3q^7+4q^8+6q^9+7q^{10}+9q^{11}+10q^{12}+10q^{13}+10q^{14}+10q^{15}+9q^{16}+8q^{17}$\\\quad\quad~$
+7q^{18}+5q^{19}+4q^{20}+3q^{21}+2q^{22}+q^{23}+q^{24}$\end{tabular} \\
        & \begin{tabular}[t]{@{}l@{}}$f^3_{11,9} = q^2+q^3+2q^4+2q^5+3q^6+3q^7+4q^8+4q^9+5q^{10}+4q^{11}+4q^{12}+3q^{13}+3q^{14}+2q^{15}$\\\quad\quad~$+2q^{16}+q^{17}+q^{18}$\end{tabular} \\
        & $f^3_{11,10} = q + q^2 + q^3 + q^4 + q^5 + q^6 + q^7 + q^8 + q^9 + q^{10}$ \\
        & $f^3_{11,11} = 1$ \\
\hline
$(3,12)$ & $f^3_{12,1} = f^3_{12,2} = f^3_{12,3} =  0$
\\
&$f^3_{12,4} =q^{32} $
\\
&$f^3_{12,5} =q^{25}+2q^{26}+3q^{27}+4q^{28}+5q^{29}+4q^{30}+4q^{31}+3q^{32}+2q^{33}+q^{34}+q^{35}$
\\
&\begin{tabular}[t]{@{}l@{}} $f^3_{12,6} = q^{18}+q^{19}+3q^{20}+5q^{21}+8q^{22}+10q^{23}+14q^{24}+14q^{25}+16q^{26}+15q^{27}+14q^{28}+11q^{29}$\\\quad\quad~$+10q^{30}+7q^{31}+5q^{32}+3q^{33}+2q^{34}+q^{35}+q^{36}$\end{tabular}
\\
        &\begin{tabular}[t]{@{}l@{}} $f^3_{12,7} = q^{13}+2q^{14}+4q^{15}+6q^{16}+10q^{17}+13q^{18}+17q^{19}+20q^{20}+23q^{21}+24q^{22}+25q^{23}+23q^{24}$
\\\quad\quad~$+22q^{25}+19q^{26}+16q^{27}+12q^{28}+10q^{29}+7q^{30}+5q^{31}+3q^{32}+2q^{33}+q^{34}+q^{35}$\end{tabular} \\
        & \begin{tabular}[t]{@{}l@{}}$f^3_{12,8} = q^8+q^9+3q^{10}+4q^{11}+7q^{12}+9q^{13}+13q^{14}+15q^{15}+19q^{16}+20q^{17}+22q^{18}+22q^{19}+23q^{20}$\\\quad\quad~$+20q^{21}+19q^{22}+16q^{23}+14q^{24}+11q^{25}+9q^{26}+6q^{27}+5q^{28}+3q^{29}+2q^{30}+q^{31}+q^{32}$\end{tabular} \\
        & \begin{tabular}[t]{@{}l@{}}$f^3_{12,9} = q^5+2q^6+3q^7+4q^8+6q^9+7q^{10}+9q^{11}+11q^{12}+12q^{13}+12q^{14}+13q^{15}+12q^{16}+12q^{17}$\\\quad\quad~$+11q^{18}+10q^{19}+8q^{20}+7q^{21}+5q^{22}+4q^{23}+3q^{24}+2q^{25}+q^{26}+q^{27}$\end{tabular} \\
        & \begin{tabular}[t]{@{}l@{}}$f^3_{12,10} = q^2+q^3+2q^4+2q^5+3q^6+3q^7+4q^8+4q^9+5q^{10}+5q^{11}+5q^{12}+4q^{13}+4q^{14}+3q^{15}$\\\quad\quad~~$+3q^{16}+2q^{17}+2q^{18}+q^{19}+q^{20}$\end{tabular} \\
        & \begin{tabular}[t]{@{}l@{}}$f^3_{12,11} = q+q^2+q^3+q^4+q^5+q^6+q^7+q^8+q^9+q^{10}+q^{11}$\end{tabular} \\
        & $f^3_{12,12} = 1$ \\
\hline
\end{tabular}

    \caption{Mathematica computations of the polynomials $f^3_{m,l}$ continued.}
    \label{tab:my_label}
\end{table}


\section{Explicit Formula For $f^n_{m,l}(q)$}\label{app:new_id}    
In this appendix, we find an explicit form of the polynomials $f_{m,l}^n(q)$.
Recall that  
\begin{equation}\label{eq:main_id_gen_app}
\frac{\left(z ; q^{2n+1}\right)_{\infty}\left(q^{2n+1} / z ; q^{2n+1}\right)_{\infty}\left(q^{2n+1} ; q^{2n+1}\right)_{\infty}}{(z ; q)_{\infty}(q / z ; q)_{\infty}(q ; q)_{\infty}}=\sum_{m=0}^{\infty}q^m\sum_{l=0}^m \frac{f^n_{m,l}(q)}{(q ; q)_l} \sum_{k=-l}^l{
2 l \choose
l+k}_q z^k~,    
\end{equation}
where $f^n_{m,l}$ is a polynomial satisfying 
\begin{equation}\label{eq:fmnl_def_app}
    1 + \sum_{m=1}^\infty q^m \sum_{l=1}^m 
    \frac{f^n_{m,l}(q)}{(q;q)_l} (q^N; q)_l 
    (q^{1-N};q)_l = q^{nN(N-1)}~.
\end{equation}
Our goal is to find a closed form expression for the polynomials $f_{m,l}^n(q)$. Let us begin by proving a lemma.
\begin{lemma}\label{lemma:S_kqN_rel}
For each $k\in\mathds{N}$, define the sum
\begin{equation}
\begin{split}
    S_k:=\sum_{i_k,i_{k-1},\dots,i_3,i_2,i_1\geq 0}&q^{i_k^2+(i_k+i_{k-1})^2+\dots+(i_k+\dots+i_3+i_2)^2}q^{ki_k+(k-1)i_{k-1}+\dots+2i_2+i_1}
    \\
    \times&\frac{(q^N;q)_{i_k+i_{k-1}+\dots +i_2+i_1}(q^{1-N};q)_{i_k+i_{k-1}+\dots +i_2+i_1}}{(q;q)_{i_k}(q;q)_{i_{k-1}}\dots (q;q)_{i_3}(q;q)_{i_2}(q;q)_{i_1}}~.
\end{split}    
\end{equation}
Then we have 
\begin{equation}
    S_k=q^{kN(N-1)}~.
\end{equation}
\end{lemma}
\begin{proof}
We prove this by induction on $k$. For $k=1$, we have 
\begin{equation}
    S_1=\sum_{i_1\geq 0}q^{i_1}\frac{(q^N;q)_{i_1}(q^{1-N};q)_{i_1}}{(q;q)_{i_1}}~.
\end{equation}
This is the hypergeometric sum in \eqref{eq:S_1_sum} and is equal to $q^{N(N-1)}$. We now assume that $S_{k-1}=q^{(k-1)N(N-1)}$ and show that $S_{k}=q^{kN(N-1)}$. We have 
\begin{equation}
\begin{split}
    S_k=\sum_{i_k,i_{k-1},\dots,i_3,i_2\geq 0}&q^{i_k^2+(i_k+i_{k-1})^2+\dots+(i_k+\dots+i_3+i_2)^2}q^{ki_k+(k-1)i_{k-1}+\dots+2i_2}
    \\
    \times&\frac{(q^N;q)_{i_k+i_{k-1}+\dots +i_2}(q^{1-N};q)_{i_k+i_{k-1}+\dots +i_2}}{(q;q)_{i_k}(q;q)_{i_{k-1}}\dots (q;q)_{i_3}(q;q)_{i_2}}
    \\
    \times&\sum_{i_1\geq 0}q^{i_1}\frac{(q^{N+i_k+i_{k-1}+\dots+i_2};q)_{i_1}(q^{1-N+i_k+i_{k-1}+\dots+i_2};q)_{i_1}}{(q;q)_{i_1}}~,
\end{split}    
\end{equation}
where we used 
\begin{equation}
    (a;q)_{n+k}=(a;q)_n(aq^n;q)_k~.
\end{equation}
The inner sum can again be evaluated using the hypergeometric sum \eqref{eq:phi_identity} to be
\begin{equation}
\begin{split}
\sum_{i_1\geq 0}q^{i_1}\frac{(q^{N+i_k+i_{k-1}+\dots+i_2};q)_{i_1}(q^{1-N+i_k+i_{k-1}+\dots+i_2};q)_{i_1}}{(q;q)_{i_1}}&=q^{(N+i_k+i_{k-1}+\dots+i_2)(N-i_k-i_{k-1}-\dots-i_2-1)}
\\
&=q^{N(N-1)-(i_k+\dots+i_2)^2-i_k-i_{k-1}-\dots-i_2}~.
\end{split}
\end{equation}
Plugging in $S_k$, we get 
\begin{equation}
\begin{split}
    S_k=\sum_{i_k,i_{k-1},\dots,i_3,i_2\geq 0}&q^{i_k^2+(i_k+i_{k-1})^2+\dots+(i_k+\dots+i_3+i_2)^2}q^{ki_k+(k-1)i_{k-1}+\dots+2i_2}
    \\
    \times&\frac{(q^N;q)_{i_k+i_{k-1}+\dots +i_2}(q^{1-N};q)_{i_k+i_{k-1}+\dots +i_2}}{(q;q)_{i_k}(q;q)_{i_{k-1}}\dots (q;q)_{i_3}(q;q)_{i_2}}
    \\
    \times&q^{N(N-1)-(i_k+\dots+i_2)^2-i_k-i_{k-1}-\dots-i_2}
    \\
    =q^{N(N-1)}&\sum_{i_k,i_{k-1},\dots,i_3,i_2\geq 0}q^{i_k^2+(i_k+i_{k-1})^2+\dots+(i_k+\dots+i_3)^2}q^{(k-1)i_k+(k-2)i_{k-1}+\dots+i_2}
    \\
    &\hspace{3cm}\times\frac{(q^N;q)_{i_k+i_{k-1}+\dots +i_2}(q^{1-N};q)_{i_k+i_{k-1}+\dots +i_2}}{(q;q)_{i_k}(q;q)_{i_{k-1}}\dots (q;q)_{i_3}(q;q)_{i_2}}~.
\end{split}    
\end{equation}
Renaming the dummy variables $i_j\to i_{j-1},~2\leq j\leq k$ in the last sum, we see that
\begin{equation}
    S_k=q^{N(N-1)}S_{k-1}=q^{(k-1)N(N-1)}q^{N(N-1)}=q^{kN(N-1)}~.
\end{equation}
\end{proof}
Now in $S_k$ we change variable to  
\begin{equation}
    i_1+i_2+\dots+i_k=:i_1'~.
\end{equation}
Then the sum becomes 
\begin{equation}
\begin{split}
    S_k=\sum_{i_k,i_{k-1},\dots,i_3,i_2,i_1'\geq 0}&q^{Q(i_2,i_3,\dots,i_k)}q^{(k-1)i_k+(k-2)i_{k-1}+\dots+i_2+i_1'}
    \\
    \times&\frac{(q^N;q)_{i'_1}(q^{1-N};q)_{i'_1}}{(q;q)_{i_k}(q;q)_{i_{k-1}}\dots (q;q)_{i_3}(q;q)_{i_2}(q;q)_{i'_1-i_2-\dots-i_k}}~,
\end{split}    
\end{equation}
where we have defined 
\begin{equation}\label{eq:Q_def}
    Q(i_2,i_3,\dots,i_k):=i_k^2+(i_k+i_{k-1})^2+\dots+(i_k+\dots+i_3+i_2)^2~.
\end{equation}
Next we change variables to 
\begin{equation}
    (k-1)i_k+(k-2)i_{k-1}+\dots+i_2+i_1'=:i_2'~.
\end{equation}
Noting that 
\begin{equation}
    \frac{1}{(q;q)_k}=\frac{(q^{k+1};q)_\infty}{(q;q)_\infty}=0,\quad \text{for}\quad k<0~,
\end{equation}
we get 
\begin{equation}
\begin{split}
    S_k&=\sum_{i_k,i_{k-1},\dots,i_3,i_2',i_1'\geq 0}\frac{q^{Q(i_2'-(k-1)i_k-(k-2)i_{k-1}-\dots-2i_3-i_1',i_3,\dots,i_k)}q^{i_2'}}{(q;q)_{i_k}(q;q)_{i_{k-1}}\dots (q;q)_{i_3}}
    \\
    &\hspace{4cm}\times\frac{(q^N;q)_{i'_1}(q^{1-N};q)_{i'_1}}{(q;q)_{i_2'-(k-1)i_k-(k-2)i_{k-1}-\dots-2i_3-i_1'}(q;q)_{2i'_1+i_3+2i_4+\dots+(k-2)i_k-i_2'}}
    \\
    &=\sum_{i_2'\geq 0}q^{i_2'}\sum_{i_1'\geq 0}\frac{(q^N;q)_{i'_1}(q^{1-N};q)_{i'_1}}{(q;q)_{i_1'}}
    \\
    &\times\sum_{i_k,i_{k-1},\dots,i_3\geq 0}\frac{q^{Q(i_2'-(k-1)i_k-(k-2)i_{k-1}-\dots-2i_3-i_1',i_3,\dots,i_k)}(q;q)_{i_1'}}{(q;q)_{i_k}\dots (q;q)_{i_3}(q;q)_{i_2'-(k-1)i_k-(k-2)i_{k-1}-\dots-2i_3-i_1'}(q;q)_{2i'_1+i_3+2i_4+\dots+(k-2)i_k-i_2'}} ~.
\end{split}      
\end{equation}
Note that the second sum is bounded above by $i_2'$ since 
\begin{equation}
    \frac{1}{(q;q)_{i_2'-(k-1)i_k-(k-2)i_{k-1}-\dots-2i_3-i_1'}}=0\quad\text{unless}\quad i_1'\leq i_2'~.
\end{equation}
Thus we have 
\begin{equation}
    S_k=\sum_{m=0}^\infty q^{m}\sum_{l= 0}^{m}f^k_{m,l}(q)\frac{(q^N;q)_{l}(q^{1-N};q)_{l}}{(q;q)_{l}}~,
\end{equation}
where 
\begin{equation}\label{eq:fnml_closed_form}
f^k_{m,l}(q)=\sum_{i_k,i_{k-1},\dots,i_3\geq 0}\frac{q^{Q(m-(k-1)i_k-(k-2)i_{k-1}-\dots-2i_3-l,i_3,\dots,i_k)}(q;q)_{l}}{(q;q)_{i_k}\dots (q;q)_{i_3}(q;q)_{m-(k-1)i_k-(k-2)i_{k-1}-\dots-2i_3-l}(q;q)_{2l+i_3+2i_4+\dots+(k-2)i_k-m}}~.    
\end{equation}
For $k=1,2$, we define $f^{k}_{m,l}$ as in \eqref{eq:f1ml_def} and \eqref{eq:f2ml_def} respectively.\par
By Lemma \ref{lemma:S_kqN_rel}, we see that $f^n_{m,l}$ in \eqref{eq:fnml_closed_form} satisfies \eqref{eq:fmnl_def_app}. 
\subsection{Proof Of Recursion Relation \eqref{eq:rec_fnml} For $k=2$}
In this appendix, we prove that the formula \eqref{eq:fnml_closed_form} satisfies the recursion relation \eqref{eq:rec_fnml} for $k=2$. First note that by definition $f^k_{m,l}=0$ for $m\leq 0$ or $l\leq 0$. For $m=l$, we see that only the term $i_3=i_4=\dots i_k=0$ is nonvanishing since 
\begin{equation}
    \frac{1}{(q;q)_r}=0 \quad\text{for}\quad r<0~.
\end{equation}
Clearly, the term $i_3=i_4=\dots i_k=0$ in the sum \eqref{eq:fnml_closed_form} is 1. For $m=kl$, the only nonvanishing term in the sum is $i_3=i_4=\dots i_{k-1}=0,i_k=l$. This term in the sum is given by 
\begin{equation}
    \frac{q^{(k-1)l^2}(q;q)_l}{(q;q)_l}=q^{(k-1)l^2}~.
\end{equation}
We now prove the final recursion relation
\begin{equation}
f^2_{m,l}=q^l f^{2}_{m-1,l} + f^{2}_{m-1,l-1}-q^{2m-l-3} f^{2}_{m-3,l-1}~.    
\end{equation}
For $k=2$, the sum is empty and we simply have 
\begin{equation}
f^2_{m,l}(q)=q^{(m-l)^2}\frac{(q;q)_{l}}{(q;q)_{m-l}(q;q)_{2l-m}}=q^{(m-l)^2}{l\choose m-l}_q~.    
\end{equation}
We have 
\begin{equation}
\begin{split}
q^l f^{2}_{m-1,l} &+ f^{2}_{m-1,l-1}-q^{2m-l-3} f^{2}_{m-3,l-1}\\&=\frac{q^{(m-l-1)^2+l}(q;q)_{l}}{(q;q)_{m-l-1}(q;q)_{2l-m+1}}+ \frac{q^{(m-l)^2}(q;q)_{l-1}}{(q;q)_{m-l}(q;q)_{2l-m-1}}-\frac{q^{(m-l-2)^2+2m-l-3}(q;q)_{l-1}}{(q;q)_{m-l-2}(q;q)_{2l-m+1}}
\\&=\frac{q^{(m-l)^2}(q;q)_{l}}{(q;q)_{m-l}(q;q)_{2l-m}}\left[q^{l-2(m-l)+1}\frac{1-q^{m-l}}{1-q^{2l-m+1}}+\frac{1-q^{m-2l}}{1-q^l}\right.\\&\left.\hspace{6cm}-q^{-2m+3l+1}\frac{(1-q^{m-l-1})(1-q^{m-l})}{(1-q^l)(1-q^{2l-m+1})}\right]~.
\end{split}
\end{equation}
Straightforward calculation gives
\begin{equation}
q^{l-2(m-l)+1}\frac{1-q^{m-l}}{1-q^{2l-m+1}}+\frac{1-q^{m-2l}}{1-q^l}-q^{-2m+3l+1}\frac{(1-q^{m-l-1})(1-q^{m-l})}{(1-q^l)(1-q^{2l-m+1})}= 1~.   
\end{equation}
\end{appendix}

\bibliography{mtc}
\end{document}